\documentclass[10pt,a4paper]{article}
\usepackage{mathtools} 
\usepackage[utf8]{inputenc}
\usepackage[T1]{fontenc}
\usepackage{microtype}
\usepackage{fourier}
\usepackage{enumitem}

\usepackage{tikz}
\usepackage[margin=10pt,font=small,labelfont=bf,labelsep=period]{caption}
\usepackage{amsmath}
\usepackage{amsthm}
\usepackage{amssymb}
\usepackage{mathscinet}
\usepackage{pdfpages} 
\usepackage{appendix}
\numberwithin{equation}{section}

\usepackage[nottoc,notlot,notlof]{tocbibind}

\usepackage{aliascnt}
\usepackage[colorlinks,linkcolor=blue,citecolor=blue]{hyperref}

\theoremstyle{plain}

\newtheorem{theorem}{Theorem}[section]

\newaliascnt{lemma}{theorem}
\newtheorem{lemma}[lemma]{Lemma}
\aliascntresetthe{lemma}

\newaliascnt{corollary}{theorem}
\newtheorem{corollary}[corollary]{Corollary}
\aliascntresetthe{corollary}

\theoremstyle{definition}

\newaliascnt{definition}{theorem}
\newtheorem{definition}[definition]{Definition}
\aliascntresetthe{definition}

\newaliascnt{example}{theorem}
\newtheorem{example}[example]{Example}
\aliascntresetthe{example}

\newaliascnt{remark}{theorem}
\newtheorem{remark}[remark]{Remark}
\aliascntresetthe{remark}

\newaliascnt{proposition}{theorem}
\newtheorem{proposition}[proposition]{Proposition}
\aliascntresetthe{proposition}

\newaliascnt{conjecture}{theorem}
\newtheorem{conjecture}[conjecture]{Conjecture}
\aliascntresetthe{conjecture}

\newcommand{\R}{\mathbf{R}}
\newcommand{\C}{\mathbf{C}}

\newcommand{\m}{\mathbf{m}}
\newcommand{\ms}{\mathbf{m}^*}
\newcommand{\ue}{\mathbf{u}}
\newcommand{\us}{\mathbf{u}^*}

\renewcommand{\epsilon}{\varepsilon}

\newcommand{\abs}[1]{\left\lvert #1 \right\rvert}
\newcommand{\avg}[1]{\bigl\langle #1 \bigr\rangle}

\renewcommand{\l}{\lambda}

\renewcommand{\l}{\lambda}

\newcommand{\bm}[1]{\mbox{\boldmath{$#1$}}}

\DeclareMathOperator{\sgn}{sgn}

\begin{document}

\title{On the multipeakon system of a two-component Novikov equation } 
%\author{Xiang-Ke Chang}
%%\thanks{$^*$Corresponding author (changxk@lsec.cc.ac.cn).}
%\address{ LSEC, ICMSEC, Academy of Mathematics and Systems Science, Chinese Academy of Sciences, P.O. Box 2719, Beijing 100190, PR China; and School of Mathematical Sciences, University of Chinese Academy of Sciences, Beijing 100049, PR China.}
%\email{changxk@lsec.cc.ac.cn}
%
%\author{Jacek Szmigielski}
%\address{ Department of Mathematics \& Statistics and Centre for Quantum Topology and Its Applications (quanTA), University of Saskatchewan, Saskatoon, SK, CANADA S7N 5E6.}
%\email{szmigiel@math.usask.ca}

\author{Xiang-Ke Chang
\thanks{ LSEC, ICMSEC, Academy of Mathematics and Systems Science, Chinese Academy of Sciences, P.O.Box 2719, Beijing 100190, PR China; and School of Mathematical Sciences, University of Chinese Academy of Sciences, Beijing 100049, PR China; changxk@lsec.cc.ac.cn}
   \and
  Jacek Szmigielski\thanks{Department of Mathematics \& Statistics and Centre for Quantum Topology and Its Applications (quanTA), University of Saskatchewan, Saskatoon, SK, CANADA S7N 5E6; szmigiel@math.usask.ca}}
\date{}
%\date{\today}
%\date{ 14 April  2021}

\makeatletter
\hypersetup{%
   pdfauthor={Jacek Szmigielski},
  pdftitle={\@title},
   }
\makeatother

\maketitle
\begin{abstract}

We are exploring variations of the Novikov equation that have weak solutions called peakons. Our focus is on a two-component Novikov equation with a non-self-adjoint $4\times 4$ Lax operator for which we examine the related forward and inverse spectral maps  for the peakon sector. To tackle the forward spectral problem, we convert it into a matrix eigenvalue problem,  for which the original boundary value problem is a 
two-fold cover. We then use an isospectral deformation to the long-time regime to calculate the relevant eigenvalues.
To support the long-term deformation, we prove the global existence of the peakon flows using ideas based on Moser's deformation method, which was used in his study of the finite Toda lattice.
We subsequently solve the inverse problem by studying a trio of Weyl functions that can be approximated simultaneously by rational functions that involve tensor products and an additional symmetry condition. This part of our paper is rooted in Krein's solution to the inverse problem for the Stieltjes string. 
\end{abstract} 
\tableofcontents{}

\section{Introduction} 
\subsection{On Camassa-Holm type equations and peakons}
The Korteweg-de Vries equation models well the shallow water waves in one spatial dimension.  It does not, however, account for known nonlinear phenomena like the extreme wave of Stokes or the breakdown of regularity.  
In 1993 Camassa and Holm \cite{camassa-holm}, building on an earlier work of Green and Naghdi, derived a nonlinear PDE with 
a quadratic dispersion which turned out to have not  only the integrable structure similar to the Korteweg-de Vries equation but 
was also shown to exhibit a breakdown of regularity \cite{McKean-breakdown, constantin-escher}.
  The Camassa-Holm (CH) equation reads: 
\begin{equation} \label{eq:CH}
m_t+(um)_x+u_x m=0, \quad m=u-u_{xx},  
\end{equation} 
where $f_x=\frac{\partial f}{\partial x}$, $f_t=\frac{\partial f}{\partial t}$ etc. 

This CH equation has a host of features that make it special. In addition to the already mentioned integrability and the  breakdown of regularity, we  
would like to highlight the following features: 
\begin{enumerate} 
\item the spectral problem for the CH Lax pair is unitary equivalent to the spectral problem for an inhomogeneous string (\cite{beals-sattinger-szmigielski:1998:acoustic-scattering-KdV-hierarchy}); 
\item $(um)_x+u_xm$ term represents a co-adjoint action 
of the Lie algebra of diffeomorphisms of the circle $S^1$, and the relation $m=u-u_{xx}$ is the analog of the inertia tensor appearing in the 
Euler equation of the rigid body, induced from the $H^1$ Sobolev norm on the space of vector fields $u\partial _x$ (\cite{misiolek:1998:CH-as-geodesic-flow-on-Bott-Virasoro-group}). This feature makes \eqref{eq:CH} belong to the same class of equations as the Euler equation of the rigid body and 
the Euler equation of fluids, to name just a few.  Moreover, 
the viscous generalization of the CH equation can be viewed as a closure approximation for the Reynolds-averaged equations of the incompressible Navier-Stokes equations \cite{holm:1998:CH-turbulent-pipes, holm:1998:CH-turbulence}; 
\item equation \eqref{eq:CH} possesses  non-smooth soliton solutions, dubbed \emph{ peakons} (peaked solitons), given by the ansatz \cite{camassa-holm}: 
\begin{equation}
  \label{eq:CH peakons}
  u(x,t) = \sum_{k=1}^N m_k(t) \, e^{-\abs{x - x_k(t)}}
  ,
\end{equation}
where $m_k(t), x_k(t)$ are smooth functions of $t$.  

\end{enumerate} 

In subsequent years  after the formulation of \eqref{eq:CH} peakon solutions have attracted much attention and many new peakon-bearing equations were discovered.  
In particular, V. Novikov 
considered in \cite{novikov:generalizations-of-CH} a family of nonlinear equations of CH-type: 
\begin{equation} \label{eq: CH-Novikov} 
(1-\partial_x^2) u_t=F(u, u_x, u_{xx}, \dots), \qquad u=u(t,x), \qquad \partial_x=\frac{\partial}{\partial x}, 
\end{equation} 
to classify those scalar equations of CH-type which possess infinite hierarchies of higher symmetries.  
The family \eqref{eq: CH-Novikov}  includes the CH equation \eqref{eq:CH} 
and the Degasperis-Procesi (DP) equation 
\begin{equation} \label{eq:DP} 
m_t+(um)_x+2u_xm=0, \qquad m=u-u_{xx}.  
\end{equation} 
Both these equations have quadratic nonlinearities.  One of the new equations on the list, called
the Novikov (NV1) equation , reads  
\begin{equation} \label{eq:NV1}
m_t+((um)_x+2u_xm)u=0, \qquad m=u-u_{xx}, 
\end{equation} 
and has, in contrast, a cubic nonlinearity.

From the perspective of partial differential equations, peakons are special weak solutions of integrable CH-type equations such as the CH, DP, and NV1.  Yet, they capture the main attributes of solutions in a particularly appealing way, for example, 
\begin{enumerate} 
\item the breakdown of regularity manifests itself as collisions of peakons \cite{camassa-holm, camassa-holm-hyman:1994:CH-new-integrable,McKean-breakdown, beals-sattinger-szmigielski:moment,lundmark:shockpeakons,szmigielski-zhou:shocks-DP, szmigielski-zhou:DP-peakon-antipeakon}; 
\item the long time asymptotics of the field $u$ is vastly simplified by observing that  peakons become  free particles in the asymptotic region \cite{beals-sattinger-szmigielski:moment,McKean-Fred,eckhardt2013isospectral}. 
\end{enumerate} 

Furthermore, the complicated dynamics of peakons can be significantly simplified by employing the inverse spectral methods.The study of the forward and inverse spectral problems associated with peakons often involves the theories of totally non-negative or even oscillatory kernels, (bi)orthogonal polynomials and approximations, and random matrices \cite{beals-sattinger-szmigielski:moment,beals2001peakons,bertola-gekhtman-szmigielski:cauchy,bertola-gekhtman-szmigielski:twomatrix,bertola-gekhtman-szmigielski:meijerG,chang2017lax,chang2019isospectral,eckhardt2017camassa,hone-lundmark-szmigielski:novikov,lundmark-szmigielski:DPlong,lundmark-szmigielski:GX-inverse-problem}. More recent studies revealed that the peakon problem for the NV1 equation is associated with a class of the so-called partial-skew-orthogonal polynomials and Hermite--Pad\'{e} approximation problems with Pfaffian structures, as well as the Bures random matrix ensemble \cite{chang:2022:nv-pfaffians,chang2018partial,chang2018application}. 
 For a comprehensive introduction to peakons, readers might want to consult a recent review paper \cite{lundmark2022view} and references therein.

%It is perhaps fitting to explain why studying multi-component peakon systems is not only of potential interest but also  may play a significant role in understanding what peakon equations are.
  Even though peakon equations have been around for almost 30 years and have been widely studied, it is still a mystery what peakon equations are. In fact, there is no classification to clarify the common features of these equations.  
The situation with peakons is in stark contrast to soliton equations, at least in dimension $1+1$, for which a beautiful Lie algebra classification 
exists due to Drinfeld and Sokolov associating every soliton equation to an affine Kac-Moody Lie algebra data.  To the best of our knowledge, 
there is nothing equally aesthetically satisfying for peakon equations.  Our research has been informed by one very important fact: all peakon equations known to us are isospectral deformations of boundary value problems with  \emph{simple, point spectra}.  
Thus in the case of the CH equation, with no dispersion term, one is dealing with isospectral deformations of a self-adjoint boundary value problem with a positive, simple, 
point spectrum.  It is an old result of M.G. Krein that a self-adjoint operator on a Hilbert space (finite or infinite dimensional) with a positive, simple, spectrum can be realized as 
an inhomogeneous string boundary value problem.  From that perspective, the CH equation is one of the isospectral deformations of the inhomogeneous string problem; peakons are nothing but discrete strings, with densities corresponding to finite sums of Dirac measures (\cite{beals-sattinger-szmigielski:1998:acoustic-scattering-KdV-hierarchy}). 
The story does not end here. Almost without exception, all other scalar peakon equations known to us have simple point spectra, but, very importantly, as opposed to the CH equation, they are not coming from self-adjoint problems. Yet, miraculously, their boundary problems after some transformations lead to spectral problems with totally non-negative kernels, or even oscillatory kernels, introduced by Gantmacher and Krein in their quest to understand the properties of oscillatory mechanical systems. In general, totally non-negative kernels do not have to be self-adjoint, and the cases of the DP \cite{lundmark-szmigielski:DPlong} or  NV1 \cite{hone-lundmark-szmigielski:novikov} equations are good examples of 
equations leading to non-self-adjoint oscillatory systems. For example, the peakon problem for the NV1 equation was solved with the help of the dual cubic string and the theory of oscillatory kernels. 

Very recently, there has been some interest in the two-component peakon-bearing equations. Intriguing phenomena have been observed in the behavior of the two-component modified CH equation and the Geng-Xue equation   \cite{chang2016multipeakons,lundmark-szmigielski:GX-inverse-problem}.
%\todo{get chang2016multipeakons} 
 The main objective of the present work is to clarify the situation for yet another two-component system, namely, the two-component Novikov equation---a much more challenging peakon system.

\subsection{On the two-component generalizations of NV1} \label{sec:two-component generalizations} 

The NV1 equation \eqref{eq:NV1} (\cite{novikov:generalizations-of-CH}) is an intriguing modification of the Camassa-Holm equation, with a cubic nonlinearity rather than quadratic nonlinearity. The study of its peakon sector exhibits a plethora of new, collective phenomenon \cite{chang:2022:nv-pfaffians,chang2018partial,chang2018application,hone-lundmark-szmigielski:novikov,kardell:2015:CH-novikov-peakon-creation, kardell:2016:phdthesis} .  

The presence of its infinitely many symmetries (\cite{novikov:generalizations-of-CH}) suggests the Lax integrability, and one proposal for a Lax pair of the NV1 equation was 
put forward in \cite{novikov:generalizations-of-CH}, while another Lax pair was 
given by Hone and Wang \cite{hone-wang:cubic-nonlinearity}
\begin{subequations}
  \label{eq:Novikov-lax}
  \begin{equation}
    \label{eq:Novikov-lax-x}
    \frac{\partial}{\partial x}
    \begin{pmatrix} \psi_1 \\ \psi_2 \\ \psi_3 \end{pmatrix} =
    \begin{pmatrix}
      0 & zm & 1 \\
      0 & 0 & zm \\
      1 & 0 & 0
    \end{pmatrix}
    \begin{pmatrix} \psi_1 \\ \psi_2 \\ \psi_3 \end{pmatrix}
    ,
  \end{equation}
  \begin{equation}
    \label{eq:Novikov-lax-t}
    \frac{\partial}{\partial t}
    \begin{pmatrix} \psi_1 \\ \psi_2 \\ \psi_3 \end{pmatrix} =
    \begin{pmatrix}
      -u u_x & \frac{u_x}{z}-u^2 mz & u_x^2 \\
      \frac{u}{z} & - \frac{1}{z^2} & - \frac{u_x}{z} - u^2 mz \\
      -u^2 & \frac{u}{z} & uu_x
    \end{pmatrix}
    \begin{pmatrix} \psi_1 \\ \psi_2 \\ \psi_3 \end{pmatrix}
    ,
  \end{equation}
\end{subequations}
where $z$ is the spectral parameter.   It is a routine, although lengthy, exercise to verify that 
for a smooth $u$ these two matrix equations are compatible if and only if \eqref{eq:NV1} holds.  
We mention here that \eqref{eq:NV1} can be rewritten, assuming smooth $u$, as 
\begin{equation} \label{eq:NV-second}
m_t+(u^2m)_x+u_xu m=0, \qquad m=u-u_{xx}.  
\end{equation} 

Extending Lax pairs to 
non-smooth $u$ requires special care and the reader is referred to \cite{hone-lundmark-szmigielski:novikov} for a discussion of this point.  In short, 
for a non-smooth $u$, the distributional Lax pairs are necessary, especially if one wants to study the impact of the Lax integrability  on  the peakon sector of solutions \eqref{eq:CH peakons}.  
The peakon ansatz satisfies \eqref{eq:NV1} (see \cite{hone-wang:cubic-nonlinearity}) 
if and only if the system of ODEs holds:
\begin{equation} \label{eq:NVpeakonODEs}
\dot x_k=u^2(x_k), \qquad 
\dot m_k=-m_k \langle u_x \rangle (x_k) u(x_k),  
\end{equation} 
where $\langle f \rangle (x_j)$ is the arithmetic mean  of the left- and right-hand limits of $f$ at $x_k$.  
From a more systematic point of view, the same system of ODEs can be obtained as the compatibility result of a properly defined distributional Lax pair \cite{hone-lundmark-szmigielski:novikov}.  
We point out that for peakons, the term $u_xu\, m$ in \eqref{eq:NV-second} is originally ill-defined, and the distributional Lax integrability amounts to 
defining $ u_x u \delta_{x_k}=\avg{uu_x}(x_k)\delta_{x_k}$ at each point $x_k$ in the support of $m$.

There were multiple attempts to generalize \eqref{eq:NV1} to two-component systems. We discuss below three of those attempts.

First, in 2009, Geng and Xue ~\cite{geng-xue:cubic-nonlinearity} proposed the system 
\begin{equation}
  \label{eq:GX-first}
  \begin{gathered}
    m_t + \bigl( (u m)_x + 2 u_x m \bigr) \, v = 0
    , \\
    n_t + \bigl( (v n)_x + 2 v_x n \bigr) \, u = 0
    , \\
    m = u - u_{xx}
    ,\quad
    n = v - v_{xx}
    ,
  \end{gathered}
 \end{equation}
which reduces to \eqref{eq:NV1} when $u=v$ and thus $m=n$.  
This system can be obtained, at least for smooth $u$ and $v$,  \emph{via} a compatibility condition of the Lax pair
\begin{subequations}
  \label{eq:GX-laxI}
  \begin{equation}
    \label{eq:GX-laxI-x}
    \frac{\partial}{\partial x}
    \begin{pmatrix} \psi_1 \\ \psi_2 \\ \psi_3 \end{pmatrix} =
    \begin{pmatrix}
      0 & zn & 1 \\
      0 & 0 & zm \\
      1 & 0 & 0
    \end{pmatrix}
    \begin{pmatrix} \psi_1 \\ \psi_2 \\ \psi_3 \end{pmatrix}
    ,
  \end{equation}
  \begin{equation}
    \label{eq:GX-laxI-t}
    \frac{\partial}{\partial t}
    \begin{pmatrix} \psi_1 \\ \psi_2 \\ \psi_3 \end{pmatrix} =
    \begin{pmatrix}
      -v_xu & \frac{v_x}{z}-vunz & v_xu_x \\
      \frac{u}{z} & v_xu - vu_x - \frac{1}{z^2} & - \frac{u_x}{z} - vumz \\
      -vu & \frac{v}{z} & vu_x
    \end{pmatrix}
    \begin{pmatrix} \psi_1 \\ \psi_2 \\ \psi_3 \end{pmatrix}
    ,
  \end{equation}
\end{subequations}
which naturally reduces to the Lax pair~\eqref{eq:Novikov-lax} for Novikov's equation.  We note that that the Lax pair above is still of $3\times3$;  the spectral and inverse spectral problems corresponding to the multipeakon sector  were studied in \cite{lundmark-szmigielski:GX-inverse-problem}.

We now turn to two other two-component generalizations of the NV1 equation \eqref{eq:NV-second}, which will be referred as \emph{hyperbolic, elliptic} cases, respectively. The justification for the names will be explained in a separate  paper.

\noindent 1. {\large  \bf the elliptic case} 

This candidate for a two-component Novikov equation was introduced in  \cite{li-li-chen:multi-component-novikov, Li-Hongmin:2019:twoNV}.  In its original form the system reads
\begin{subequations} \label{eq:2nv}
\begin{align}
&m_t+(u^2+v^2)m_x+3(uu_x+vv_x)m+(u_xv-uv_x)n=0,\\
&n_t+(u^2+v^2)n_x+3(uu_x+vv_x)n+(v_xu-vu_x)m=0,\\
&m=u-u_{xx},\quad n=v-v_{xx}. 
\end{align}
\end{subequations}
For consistency of presentation we will rewrite this system, again assuming $u,v$ to be smooth, as 
\begin{subequations} \label{eq:2nv-second}
\begin{align}
&m_t+((u^2+v^2)m)_x+(uu_x+vv_x)m+(u_xv-uv_x)n=0,\\
&n_t+((u^2+v^2)n)_x+(uu_x+vv_x)n+(v_xu-vu_x)m=0,\\
&m=u-u_{xx},\quad n=v-v_{xx}.
\end{align}
\end{subequations}
Clearly, the elliptic two-component Novikov equation  \eqref{eq:2nv}, to be referred in this section as NV2e, reduces  to the Novikov equation \eqref{eq:NV-second} when  either  $u=0$ or $v=0$.  In addition, the reduction $v=u$ reproduces up to a scale the Novikov equation as well.

In the peakon sector 
\begin{equation}
  \label{eq:GX-peakons}
  \begin{split}
    u(x,t) &= \sum_{k=1}^N m_k(t) \, e^{-\abs{x - x_k(t)}}
    , \\
    v(x,t) &= \sum_{k=1}^N n_k(t) \, e^{-\abs{x - x_k(t)}},   
  \end{split}
\end{equation}
we employ the same regularization scheme as was used in the CH case, obtaining 
\begin{subequations} \label{2nv_ode}
\begin{align}
& \dot x_k=(u^2(x_k)+v^2(x_k)),\\
&\dot m_k=\left(u(x_k)\avg{v_x}(x_k)-v(x_k)\avg{u_x}(x_k)\right)n_k- \left(u(x_k)\avg{u_x(x_k)}+ v(x_k)\avg{v_x(x_k)}\right)m_k,\\
&\dot n_k=\left(v(x_k)\avg{u_x}(x_k)-u(x_k)\avg{v_x}(x_k)\right)m_k- \left(u(x_k)\avg{u_x(x_k)}+ v(x_k)\avg{v_x(x_k)}\right)n_k.
\end{align}
\end{subequations}
Qu and Fu showed in \cite{qu-fu:2020:cauchy-2nv-peakons} that the  distributional Lax  compatibility yields exactly this system of equations.  

 \noindent 2. {\large \bf the hyperbolic case}
 
 \noindent %We consider the bilinear form given by \eqref{eq:hbf} and define  the system
 An alternative two-component system reads
\begin{subequations} \label{eq:NV2}
\begin{align} 
&m_t+(uvm)_x+u_xv m=0, \label{eq:NV2m}\\
&n_t+(uvn)_x+uv_x n=0, \label{eq:NV2n}\\
&m=u-u_{xx}, \qquad n=v-v_{xx},  
\end{align} 
\end{subequations} 
with a $4\times4$ Lax pair given by \eqref{eq:NVstar} below.   

Up to a scale, this system was another candidate for a two-component Novikov equation proposed by Li \cite{Li-Hongmin:2019:twoNV}.  
We will refer to this system as a hyperbolic two-component Novikov equation, or NV2h for short.   For some earlier work on this system 
the reader  might wish to consult  \cite{he-qu:2020:global-weak-2nv, min-qu:2021:2nv-peakons-h1-conservation}.  

We will argue in the present paper that the hyperbolic system given by equation \eqref{eq:NV2} is a natural, and non-trivial, generalization of 
the original Novikov equation (NV1).  In support of that claim we will frequently refer to \cite{hone-lundmark-szmigielski:novikov} to highlight similarities, which are far less transparent in the case of the NV2e.    For example, recall that 
NV1 can be written as
\begin{equation*}
m_t+(u^2 m)_x+u_x um =0, \qquad m=u-u_{xx}, 
\end{equation*} 
which, in particular, shows that NV2h reduces to NV1 whenever $u=v$ in NV2h.  This is, of course, also true for other generalizations of NV1, but this formal reduction is especially straightforward in the  case of \eqref{eq:NV2}.  
A quick look at the peakon sector, defined in the same way as for the other equations, 
using the same regularization of the terms $u_x v m$ and $v_xu n$ as before, yields
\begin{equation} \label{eq:epeakons} 
\begin{gathered} 
\dot x_j=u(x_j) v(x_j), \\
\dot  m_j =-m_j \langle u_x \rangle (x_j) v(x_j), \qquad \qquad \dot n_j =-n_j \langle v_x \rangle  (x_j) u(x_j).  
\end{gathered} 
\end{equation} 
For comparison's sake we include the peakon equations for NV1 \eqref{eq:NVpeakonODEs} 
$$ 
\dot x_j=u^2(x_j), \qquad 
\dot m_j =-m_j \langle u_x \rangle (x_j) u(x_j).  
$$ 

In the present work we will perform a detailed analysis of the forward and inverse spectral problems associated with the peakon system \eqref{eq:epeakons} subject to the following restrictions on the initial values:
\begin{enumerate}
\item[(1)] $x_1(0)<x_2(0)<\cdots <x_N(0),$
\item[(2)] $m_j(0)>0$ and $n_j(0)>0,$
\end{enumerate}
which corresponds to the \emph{pure peakon case}.

We note that, even though we keep emphasizing commonalities shared by NV2h and NV1, there is also a great deal of difference between the analysis for NV2h and NV1 \cite{chang:2022:nv-pfaffians,hone-lundmark-szmigielski:novikov}. In particular, we note that NV2h has a $4\times4$ Lax pair as opposed to $3\times 3$. Moreover, as discused in \cite{hone-lundmark-szmigielski:novikov} the  total positivity discovered for the  NV1 peakons does not directly apply here. In summary, the present work presents several challenging new avenues for research going beyond previous work on peakon problems \cite{beals-sattinger-szmigielski:moment,chang:2022:nv-pfaffians,chang2016multipeakons,chang2017lax,hone-lundmark-szmigielski:novikov,lundmark-szmigielski:DPlong,lundmark-szmigielski:GX-inverse-problem}.

%This is perhaps a good place to explain why we have not reduced the bilinear form \eqref{eq:hbf} to its diagonal form. In this paper, we are studying the \emph{pure peakon case}, namely, $m_j(0)>0, n_j(0)>0$.   Upon bringing the bilinear form to the diagonal form one of the fields would have coefficients proportional to  $m_j-n_j$, hence no longer positive. That lack of positivity would significantly complicate our arguments. 

\subsection{Outline and highlights of this paper}

\noindent In Sec. \ref{sec:basic}, we discus some basic properties of NV2h \eqref{eq:NV2}, including the constants of motion and the weak formulation. 

\vspace{0.3cm}

\noindent In Sec. \ref{sec:Global}, we analyze the forward spectral problem associated with the multipeakon system \eqref{eq:epeakons} using its  $4\times4$ Lax pair. That spectral problem is non-self-adjoint. We associate it with a specific matrix eigenvalue problem (see \eqref{eq:MEVP}). The original boundary value problem spectrum is a $2$-fold cover of the non-zero spectrum of the matrix eigenvalue problem.   We show that given ordered positions $x_j$ and positive masses $m_j,n_j$, the non-zero eigenvalues of the matrix eigenvalue problem are positive and simple. The main idea behind our approach, which goes back to Moser's work \cite{moser:1975:Toda} on the finite Toda lattice (see a subsequent work \cite{moser:three-integrable} as well), is to use isospectrality to drive the matrix system to the asymptotic long-time region where 
the dynamics simplifies.   To justify the isospectral deformation, we prove the global existence of the isospectral flow induced from the peakon system \eqref{eq:epeakons} and establish a link between the long-time asymptotics of positions and momenta and the non-zero eigenvalues of the 
matrix problem.  We note that Moser's integrability argument has been recently used to establish integrability of the perturbed Toda chain by Deift, Li, Spohn, Tomei and Trogdon \cite{deift-li-spohn-tomei-trogdon:todawforcing}.

\vspace{0.3cm} 

\noindent 
Sec. \ref{sec:weyl} is a continuation of the study of the forward spectral problem, with the final goal of determining the spectral measure/spectral data. To this end, we introduce a triple of Weyl functions and demonstrate that they are matrix-valued Stieltjes transforms of discrete positive measures.   The proof, again, relies on the long-time asymptotics of positions and momenta.  
Finally, we show that the Weyl functions satisfy a tensor product version of a mixed Hermite--Pad\'{e} approximation problem with an extra symmetry condition.

\vspace{0.3cm}

\noindent 
In Sec. \ref{sec:inverse}, we treat the inverse problem, i.e., the problem of recovering the peakon positions and positive masses from the spectral data. To obtain a definite solution to the inverse problem, we solve the mixed Hermite--Pad\'{e} approximation problem for the Weyl functions under some mild technical assumptions.

\vspace{0.3cm}

\noindent 
Sec. \ref{sec:peakons} is devoted to concrete examples of global multipeakons written in terms of the spectral data for NV2h. We show that the actual spectral data constitutes a proper subset of what we call \emph{extended spectral data} known from previously studied peakon systems. Conjecturally, the spectral data is cut out from the extended spectral data by systems of determinant inequalities. Explicit characterizations of the nature of that embedding are given for one and two peakon cases. As for the general case, we provide a sufficient condition characterizing the embedding. We conjecture that a complete description of the spectral data set can be given in terms of a specific generalization of  \emph{total positivity} known from the Gantmacher-Krein's theory of oscillatory kernels.

\vspace{0.3cm}

\noindent Finally, we include three appendices, one of which, Appendix \ref{app3},  
addresses the question of the reduction from NV2h to NV1, while the remaining two deal with technical issues raised in the main body of the paper.  

\section{Basic constants of motion and the weak formulation of NV2} \label{sec:basic}
Since the remainder of the paper is entirely devoted to the hyperbolic two-component Novikov system we will simply refer to NV2h as NV2.  
In this section we briefly discus some basic properties of \eqref{eq:NV2}.  
We begin with the conservation law for the Sobolev $H^1$-norm (for more details see e.g. \cite{min-qu:2021:2nv-peakons-h1-conservation, he-liu-qu:2021:orbital-stab-NV2}). 

\begin{proposition}
The $H^1$ norms $||u||_{H^1}$ and $||v||_{H^1}$  are constants of motion.  
\end{proposition} 
Since the proof for the general case is laborious, we give below a short proof for peakons satisfying \eqref{eq:epeakons}.  
Recall that \textit{peakon solutions} are defined 
$$m(t)=2\sum_{j=1}^N m_j(t) \delta_{x_j(t)}, \qquad  n(t) =2\sum_{j=1}^N n_j(t) \delta_{x_j(t) }, $$
 and,  
 \begin{equation} \label{eq:uvpeakons} 
u(t,x)=\sum_{j=1}^N m_j(t) e^{-\abs{x-x_j(t)}}, \qquad v(t,x)= \sum_{j=1}^N n_j(t) e^{-\abs{x-x_j}},   
\end{equation} 
where $x_j(t), m_j(t)$ are smooth functions of the time variable $t$.  We will frequently suppress in the notation the dependence on $t$ or $x$, unless the context requires an explicit reference to these variables.  
Let us define $h_1= ||u||_{H^1}^2, \, h_2 = ||v||_{H^1}^2$, where the $H^1$-norm is taken with respect to the $x$ variable.   It is easy to check that for the system of peakons given by \eqref{eq:uvpeakons} 
$$ 
h_1=\int_\R u m dx=2 \sum_{i,j} m_i m_j e^{-\abs{x_i-x_j}}, \qquad h_2=\int_\R v n dx=2 \sum_{i,j} n_i n_j e^{-\abs{x_i-x_j}}.  
$$  

\begin{proposition} \label{prop:h1h2p}
Suppose the positions $x_j$ and masses $m_j, n_j$ satisfy equations \eqref{eq:epeakons}.    Then the norms $h_1$ and $h_2$ 
are conserved.  
\end{proposition} 
\begin{proof} 
\begin{equation*} 
\begin{gathered} 
\frac{d}{dt} h_1=\frac{d}{dt} \sum_{j=1}^n u(x_j) m_j=\sum_{j=1}^n\big( \frac{d}{dt} u(x_j)  m_j + u(x_j)\dot m_j\big)=\\
2\sum_{j=1}^n\big( \avg{u_x}(x_j)\dot x_j  m_j -u(x_j)m_j\avg{u_x}(x_j) v(x_j)\big)=\\
2\sum_{j=1}^n\big( \avg{u_x}(x_j) u(x_j) v(x_j)  m_j - u(x_j)m_j\avg{u_x}(x_j) v(x_j)\big)=0. 
\end{gathered} 
\end{equation*} 
The proof for $v$ follows the same steps.  
\end{proof} 
We comment now on the weak formulation of NV2 as a system of equations in $u,v$ which is consistent with our 
Lax-inspired distributional formulation.  We will eventually focus our discussion on the peakon sector, but first 
we reformulate \eqref{eq:NV2} for smooth data.  
%\cxk{For notational convenience we will use a short-hand $\partial_x=\frac{\partial}{\partial x} $} \cxk{We have had this abbreviation in the begging of the introduction. Shall we delete this sentence?}.  
We have  
\begin{theorem} 
For smooth $u, v$, \eqref{eq:NV2} is equivalent to the system: 
\begin{subequations} \label{eq:NV2uv}
\begin{align} 
&m_t+(1-\partial_x^2)(uv u_x)+\partial_x(\frac{u_x^2}{2} v+u^2 v+uu_x v_x)+\frac{u_x^2 v_x}{2}=0, \label{eq:NV2muv}\\
&n_t+(1-\partial_x^2)(uv v_x)+\partial_x(\frac{v_x^2}{2} u+v^2 u+vv_x u_x)+\frac{u_x v_x^2}{2}=0\label{eq:NV2nuv}, \\
&m=u-u_{xx}, \qquad n=v-v_{xx}.  \notag
\end{align} 
\end{subequations} 
\end{theorem} 
\begin{proof} 
We give the proof for $m$; switching $m$ with $n$ and $u$ with $v$ renders the proof for $n$.  
This is a computational proof and we only indicate the relevant steps:
\begin{enumerate} 
\item Claim 1: $(uvm)_x=(u^2v)_x+(1-\partial_x^2)(uvu_x)-uvu_x+((uv)_x u_x)_x;$
\item Claim 2: $u_xvm=u u_xv -(\frac{u_x^2}{2} v)_x+\frac{u_x^2 v_x}{2};$
\item Adding Claims 1 and 2 and substituting into \eqref{eq:NV2m} gives the desired result.  
\end{enumerate} 
\end{proof} 
To define \emph{weak solutions} we adopt a strategy similar to the one used for the CH equation  (see e.g. \cite{holden-raynaud:2006:CH-convergent-scheme-based-on-multipeakons}).  
\begin{definition}\label{def:weaksol}  
The pair $u, v$ of functions in $L^1_{\text{loc}} ([0,T), W^{1,3}_{\text{loc}} (\R))$ is a weak solution of 
the two-component Novikov system if equations \eqref{eq:NV2uv} hold in the sense of distributions.  
\end{definition}

We note that the statement about the choice of the Sobolev space $W^{1,3}_{\text{loc}} (\R)$ follows from the routine generalized H\"{o}lder inequality: 
$$ 
\int_\Omega |f_1..f_k| dx\leq \prod_{j=1} ^k ||f_j||_{L^{p_j}(\Omega)}, 
$$ 
valid for all $f_j\in L^{p_j}(\Omega)$ such that $1/p_1+1/p_2+\dots 1/p_k=1$.  
This shows that triple products occurring in the proposed equation are in $L^1_{\text{loc}}(\R)$ and thus, the proposed nonlinear PDE is a distributional 
equation.  
Finally, the following lemma for peakons can be directly verified.  
\begin{lemma} 
$u(t,x)$ and $v(t,x)$ given by \eqref{eq:uvpeakons} are weak solutions in the sense of \autoref{def:weaksol} if and only if equations \eqref{eq:epeakons} hold.  
\end{lemma} 
 \begin{remark} 
There is no preferred choice of the space of functions in the existing literature on Novikov systems.  
In the paper \cite{min-qu:2021:2nv-peakons-h1-conservation} of Zhao and Qu there is a definition (see Def 2.1 there) of \emph{weak solutions} applied to the case of two peakons with emphasis on $H^1(\R)$ conservation laws.  In the paper 
by He and Qu \cite{he-qu:2020:global-weak-2nv}, $W^{1,\infty}(\R) $ and $H^1(\R) $ weak solutions are mentioned (see Thm. 3.1 there). In yet another paper by He and Qu \cite{he-qu:2021:global-weak-2nv}, $W^{1,4}(\R) $ and $H^1(\R) $ weak solutions are explored in Lemma 4.1 and Thm. 5.1 therein.
However, the paper \cite{qu-fu:2020:cauchy-2nv-peakons} by Qu and Fu uses as a definition $W^{1,3}(\R) $ weak solutions to NV2e (see \eqref{eq:2nv}).  We note that peakon solutions we are interested in belong to $W^{1,p} $ for any $1\leq p\leq \infty$. 
\end{remark}

\section{The spacial Lax operator, the global existence of flows, and the spectral problem} \label{sec:Global}

In this section, we formulate the appropriate boundary value problem based on the Lax pair of NV2 (see \eqref{eq:NV2})  that turns the peakon flows \eqref{eq:epeakons} into isospectral flows. Subsequently, we establish the global existence of those flows and use the asymptotic behavior at an infinite time to determine a sharp characterization of the spectrum of the boundary value problem.
%\js{ Ok, I included your paragraph above.  } 

\noindent {\bf Notation}: we will denote by ${1}_k$ the identity matrix in $M_{k,k}(\C)$ and by $0_{k\times l}$ the zero matrix in $M_{k,l}(\C)$.  In either case, we will omit the index when $k=1$, $k=1=l$, respectively.  Also, in the special case of $k=2$, we will use $\mathbf{1}$ to denote $1_2$.  
Since we are using the field of complex numbers throughout the paper 
we will also omit to mention the field in the notation, i.e., we will write $M_{k,k} $ rather than $M_{k,k} (\C)$.  

For convenience, let us define a two-dimentional vector space $W_2$ and denote by $W_2^*$ the dual of $W_2$.  From this point onward, the vectors from the dual 
will carry the star, i.e. $\alpha^* \in W_2^*$, and $\alpha^*(v)$ is the value of $\alpha^*$ on $v\in W_2$.  Moreover, we assume that $W_2$ is equipped with a symmetric bilinear form  $\ue^*(\mathbf{w})=\langle \ue, \mathbf{w}\rangle=u_1w_2+u_2w_1$ for $\mathbf{u}=\begin{bmatrix} u_1\\ u_2 \end{bmatrix}$ and  $\mathbf{w}=\begin{bmatrix} w_1\\ w_2 \end{bmatrix}$ with $\mathbf{u}^*=[u_2, \, u_1]$ and $\mathbf{w}^*=[w_2, \, w_1]$. 

 %To ease the notation burden we will suppress the dependence on $x,t$, and exclusively use \emph{subscripts} for derivatives. 
   
%The problem depends on two, vector-valued, functions $\m(x,t)\in W_2$ and $\ns(x,t)\in W_2^*$ and their ``potentials '' $\ue(x,t)\in W_2$ and $\vs(x,t) \in W_2^*$, satisfying 
%\begin{equation*} 
%\ue(x,t)- \ue_{xx} (x,t)=\m(x,t), \qquad \vs (x,t)-\vs_{xx} (x,t)=\ns(x,t).  
%\end{equation*}

\subsection{The spacial Lax operator} \label{subsec:xLax} 
%For the remainder of this paper, we fix the symmetric bilinear form $\langle\ve, \we \rangle =\langle \begin{bmatrix}v_1\\v_2\end{bmatrix},  \begin{bmatrix} w_1\\ w_2 \end{bmatrix}\rangle=v_1w_2+v_2w_1$.  
We consider the discrete vector valued measure $\m\in W_2$  (corresponding to  equation \eqref{eq:uvpeakons} ), 
and its dual $\m^*$, explicitly given by
\begin{equation} \label{eq:MMsigma} 
\begin{split} 
\m=2\sum_{j=1}^N \m_j \delta_{x_j}, \qquad &\m^*=2\sum_{j=1}^N \m^* _j \delta_{x_j}, \qquad   \\
\m _j=\begin{bmatrix} m_j\\ n_j \end{bmatrix}, \qquad   &\m^*_j=[n_j, \, m_j],   \quad  \, 
\end{split} 
\end{equation}
with potentials $\ue, \us$ satisfying 
\begin{equation*} 
\ue=\sum_{j=1}^N \m_j e^{-\abs{x-x_j}}, \qquad \us=\sum_{j=1}^N \ms_j e^{-\abs{x-x_j}}  
\end{equation*} 
We will assume that, at least initially, the $x_j$s are ordered: $x_1<x_2<\dots<x_N$.

%We start with the direct sum decomposition \eqref{eq:Wdirect}, where  now $\dim W_2=2$,  and 

Consider the following Lax system: 
\begin{subequations} \label{eq:NVstar}
\begin{align}
 \Phi_x&=\begin{bmatrix} 0&z\ms&1\\0_{2\times 1} &0_{2\times 2} &z\m\\1&0_{2\times 1} &0 \end{bmatrix} \Phi, \label{eq:xLax}\\
 \Phi_t&=\frac12\begin{bmatrix} -\langle \ue_x, \ue \rangle&\frac{\us_x}{z} -z\langle \ue,\ue\rangle  \ms&\langle \ue_x, \ue_x \rangle\\\frac{\ue}{z}&-\frac{\mathbf{1}}{z^2} +\ue\otimes \us_x -\ue_x\otimes \us &-\frac{\ue_x}{z} -z\langle\ue, \ue\rangle \m\\-\langle \ue, \ue \rangle &\frac{\us}{z} &\langle \ue_x, \ue\rangle \end{bmatrix} \Phi, \label{eq:tLax} 
 \end{align}
 \end{subequations} 
splitting $\Phi$ as $\Phi=[\phi_1, \mathbf{\phi_2}, \phi_3]^T$.  We note that the factor 
$\frac12$ in the $t$-flow is needed for consistency  with \eqref{eq:NV2}.

Writing \eqref{eq:xLax} in components leads to the system of ODEs 
\begin{subequations}\label{eq:ODEs} 
\begin{align} 
\partial_x\phi_1&=z\ms \mathbf{\phi_2} + \phi_3=z \langle \m, \mathbf{\phi_2}\rangle +\phi_3, \label{eq:psi123}\\
\partial_x \mathbf{\phi_2}&=z \m \phi_3, \label{eq:psi2psi3} \\
\partial_x\phi_3&=\phi_1.   \label{eq:psi3psi1}
\end{align} 
\end{subequations}
We furthermore impose the boundary conditions:  
\begin{equation} \label{eq:BCs}
\phi_3(-\infty)=0, \quad  \mathbf{\phi_2} (-\infty)=0_{2\times 1}, \text{   and } \phi_3(+\infty)=0. 
\end{equation}
The multiplication in the first equation \eqref{eq:psi123} needs to be defined, and we follow the procedure discussed in \cite{hone-lundmark-szmigielski:novikov} and also earlier in the introduction: since 
$\mathbf{\phi_2}$ is a function of bounded variation, it has both one-sided limits at every point, and we assign the arithmetic average, denoted by $\avg{}$, 
as the value of the function at the point of the singular support of $\ms$.  Thus  we define
\begin{equation} \label{eq:Mmultiplication} 
\ms \mathbf{\phi_2}:=\ms \avg{\mathbf{\phi_2}}.   
\end{equation} 
We specify the associated initial value problem by fixing the initial conditions 
for $x< x_1$ to be 
\begin{equation}\label{eq:IVP} 
\phi_3(x)=e^x=\phi_1(x), \qquad \mathbf{\phi_2}(x)=0_{2\times 1}.  
\end{equation} 
The uniqueness of the solution to the initial value problem guarantees several elementary symmetry properties of 
$\phi_1, \mathbf{\phi_2}, \phi_3$.  

\begin{proposition} \label{prop:IVP} 
The solutions to the IVP \eqref{eq:xLax}, \eqref{eq:IVP}, satisfy the following symmetry conditions: 
\begin{enumerate}[label=\alph*)]
\item $\phi_1(x,-z)=\phi_1(x,z), \quad \mathbf{\phi_2}(x,-z)=-\mathbf{\phi_2}(x,z), \quad \phi_3(x,-z)=\phi_3(x,z). $
\item $\overline{\phi_1(x,z)}=\phi_1(x,\bar z), \quad \overline{\mathbf{\phi_2}(x,z)}=\mathbf{\phi_2}(x,\bar z),\quad  \overline{\phi_3(x,z)}=\phi_3(x,\bar z)$.   
\end{enumerate} 

\end{proposition} 
Let us set $$\boxed{\l=-z^2,}$$ 
and consider $x\in (x_k, x_{k+1})$.  We write  
\begin{equation} \label{eq:ABC}
\begin{bmatrix} \phi_1(x,z) \\ \mathbf{\phi_2}(x,z)\\ \phi_3(x,z) \end{bmatrix}=\begin{bmatrix} A_k(\l) e^x-\l C_k(\l)e^{-x} \\ 2z\mathbf{B}_k(\l)\\
A_k(\l) e^x+\l C_k(\l)e^{-x}\end{bmatrix}.   
\end{equation} 
\begin{proposition}\label{prop:transition}
The system of differential equations \eqref{eq:ODEs} with the initial conditions \eqref{eq:IVP} is equivalent to the system of difference equations on $A_k, \mathbf{B}_k,C_k$:
\begin{align*} 
\begin{bmatrix} A_k(\l) \\ \mathbf{B}_k(\l)\\ C_k(\l)  \end{bmatrix}=&\underbrace{\begin{bmatrix} 1-\l \langle\m_k,  \m_k\rangle&-2\l \ms_ke^{-x_k} &-\l^2 \langle\m_k,  \m_k\rangle e^{-2x_k} \\
\m_k e^{x_k}&\mathbf{1}&\l  \m_k e^{-x_k} \\
 \langle\m_k,  \m_k\rangle e^{2x_k}& 2 \ms e^{x_k} & 1+\l \langle\m_k,  \m_k\rangle
\end{bmatrix}}_{S_k(\l)} \begin{bmatrix} A_{k-1}(\l) \\ \mathbf{B}_{k-1}(\l)\\ C_{k-1}(l) \end{bmatrix},
\end{align*}
where $1\leq k\leq N, \,  A_0=1, \,  \mathbf{B_0}=0_{2\times 1}, \,  C_0=0$, and  
~$\text{deg}\,  A_k(\l)=k,\,  \text{deg}\,  \mathbf{B}_k(\l)=k-~1,\,  \newline \text{deg}\, C_k(\l)=k-1$.   
\end{proposition}

We call those $z$ for which a non-trivial solution to the boundary problem exists the \emph{spectrum} of the boundary value problem.  
We list some elementary properties of the spectrum.  
\begin{proposition} \label{prop:spectrum1}
The eigenvalues of the boundary value problem \eqref{eq:xLax}, \eqref{eq:BCs},  satisfy the following conditions: 
\begin{enumerate}
\item $z=0$ is not an eigenvalue;  
\item If $z$ is an eigenvalue, so is $-z$; 
\item If $z$ is an eigenvalue, so is $\bar z$; 
\item for each eigenvalue $z$, the corresponding eigenspace is one dimensional.  
\end{enumerate} 
\end{proposition} 
\begin{remark} 
By mapping the original boundary value problem \eqref{eq:xLax}, \eqref{eq:BCs}, to an associated matrix eigenvalue problem \eqref{eq:MEVP} we will eventually prove that the spectrum is purely imaginary, in agreement with the above proposition.  
\end{remark} 
We will conclude the preliminary study of the initial/boundary value problems by observing that in the asymptotic region to the 
right of the support of $\m$, that is for $x>x_N$, we can write (setting for simplicity $A_N(\l)=A(\l),\, \mathbf{B}_N(\l)=\mathbf{B}(\l), ~C_N(\l)=C(\l)$): 
\begin{equation} \label{eq:AC}
\phi_3(x,z)=A(\l)e^{x}+\l  C(\l) e^{-x}, 
\end{equation} 
which yields a simple characterization of the boundary value problem \eqref{eq:xLax}, \eqref{eq:BCs}.  
\begin{proposition} \label{prop:Spectrum1} 
$z$ is an eigenvalue of \eqref{eq:xLax}, \eqref{eq:BCs} if and only if 
\begin{equation} 
\boxed{ A(\l)=0,   \qquad \l=-z^2.  }
\end{equation} 
\end{proposition}

Obtaining more information about the spectrum directly from the system \eqref{eq:xLax}, \eqref{eq:BCs}, or equivalently 
from the discrete system obtained in \autoref{prop:transition}, appears to be complicated.  We, instead, 
will formulate yet another eigenvalue problem implied by the boundary value problem \eqref{eq:xLax}, \eqref{eq:BCs}.  
This step mirrors a similar approach taken in \cite{hone-lundmark-szmigielski:novikov}.  The advantage of proceeding via this route is that 
we can see more clearly the nature of the transition from the scalar NV equation to a system.  
We start by noting that we can eliminate $\psi_1$ from the system \eqref{eq:xLax} by using \eqref{eq:psi3psi1}.  
Thus, suppressing the spectral variable $z$ whenever its presence is self-evident, we get
\begin{equation*} 
(-D_x^2 +1) \phi_3=-z\ms \avg{\mathbf{\phi_2}}, 
\end{equation*} 
and, since on the spectrum $\lim_{\abs{x}\to \infty}\phi_3(x)=0$, we can solve for $\phi_3$ obtaining 
\begin{equation} \label{eq:psi3}
\phi_3(x)=-z\int_{\R} G_D(x,y) \ms (y) \avg{\mathbf{\phi_2}}(y) dy,  
\end{equation} 
where $G_D$ is the Green's function for $-D_x^2+1$  vanishing at $\pm \infty$ and given by 
$$ 
G_D(x,y)=\frac12 e^{-\abs{x-y}}.  
$$

Furthermore, since $\lim_{x\to -\infty} \mathbf{\phi_2}=0_{21}$,   we can solve \eqref{eq:psi2psi3} for $\mathbf{\phi_2}$: 
\begin{equation} \label{eq:psi2} 
\mathbf{\phi_2}(x)=z \int_{-\infty} ^x \m (y) \phi_3(y) dy. 
\end{equation} 
For the case of the finite discrete measure given by \eqref{eq:MMsigma} we obtain a matrix eigenvalue problem for the components of $\avg{\mathbf{\phi_2}} $.  
Indeed, let us define 
\begin{align} 
\Psi&=\begin{bmatrix} \avg{\mathbf{\phi_2}}(x_1)\\\avg{\mathbf{\phi_2}} (x_2)\\\vdots\\ \avg{\mathbf{\phi_2}}(x_N)  \end{bmatrix} \in M_{2N,1}, \\
\notag\\
P&=\begin{bmatrix} \m_1&0_{2\times1}&0_{2\times 1}&\dots&0_{2\times1} \\
0_{2\times1}& \m_2&0_{2\times 1} &\dots&0_{2\times 1} \\
\vdots&\hdots&\hdots&\hdots&0_{2\times 1} \\
0_{2\times 1} &0_{2\times 1} &\hdots&\hdots&\m_N\end{bmatrix} \in M_{2N, N} ,\\
\notag\\
E&=\begin{bmatrix} 1&e^{-\abs{x_1-x_2}}&\hdots&e^{-\abs{x_1-x_N}}\\
e^{-\abs{x_2-x_1}}&1&\hdots&e^{-\abs{x_2-x_N}}\\
\vdots&\vdots&\vdots&\vdots\\
e^{-\abs{x_N-x_1}}&e^{-\abs{x_N-x_2}}&\hdots&1 \end{bmatrix} \in M_{N,N}, \\
\notag 
\end{align} 
\begin{align} 
P^*&=\begin{bmatrix} \ms_1&0_{1\times 2}&0_{1\times 2}&\dots&0_{1\times 2} \\
0_{1\times 2} &\ms_2&0_{1\times 2} &\dots&0_{1\times 2} \\
\vdots&\hdots&\hdots&\hdots&0_{1\times 2} \\
0_{1\times 2} &0_{1\times 2} &\hdots&\hdots&\ms_N \end{bmatrix} \in M_{N, 2N}, \\
\notag\\
T&=\begin{bmatrix} 1&0&\hdots&0&0\\
2&1&0&\hdots&0\\
\vdots&\vdots&\vdots&1&0\\
2&2&2&\hdots&1 \end{bmatrix} \in M_{N,N}, 
\end{align} 
and set again
$$ 
\l=-z^2. 
$$
\begin{proposition} \label{prop:matrix spectral problem} 
The column matrix $\Psi\in M_{2N,1}$ solves the eigenvalue problem 
\begin{equation} \label{eq:MEVP}
\Psi=\lambda \big[(T\otimes \mathbf{1} )PE P^*\big] \, \Psi . 
\end{equation} 

\end{proposition} 
\begin{remark} 
The matrix $\big[(T\otimes \mathbf{1})P E P^*\big] \in M_{2N, 2N}$ generalizes the matrix 
$TPEP\in M_{N,N}$, with $P=\text{diag} (m_1,m_2, \hdots, m_N)\in M_{N,N}$, known from the treatment of the NV peakon problem in 
\cite{hone-lundmark-szmigielski:novikov}.  
\end{remark} 
\begin{proof} 
We compute $\avg{\mathbf{\phi_2}}(x_k)$ using \eqref{eq:psi2} and \eqref{eq:MMsigma} to get
\begin{equation*} 
\avg{\mathbf{\phi_2}}(x_k)=z\big[\sum_{j=1}^{k-1} 2 \m_j\phi_3(x_j)+ \m_k \phi_3(x_k)\big]. 
\end{equation*} 

Likewise, \eqref{eq:psi3} gives
$$ 
\phi_3(x_k)=-z \sum_{l=1}^N G_D(x_k, x_l) 2\ms_l \avg{\mathbf{\phi_2}}(x_l), 
$$
hence, by combining these two equations, we obtain the claim.  
\end{proof} 
The spectrum of the original boundary value problem \eqref{eq:xLax}, \eqref{eq:BCs} is characterized in 
\autoref{prop:Spectrum1}.  Subsequently, we map that boundary value problem to the matrix eigenvalue problem \eqref{eq:MEVP}
whose non-zero spectrum, or to be more precise, the reciprocals of the non-zero eigenvalues are given by the zeros in $\l$ of the characteristic 
polynomial $\det (I_{2N}-\l \big[(T\otimes \mathbf{1} )P E P^*\big])$ with the caveat that the original spectral variable 
$z$ is now mapped to $\lambda=-z^2$.  The question that presents itself is a relation between these two spectral problems.  
The following proposition clarifies the nature of that relation.  
\begin{proposition} \label{prop:characteristic polynomial} 

\mbox{}
\begin{enumerate} 
\item The characteristic polynomial $\det (I_{2N}-\l \big[(T\otimes \mathbf{1} )PE P^*\big])$ has degree 
$N$.  
\item $A(\l)$ defined in \eqref{eq:AC} satisfies
$$\boxed{ A(\l)=\det (I_{2N}-\l \big[(T\otimes \mathbf{1} )P E P^*\big]).} $$ 
\end{enumerate} 
\end{proposition} 
\begin{proof} 
The first item follows from the Cauchy-Binet formula applied to minors of $\big[(T\otimes \mathbf{1} )PE P^*\big]$; indeed, 
by the Cauchy-Binet formula any minor of $\big[(T\otimes \mathbf{1}  )PE P^*\big]$ of degree $k>N$ is zero because $E\in M_{N,N}$.   

The second claim easily follows from the fact that both $A(\l)$ and \newline $\det (I_{2N}-\l \big[(T\otimes \mathbf{1}  )PE P^*\big])$ have zeros 
at the points of the spectrum of the original boundary value problem \eqref{eq:xLax}, \eqref{eq:BCs}, have the same degree in $\l$ and,  finally, both have value $1$ at $\l=0$.  
\end{proof} 
Until now, we were ignoring the time evolution.  This can be repaired easily by 
evaluating \eqref{eq:tLax} 
in the asymptotic region $x_N< x$.  Using standard methods we obtain the evolution of $A(\l), \mathbf{B}(\l), C(\l)$.  
\begin{proposition} 
Set $\mathbf{M}_+=\sum_{k=1}^N \m_k e^{x_k}$.  The coefficients $A(\l), \mathbf{B}(\l), C(\l)$ 
undergo the time evolution: 
\begin{equation} 
\dot A(\l) =0, \qquad \dot {\mathbf{B}}(\l)=\frac{\mathbf{B}(\l)-A(\l) \mathbf{M}_+}{2\l}, \qquad 
\dot C(\l)=\frac{\langle \mathbf{M}_+, \mathbf{B}(\l)-\mathbf{M}_+\rangle}{\l}. 
\end{equation} 
\end{proposition} 
Thus the coefficients of the characteristic polynomial $A(\l)$ are constants of motion.  
Let us write 
\begin{equation} 
A(\l) =1+\sum_{k=1}^N (-\l)^k H_k.  
\end{equation} 
In subsequent arguments, we will use $H_1$ and $H_N$.  To compute them we use the canonical 
basis $\{e_i: 1\leq i\leq N\}$  of $\R^N$ and the basis $\{e_i\otimes e_{2a}: 1\leq i\leq N, \, 1\leq a\leq 2\}$ of $\R^{2N}$, where $\{e_{2a}, a=1,2\}$ is the canonical basis of 
$\R^2$.  The matrix entries of $X\in M_{2N,2N}$ will be listed as $X_{ia,jb}$.  Likewise, if $X\in M_{2N,N}$ then the matrix entries 
of $X$ will be written as $X_{ia, j}$, etc.  We note that 
\begin{equation} \label{eq:matrix entries}
(T\otimes\mathbf{1})_{ia,jb}=T_{ij} \delta_{ab}, \qquad P_{ia,j}=\delta_{ij} \m_{ia},  \qquad 
P^*_{i,ja}=\delta_{ij} \ms_{ja}.  
\end{equation} 

\begin{proposition} \label{prop:H1HN}
The constants of motion $H_1$ and $H_N$ are given by: 
\begin{equation} \label{eq:H1} 
\boxed{ H_1=\sum _{i,j=1}^N \langle\m_i, \m_j \rangle E_{ij},}
\end{equation} 

\begin{equation} \label{eq:HN}
\boxed{
H_N=\prod_{i=1}^{N-1} (1-E_{i,i+1}^2)\, \prod_{j=1}^N\langle\m_j, \m_j \rangle. } 
\end{equation} 
\end{proposition} 
\begin{remark} These two formulas generalize the formulas previously obtained in \cite{hone-lundmark-szmigielski:novikov}: 
$$ H_1=\sum _{i,j=1}^N m_i m_j E_{ij}, \quad H_N=\prod_{i=1}^{N-1} (1-E_{i,i+1}^2)\, \prod_{j=1}^N m_j^2. 
$$
\end{remark} 
\begin{proof} 
We will omit writing the ranges in the summation formulas.  First, to prove the formula for $H_1$, we write 
$$ 
\begin{gathered} 
\sum_{i, a}\big [T\otimes \mathbf{1} P E P^*\big]_{ia,ia}=\sum_{i,a, j, b, k, l} 
(T\otimes\mathbf{1})_{ia, jb} P_{jb,k} E_{kl} P^*_{l,ia}=\\ \sum_{i,a, j,b, k, l}T_{ij} \delta_{ab} \delta_{jk} \m_{k,b}E_{kl} 
\delta_{li} \ms_{la} =\sum_{i, j, a}T_{ij}  \m_{ja}E_{ji} 
 \ms_{ia}=\sum_{i, a} \ms_{ia}\m_{ia} +\\2 \sum_{\stackrel {a} {i<j}} \m_{ia} \ms_{ja} E_{ij}=\sum_{i,j} \langle \m_i, \m_j \rangle E_{ij},  
\end{gathered} 
$$ 
where in the last two steps we used that $E$ is a symmetric matrix and $\ms_{i}\m_{j}=\ms_{j}\m_i$.  To prove the second 
result, we need more notation.  Let $\mathcal{S}$ be a set of indices $\{11,12, 21,22, \dots, N1, N2\}$ and let $\mathcal{T}=\{1,2,\dots, N\}$.  We denote by $I \subset \mathcal{S}$ 
an index set with cardinality $\abs{I}$.  Given a pair of index sets $I,J$ of cardinality $k$, we write the minor of a matrix $A$ whose rows are indexed by $I$ and columns by $J$ as 
$\det{A}_{I,J}$.  Thus, by definition, 
$$ 
H_N=\sum_{\stackrel{I\in \mathcal{S}}{|I|=N}}\det [(T\otimes \mathbf{1}  )PE P^*]_{I,I}.  
$$
Then the Cauchy-Binet formula implies
$$
H_N= \sum_{\stackrel{I, J\in \mathcal{S}, K, L\in \mathcal{T} }{|I|=|J|=|K|=|L|=N}}\det (T\otimes \mathbf{1})_{I,J}  \det P_{J,K} \det E_{K,L} P^*_{L,I}.  
$$ 
Since there is only one index set in $\mathcal{T}$ of cardinality $N$, namely $\underline{N}=\{1,2,3, \dots, N\}$, we obtain 
$$
H_N=\det E \, \sum_{\stackrel{I, J\in \mathcal{S} }{|I|=|J|=N}}\det (T\otimes \mathbf{1} )_{I,J}  \det P_{J,\underline{N}} \det P^*_{\underline N,I}=\det E \det \big(P^*(T\otimes \mathbf{1} ) P\big).  
$$
Let us now compute the matrix entries of $P^*(T\otimes \mathbf{1} ) P$.  We obtain
$$ 
\begin{gathered} 
\big(P^*(T\otimes \mathbf{1}) P)_{ij}=\sum_{k, a, l, b}  P^*_{i, ka} (T\otimes \mathbf{1} )_{ka,lb} P_{lb, j}=\\
\sum_{k, a, l, b}  \delta_{ik} \ms_{ia} T_{kl} \delta_{ab} \delta_{lj} \m_{jb}=\ms_i\m _j T_{ij}=\langle \m_i, \m_j \rangle T_{ij}.  
\end{gathered} 
$$
Thus $\big(P^*(T\otimes \mathbf{1} ) P)$ is lower-triangular with  diagonal entries $\langle \m_j, \m_j\rangle$.  \newline
Hence $$\det\big(P^*(T\otimes \mathbf{1}) P\big)=\prod_{j=1}^N \langle \m_j, \m_j\rangle. $$

 Finally, $\det E$ was computed in \cite{hone-lundmark-szmigielski:novikov} to be
$$ 
 \det E=\prod_{i=1}^{N-1} (1-E_{i, i+1}^2).  
$$  
This concludes the proof of the second formula.  
\end{proof}

\subsection{Global existence of peakon flows} \label{subsec:global existence}
First, we briefly summarize our strategy.  The spectral problem \eqref{eq:xLax}, \eqref{eq:BCs} is not self-adjoint.  
In order to get information about the spectrum, we take advantage of isospectrality: we isospectrally deform the boundary value problem to $t\rightarrow \infty$ where the problem simplifies.  For this strategy to work, it is necessary to establish the global existence of 
the isospectral flow induced from 
the peakon equations \eqref{eq:epeakons} 
\begin{equation*} 
\boxed{ 
\begin{gathered} 
\dot x_j=u(x_j) v(x_j), \\
\dot  m_j =-m_j \langle u_x \rangle (x_j) v(x_j), \qquad \hspace{1cm} \qquad \qquad \dot n_j =-n_j \langle v_x \rangle  (x_j) u(x_j).  
\end{gathered} }
\end{equation*}

\begin{remark} 
To the best of our knowledge, J. Moser was the first to use this type of approach in his paper on the finite Toda lattice \cite{moser:1975:Toda, moser:three-integrable}. 
His method was adapted to the peakon problem for the DP equation in \cite{lundmark-szmigielski:DPlong}.  The remainder of this section builds on that latter work.  We also remark that recently Moser's integrability argument was used in  \cite{deift-li-spohn-tomei-trogdon:todawforcing} to establish Lax integrability of the open Toda chain with forcing.  
\end{remark}

First, we will comment on the \emph{velocities} $\dot x_j$ and the \emph{momenta} $m_j, n_j$.  

\begin{proposition} \label{prop:dotxj} 
Let $\mathcal{P}=\{x_1<x_2,\dots<x_N, m_j, n_j >0, \text{ for all } 1\leq j\leq N \}$ and let $\underline x, \underline m, \underline n$ 
denote the $N$-tuples of positions and momenta.  Suppose $(\underline x, \underline m, \underline n)(t=0) \in \mathcal{P}. $ 
Then 
\begin{enumerate} 
\item  $(\underline x, \underline m, \underline n)(t) \in \mathcal{P}$ for any  $t>0$. 
\item peakons move to the right; 
\item the velocities of peakons are uniformly bounded from above:
$$ 
\dot x_j\leq \frac12 H_1, \hspace{2cm} 1\leq j\leq N; 
$$
 \item the momenta $m_j, n_j$ are bounded from above:
$$ 
m_j\leq \sqrt{h_1}, \qquad  n_j\leq \sqrt{h_2}; 
$$
\item the momenta $m_j, n_j$ are bounded away from zero.  
\end{enumerate} 
\end{proposition} 
\begin{proof} 
Restricted to $\mathcal{P}$, the vector field defining the system of ODEs \eqref{eq:epeakons} is Lipschitz, so the 
local existence of solutions is ensured.  The equation for $m_j$ implies that if $m_j(0)>0$ then $m_j(t)>0$ as long as the solution 
exists.  The same is true for $n_j$.  From the form of $H_1$ we conclude that 
$$ 
\langle\m_j, \m_j \rangle=2m_jn_j< H_1, 
$$ 
while $H_N$ guarantees that all these products are bounded away from $0$.  Thus, at all times, 
$$0<c_j<\langle \m_j, \m_j \rangle <d_j$$ for some positive constants $c_j, d_j$.  

Moreover, $x_j(t)\neq x_{j+1}(t)$ (no collisions), since otherwise the constancy of $H_N$ would be violated.  
Before we prove that neither $x_j$ nor $m_j, n_j$ can escape to infinity in finite time we note that $uv>0$ on $\mathcal{P}$,  
hence $0<\dot x_j(t)$ (which means that peakons are moving to the right).  
Moreover, upon explicitly writing $\dot x_j$ we get: 
$$ 
\begin{gathered} 
\dot x_j=\sum_{i, k}m_in_k e^{-(\abs{ x_j-x_i}+\abs{x_j-x_k})}\leq \sum_{i,k} m_i n_k e^{-\abs{x_i-x_k}}=\\
\frac12 \sum_{i,k} (m_in_k+m_kn_i) e^{-\abs{x_i-x_k}}=\frac12 H_1, 
\end{gathered} 
$$
which proves the claim.  Thus $0<x_j(t)\leq \frac12 H_1 t +C $; hence $x_j$ cannot diverge to infinity in finite time.  
As to $m_j$ and $n_j$ we note that 
\begin{equation} \label{eq:mjbound}
m_j^2\leq m_j \sum_{i} m_i e^{-\abs{x_i-x_j}}\leq 2m_ju(x_j)\leq h_1
\end{equation} 
by \autoref{prop:h1h2p}.  Hence $m_j$ remains bounded and the same is true for $n_j$.  Finally, since 
$\langle \m_j, \m_j \rangle$ is bounded away from zero, both $m_j$ and $n_j$ are bounded away from zero as well.  
\end{proof} 
Clearly, the proof above indicates (see \eqref{eq:mjbound} ) that the potentials $u$ and $v$ are uniformly bounded in $x$ and $t$.  More precisely we have the following. 
\begin{corollary} \label{cor:uvbounds} 
\mbox{}
$$ 
||u||_\infty\leq N \sqrt{h_1},  \qquad ||v||_\infty\leq N \sqrt{h_2}. 
$$ 
\end{corollary} 

Since there are no blowups of $x_j$ or $m_j, n_j$ the flows exist for all finite times.  
Next we establish the existence of limits at $t\rightarrow \infty$.

\begin{theorem} \label{thm:scattering} 
\mbox{}
  
\begin{enumerate} 
\item The momenta $m_j(t), n_j(t)$ have finite, positive, limits when $t\rightarrow \infty$.  
\item Let $i\neq j$. The integrals $\int_0^\infty e^{-\abs{x_i(t)-x_j(t)}} dt$ are bounded.  
\item The distinct ($i\neq j$) NV2 peakons scatter, that is when  $t\rightarrow \infty$, 
$$
\abs{x_i(t)-x_j(t)}\rightarrow \infty.   
$$
\end{enumerate} 

\end{theorem} 
\begin{proof}  First we note that 
$x_j$ converges to $\infty$ as $t\rightarrow \infty$.  Indeed, 
$m_j<u(x_j), n_j<v(x_j)$ hence $m_jn_j<u(x_j)v(x_j)=\dot x_j$.  Since, as was shown in the course of the proof of 
\autoref{prop:dotxj}, $m_jn_j$ is bounded away from zero, hence, in terms of the notation from that proof 
$$ 
0<c_j/2<\dot x_j, 
$$ 
and thus 
$$ 
\frac{c_j}{2}  t+d_j\leq x_j, 
$$ 
for some constant $d_j$, from which $\lim_{t\rightarrow \infty} x_j(t)=\infty$ follows.  
The remainder of the proof goes by induction on the number of sites (parametrized by $j$) measured from the right utmost site 
corresponding to the index $N$.  The base case is the statement of the theorem for $j=N$ and $i<N$.  
Recall 
$$ 
\dot m_N=-m_N \avg{u_x}(x_N)v(x_N).  
$$
An easy computation shows that 
$$ 
\avg{u_x}(x_N)=-\sum_{j<N} m_j e^{-\abs{x_N-x_j}} <0.  
$$
Hence $\dot m_N$ is strictly positive and thus $m_N(t) $ is an increasing, uniformly bounded, function on $[0,\infty)$.  
This proves that $\lim_{t\rightarrow\infty } m_N(t)$ exists and in fact is bounded from above by $\sqrt{h_1}$  (see the proof 
of \autoref{prop:dotxj}).  Essentially the same proof applies to $n_N$ with $\sqrt{h_2}$ replacing $\sqrt{h_1}$.  
Let us now write the solution  for $m_N$.  
We have 
$$ 
m_N(t)=m_N(0) e^{-\int_0^t \avg{u_x}(x_N) (\tau) v(x_N(\tau)) d\tau} .  
$$

Thus the integral $-\int_0^\infty \avg{u_x}(x_N(t) )v(x_N(t)) dt$ is convergent.  
Writing it explicitly, we see that 
$$ 
\sum_{i=1}^{N-1} \sum_{j=1}^N \int_0^\infty m_i(t) n_j(t) e^{-(\abs{x_i-x_N}+\abs{x_j-x_N})} dt
$$ 
must converge.  Since both $m_i$ and $n_j$ are bounded from below, the integrals 
$$ 
\int_0^\infty e^{-(\abs{x_i-x_N}+\abs{x_j-x_N})} dt
$$ 
must converge.  In particular, setting $j=N$ and $i<N$, we conclude that 
$$ 
\int_0^\infty e^{-\abs{x_i-x_N}} dt \qquad 
$$ 
must converge.  Note that 
$$ 
\abs{\frac{d}{dt} e^{x_i-x_N}}=\abs{(\dot x_i-\dot x_N) e^{x_i-x_N}}
\stackrel{by \ \autoref{prop:dotxj}}{\leq}H_1.
$$

By \autoref{lem:A} in \autoref{app1},  
$$\lim_{t \rightarrow \infty} e^{x_i(t)-x_N(t)}=0, \qquad i<N, $$ thus implying that $$\lim_{t\rightarrow\infty} \abs{x_i(t)-x_N(t)}=\infty, \qquad i<N, $$
which concludes the proof of the base case.  Now, we assume the induction hypothesis down to the level $j+1$, and consider 
$$ 
m_j(t)=m_j(0) e^{-\int _0^t \avg{u_x(x_j(\tau)} v(x_j(\tau)d\tau}. 
$$
We note that $-\avg{u_x(x_j)}v(x_j)$ can be written: 
$$ 
-\avg{u_x(x_j)}v(x_j)=\sum_{i<j} m_i e^{-\abs{x_i-x_j}} v(x_j) -\sum_{i>j} m_i e^{-\abs{x_i-x_j}} v(x_j). 
$$
Hence
$$ 
m_j(t) e^{ \sum_{i>j}\int_0^t  m_i(\tau)  e^{-\abs{x_i(\tau)-x_j(\tau)}} v(x_j(\tau)) d\tau} =m_j(0) e^{\sum_{i<j} \int_0^t m_i(\tau) e^{-\abs{x_i(\tau)-x_j(\tau)}}v(x_j(\tau))d\tau}.  
$$
The right hand side is increasing so the left hand side has a limit, possibly $+\infty$.  
However, the exponent on the left hand side is dominated by $C \sum_{i>j}\int_0^\infty  e^{-\abs{x_i(t)-x_j(t)}}  dt$ which by induction hypothesis is 
convergent.  Hence the exponential part on the left converges, and thus $m_j(t)$ converges.  The limit $m_j(\infty)\leq \sqrt{h_1}$.  Subsequently, 
the right hand side of the same expression above converges, and the individual integrals 
$$ 
\int_0^t m_i(\tau) n_k(\tau)e^{-\abs{x_i(\tau)-x_j(\tau)}}  e^{-\abs{x_k(\tau)-x_j(\tau) }} d\tau \qquad i<j, \quad 1\leq k\leq N, 
$$ 
must converge as well.  Repeating the same arguments as in the proof of the base case, this time choosing $k=j$,  we obtain 
$$ 
\int_0^\infty e^{-\abs{x_{i}(t)-x_j(t)}} dt<\infty, \qquad i<j, 
$$ 
and hence, for $i<j$, $\lim_{t\rightarrow \infty}e^{-\abs{x_{i}(t)-x_j(t)}}=0$,  from which $\lim_{t\rightarrow \infty} \abs{x_i(t)-x_j(t)}=\infty$ follows, thus concluding the proof.  

\end{proof} 
\begin{corollary} \label{cor:uxvint} 
The integrals 
$$ 
\int_0^\infty \avg{u_x}(x_j(t))v(x_j(t)) dt, \qquad \int_0^\infty \avg{v_x}(x_j(t))u(x_j(t))dt,  
$$ 
converge for all indices $1\leq j\leq N$.  
\end{corollary} 
\begin{proof} 
We give a proof only for the first estimate, as the second one can be obtained by switching the roles of $u$ and $v$.  We have 
$$ 
\begin{gathered} 
\abs{ \int_0^\infty \avg{u_x(x_j(t))}v(x_j(t)) dt} \leq \sum_{i\neq j}\sum_{k} \int_0^\infty m_i(t) n_k(t)e^{-(\abs{x_j(t)-x_i(t)}+\abs{x_j(t)-x_k(t)})} dt
\leq \\ C \sum_{i\neq j} \int_0^\infty e^{-\abs{x_j(t)-x_i(t)}} dt \stackrel{\autoref{thm:scattering}}{<} \infty.   
\end{gathered} 
$$
\end{proof} 
\begin{remark} We will later establish sharper estimates for the 
behaviour of  $\avg{u_x}(x_j(t))v(x_j(t))$ and $\avg{v_x}(x_j(t))u(x_j(t))$ (see \eqref{eq:uxv sharp est}). 
\end{remark} 
\begin{proposition} \label{prop:mnlimits}
\mbox{}
\begin{enumerate} 
\item $x_j(t)=m_j(\infty) n_j(\infty) t+o(t),  \qquad \text{as } t\rightarrow \infty. $
\item $0<m_1(\infty)n_1(\infty)\leq\cdots \leq m_N(\infty)n_N(\infty)$. 
\end{enumerate} 
\end{proposition} 

\begin{proof}

The first  statement is equivalent to the statement $\lim_{t\rightarrow \infty} \frac{x_j(t)}{t}=m_j(\infty)n_j(\infty)$.  
However, by l'Hospital's rule and \eqref{eq:epeakons}, 
$$ 
\lim_{t\rightarrow \infty} \frac{x_j(t)}{t}=\lim_{t\rightarrow \infty} \frac{\dot x_j(t)}{1}\stackrel{\autoref{thm:scattering}}{=} m_j(\infty)n_j(\infty). 
$$
The second item follows from the ordering condition $x_j(t)< x_{j+1}(t)$ which holds for all $t$ and the asymptotic formula for 
$x_j(t)$ already proven.

\end{proof} 
The products $m_j(\infty)n_j(\infty)$ can be interpreted as asymptotic velocities of peakons.  Below, we show that the  inequalities between asymptotic velocities can be sharpened.

\begin{theorem} \label{thm:asymptotic velo} 

The asymptotic velocities are strictly ordered: 

$$0<m_1(\infty)n_1(\infty)<m_2(\infty)n_2(\infty) \cdots <m_N(\infty)n_N(\infty). 
$$ 

\end{theorem} 
\begin{proof} 
It is easy to check that $\avg{u_x}(x_N) v(x_N)<0$ and $\avg{u_x}(x_1) v(x_1)>0$, and 
the same is true if we swap $u$ with $v$ in both statements.  In view of equation \eqref{eq:epeakons},  
$m_N(t)n_N(t)$ is a strictly increasing function, while $m_1(t)n_1(t)$ is a strictly decreasing one.  

Let us write 
$$ 
\dot x_N-\dot x_1=m_Nn_N-m_1n_1 + \text{ terms involving decaying exponentials of distances}. 
$$
Thus by \autoref{thm:scattering} 
$$ 
x_N(t)-x_1(t)=\int_0^t \big(m_N(\tau) n_N(\tau)-m_1(\tau)n_1(\tau) \big) d\tau + O(1),  \qquad t\rightarrow \infty, 
$$
and the integral has to diverge to $\infty$, which cannot happen if the integrand  is non-positive.  Thus there must 
exists a point $t_1>0$ at which $m_N(t_1)n_N(t_1)> m_1(t_1)n_1(t_1)$.  Since $m_N(t) n_N(t)$ is strictly increasing 
and $m_1(t)n_1(t)$ is strictly decreasing, the strict inequality can be extended to all $t>t_1$, implying that the limits $m_N(\infty) n_N(\infty) $ and $m_1(\infty) n_1(\infty)$ are distinct.  
So far, we proved  
$$ 
0<m_1(\infty)n_1(\infty)< m_N(\infty) n_N(\infty),  
$$
which ensures that there is at least one strict inequality between the limits in the list 
$0<m_1(\infty)n_1(\infty)\leq\cdots \leq m_N(\infty)n_N(\infty)$.  

 We will now proceed by induction on the number of strict inequalities 
counted from left to right.  First, we prove that the inequality between the first two limits is strict.  Suppose not;  
then for some $1<a\leq N-1$ we have 
$$ 
m_1(\infty)n_1(\infty)=m_2(\infty)n_2(\infty)=\cdots=m_a(\infty)n_a(\infty)<m_{a+1}(\infty)n_{a+1}\leq \cdots m_N(\infty)n_N(\infty).
$$

We claim that this implies that for sufficiently large $t$, $m_{a}(t)n_{a}(t)$ is monotonically increasing.  Indeed, 
the sign of the derivative is proportional to $-\avg{u_x}(x_{a})v(x_{a})-\avg{v_x}(x_{a})u(x_{a})$.  Since both $u,v$ are positive it suffices to check the signs of $\avg{u_x}$ and $\avg{v_x}$.  We show the proof for the former; the proof of the latter amounts to switching $m_j$s with $n_j$s.  
We start by writing $-\avg{u_x}(x_{a})$ as 
$$ 
\begin{gathered} 
-\avg{u_x}(x_{a})=\sum_{j<a} m_j e^{x_j-x_{a}} -m_{a+1} e^{x_a-x_{a+1}}-\sum_{j>a+1} m_je^{x_{a}-x_j}=\\
e^{x_a-x_{a+1}}\big(\sum_{j<a} m_j e^{x_j-x_{a}+x_{a+1}-x_a}-m_{a+1} -\sum_{j>a+1} m_je^{x_{a}-x_j+x_{a+1}-x_a}\big)=\\
e^{x_a-x_{a+1}}\big(\sum_{j<a} m_j e^{ct+o(t)}+O(1) \big), \qquad t\rightarrow \infty, 
\end{gathered} 
$$ 
where $c=m_{a+1}(\infty) n_{a+1}(\infty) -m_a(\infty)n_a(\infty)>0$.  Thus the first term dominates for large $t$ and we get that 
$$
-\avg{u_x}(x_{a})>0,   
$$ 
for sufficiently large positive $t$.  The same argument applies to $ -\avg{v_x}(x_{a})$ and this proves that 
$\frac{d m_a(t) n_a(t)}{dt} > 0$ for sufficiently large positive 
$t$, and thus $m_a(t)n_a(t)$ monotonically increases for sufficiently large $t$.  
However, 
$$
x_a(t)-x_1(t)=\int_0^t (m_a(\tau)n_a(\tau)-m_1(\tau) n_1(\tau)) d\tau +O(1), \qquad t\rightarrow \infty, 
$$ 
and the same argument as for the pair $1$ and $N$ shows that $m_1(\infty)n_1(\infty)< m_a(\infty) n_a(\infty)$, which contradicts our assumption.  
Now,  we assume that there are $a-1$ strict inequalities, and we want to show that this implies that there are $a$ strict inequalities.  
Suppose this is not the case, i.e., suppose 
$$ 
\begin{gathered} 
m_1(\infty)n_1(\infty)<m_2(\infty)n_2(\infty)<\cdots<m_{a-1}(\infty) n_{a-1}(\infty)<\\
m_a(\infty)n_a(\infty)=
m_{a+1}(\infty)n_{a+1}(\infty)=\cdots=m_b(\infty)n_b(\infty) <\\
m_{b+1}(\infty) n_{b+1}(\infty)\leq \cdots m_N(\infty)n_N(\infty),  
\end{gathered}
$$ 
where $a+1\leq b\leq N$.  
We will show that $m_a(t)n_a(t)$ is decreasing for sufficiently large $t$, while $m_b(t)n_b(t)$ is increasing for sufficiently large $t$.  
Granting that 
we see that 
$$
x_b(t)-x_a(t)=\int_0^t (m_b(\tau)n_b(\tau)-m_a(\tau)n_a(\tau))d\tau +O(1), \qquad  t\rightarrow \infty
$$ 
implies, again, that $m_a(\infty)n_a(\infty)<m_b(\infty)n_b(\infty)$, hence a contradiction.   
We now give the proof that indeed $m_b(t)n_b(t)$ is increasing for sufficiently large $t$.  
We write $-\avg{u_x}(x_b)$, namely, 
$$ 
\begin{gathered} 
-\avg{u_x}(x_b)=\sum_{j<b} m_je^{x_j-x_b} -m_{b+1} e^{x_b-x_{b+1}}-\sum_{j>b+1} m_j e^{x_b-x_j}=\\
e^{x_b-x_{b+1}}\big(\sum_{j<b} m_je^{x_j-x_b+x_{b+1}-x_b}-m_{b+1}-\sum_{j>b+1} m_j e^{x_{b+1}-x_j}\big).
\end{gathered}
$$
Now, we observe that at least one term in the first sum is unbounded.  Indeed, for $j=b-1$ the exponential is asymptotically equal 
$$ 
 e^{x_j-x_b+x_{b+1}-x_b}=e^{ct+o(t)}, 
 $$ 
 where $c=m_{b+1}(\infty)n_{b+1}(\infty)-m_b(\infty)n_b(\infty)>0$.  
 Since the terms with the negative sign in front are $O(1)$ as $t \rightarrow \infty$, we conclude that $-\avg{u_x}(x_b)>0$ for sufficiently large 
 $t$.  Repeating this argument for $-\avg{v_x}(x_b)$ we obtain that $m_b(t)n_b(t)$ is monotonically increasing for $t$ large enough.  
 Finally, we outline the steps for $-\avg{u_x}(x_a)$.  We write
 $$ 
 -\avg{u_x}(x_a)=e^{x_{a-1}-x_a}\big(\sum_{j<a-1} m_j e^{x_j-x_a+x_a-x_{a-1}}+m_{a-1}-\sum_{j>a} m_j e^{x_a-x_j+x_a-x_{a-1}}\big).
$$ 
Again, the positive terms are bounded, while the first term in the negative part reads
$$ 
e^{x_a-x_{a+1}+x_a-x_{a-1}}=e^{ct+o(t)}, 
$$
where, this time, $c=m_a(\infty)n_a(\infty)-m_{a-1}(\infty)n_{a-1}(\infty)>0$.  
The remaining steps mimic the earlier arguments in this proof, yielding the conclusion that $m_a(t)n_a(t)$ is monotonically decreasing 
for sufficiently large $t$, which concludes the proof.  
\end{proof} 
Since the asymptotic velocities $m_j(\infty)n_j(\infty)$ are distinct we have a sharper characterization of the asymptotic 
behavior of positions.  

\begin{proposition} \label{prop:as positions final} 
Asymptotically, 

$$
\boxed{
x_j(t)=m_j(\infty)n_j(\infty)t+O(1), \qquad t\rightarrow \infty.}
$$  
\end{proposition} 
\begin{proof} 
We need to estimate
$$
\begin{gathered} 
x_j(t)-m_j(\infty)n_j(\infty)t=x_j(0)+\int_0^t (\dot x_j(\tau)-m_j(\infty)n_j(\infty)) d\tau\stackrel{\autoref{thm:scattering}}{=}\\
\int_0^t \big(m_j(\tau)n_j(\tau)-m_j(\infty)n_j(\infty)\big)d\tau +O(1), \qquad t\rightarrow \infty.
\end{gathered}
$$ 
We claim that 
$$ 
\abs{m_j(\infty)n_j(\infty)-m_j(\tau)n_j(\tau)}=O(e^{-ct}), \qquad t\rightarrow \infty, 
$$ 
for some $c>0$.  
We write 
$$ 
\begin{gathered} 
\abs{m_j(\infty)n_j(\infty)-m_j(\tau)n_j(\tau)}=\abs{\int_\tau^\infty \frac{d (m_jn_j)}{ds} ds}\\\leq\int_\tau \abs{m_j(s)n_j(s)} 
\abs{\big(\avg{u_x}v+u\avg{v_x}\big)(x_j(s)}ds.   
\end{gathered}
$$
We estimate 
$$ 
\begin{gathered} 
\abs{(\avg{u_x}v)(x_j(s))}\leq \sum_{i\neq j} \sum_{k} m_i(s)n_k(s) e^{-\abs{x_j(s)-x_i(s)}-\abs{x_j(s)-x_k(s)}}\\
\leq \sum_{i\neq j} \sum_{k} m_i(s)n_k(s) e^{-c_{ijk} t+r_{ijk}}, 
\end{gathered} 
$$
where by \autoref{thm:asymptotic velo}  the $c_{ijk}$ are strictly positive, and the $r_{ijk}=o(t) \mbox{ as } t\rightarrow \infty$.  
Since all $m_i, n_k$ are bounded we get 
\begin{equation} \label{eq:uxv sharp est} 
\abs{(\avg{u_x}v)(x_j(s))}\leq D e^{-ct} , \mbox{    for some } c>0, 
\end{equation} 
and the same argument works for $\abs{(\avg{v_x}u)(x_j(s))}$.  
Hence,  $$\abs{m_j(\infty)n_j(\infty)-m_j(\tau)n_j(\tau)}=O(e^{-c\tau})$$  for some $c>0$, and 
the integral $\int_0^\infty  (\dot x_j(\tau)-m_j(\infty)n_j(\infty)) d\tau$  converges absolutely.  

\end{proof} 
\subsection{Characterization of eigenvalues} \label{subsec:eigenvalues}
Recall from \autoref{prop:characteristic polynomial} 
that the spectrum (in terms of the spectral variable $\l$) is given 
by the zeros of 
\begin{equation*} 
\boxed{ A(\l)=\det{( I_{2N}-\l[(T\otimes \mathbf{1} )P E P^*]}). }
\end{equation*} 
Since our deformation is isospectral,  and it exists for any $t>0$ including $t=\infty$,   one can obtain information 
about the spectrum by evaluating the characteristic polynomial  at $t=\infty$.  
\begin{proposition} \label{prop:eigenvalues} 
The reciprocals of the eigenvalues in the  eigenvalue problem \eqref{eq:MEVP} are: 
$$ 
\boxed{ 
\l _j=\frac{1}{2 m_j(\infty)n_j(\infty)}=\frac{1}{\langle \m_j, \m_j \rangle (\infty)}, \qquad 1\leq j\leq N.}
$$ 
In particular, all non-zero eigenvalues are simple.  

\end{proposition} 
\begin{proof} 
Since $E_{ij}=e^{-\abs{x_i-x_j}}$,  the limit of $E$ as $t\rightarrow \infty$ is the identity matrix $I_N$.  
Moreover, using isospectrality of the deformation, we can compute the asymptotic form of the characteristic polynomial $A(\l)$ without changing its roots (eigenvalues).  
Let us assign the symbol $A_\infty$ to denote that asymptotic polynomial.  Thus 
$$ 
A_\infty(\l)=\det{( I_{2N}-\l[(T\otimes \mathbf{1})P(\infty) I_N P^*(\infty)])}.   
$$
Let us now compute the matrix entries of $(T\otimes \mathbf{1})P(\infty) I_N P^*(\infty)$.  
Using the same basis as in computations in \autoref{subsec:xLax}, in particular 
formulas \eqref{eq:matrix entries},  we obtain
$$ 
[(T\otimes \mathbf{1})P(\infty)I_N P^*(\infty)]_{ia, jb}=T_{ij}\m_{ja}(\infty)  \ms_{jb}(\infty).  
$$ 
This shows that $(T\otimes \mathbf{1})P(\infty) I_N P^*(\infty)$ is lower-block-triangular: 
$$
(T\otimes \mathbf{1})P(\infty) I_N P^*(\infty)=\begin{bmatrix} \m_1\ms_1   &0_{2\times 2} &0_{2\times2}&\hdots&0_{2\times 2} \\
2\m_1\ms_1&\m_2 \ms_2&0_{2\times 2} &\hdots&0_{2\times 2} \\
\vdots&\vdots&\vdots&\m_{N-1}\ms_{N-1}&0_{2\times 2} \\
2 \m_1\ms_1&2\m_2\ms_2 &2\m_3\ms_3 &\hdots&\m_N \ms_N  \end{bmatrix}\big(\infty\big).  
$$
Hence, 
$$ 
A_\infty(\l)=\prod_{j=1}^N \det{( \mathbf{1} -\l \m_j(\infty)\ms_j(\infty)))}=\prod_{j=1}^N  \big( 1 -2\l m_j(\infty)n_j(\infty)\big).   
$$
Finally, the identity $2m_j(\infty)n_j(\infty)=\langle \m_j, \m_j\rangle(\infty) $ and \autoref{thm:asymptotic velo} imply the claim.  
\end{proof} 
Combining \autoref{prop:eigenvalues} with \autoref{prop:as positions final} results in the final description of the asymptotics of the positions.  
\begin{corollary} \label{cor:as-speeds} 
The asymptotic velocity of the $j$-th particle is $\frac{1}{2\lambda_j}$ and the asymptotic form of $x_j$ reads: 
$$
\boxed{ 
x_j(t)=\frac{t}{2\lambda_j}+O(1), \qquad \text{ for } t\rightarrow \infty.} 
$$ 
\end{corollary} 
\begin{corollary}\label{cor:lambdaord} 
The reciprocal eigenvalues $\l_j$ of the matrix eigenvalue problem  \eqref{eq:MEVP} form a decreasing sequence, i.e., 
\begin{equation} \label{eq:lord} 
\l_1>\l_2>\cdots>\l_N>0.
\end{equation} 
 
\end{corollary} 
We now revisit the original boundary value problem \eqref{eq:xLax}, \eqref{eq:BCs}.  We recall that the 
relation between the spectral variables $z$ and $\l$ is given by $z^2=-\l$.   
\begin{corollary}\label{cor:BVspectra} 
The eigenvalues $z$ of the original boundary value problem \eqref{eq:xLax}, \eqref{eq:BCs} are purely imaginary, and come in complex conjugate 
pairs $\pm i  \sqrt{\l_j}$.  
\end{corollary}
This result fully corroborates \autoref{prop:spectrum1}.  

\subsection{Characterization of $\mathbf{B(\l)} $ on the spectrum}  \label{subsec:Bvalues} 
Recall from previous sections the parametrization of $\mathbf{\phi_2}$ we are using throughout the paper, namely, 
$$ 
\mathbf{\phi_2}(x,z)=2 z \mathbf{B}(x,\l), 
$$ 
where $\l =-z^2$.  In the region $x>x_N$, we drop $x$ from the notation. In particular, $\mathbf{\phi_2}(x_N+,z)$ is denoted $2z\mathbf{B(\l)}$.    
We are interested now in describing the behavior of $\mathbf{B}(\l)$ on the spectrum; we will use it in our analysis of the Weyl functions in the 
forthcoming section.  Let us denote by $||f||_{H^1}$ the Sobolev norm of $f\in H^1$.  

\begin{proposition} \label{prop:B1B2}
Let  $B_{1} (\l)$ and $B_{2} (\l)$ denote the components of $ \mathbf{B(\l)}$.   Then 
\begin{equation*} 
\boxed{ ||\phi_3(\cdot, \l_j) ||^2_{H^1}= 4 \l_j B_{1}(\l_j) B_{2} ( \l_j)=2\l_j \langle \mathbf{B}(\l_j), \mathbf{B}(\l_j) \rangle, \text{  for all } 1\leq j\leq N.}
\end{equation*} 
\end{proposition} 
\begin{proof} 
We first give the proof in the smooth case to get an insight into the origin of the claimed result.   Thus we initially assume that $\m$ is smooth and of compact support.  
Using equations \eqref{eq:psi123}, \eqref{eq:psi2psi3} and \eqref{eq:psi3psi1}, we obtain
$$ 
\phi_3 D_x^2\phi_3-\phi_3^2 =z\langle\m, \mathbf{\phi_2}\rangle  \phi_3=z\ms \phi_3\mathbf{\phi_2}.  
$$
We now eliminate $z\ms \phi_3$ using \eqref{eq:psi2psi3} to get 
$$ 
\phi_3 D_x^2\phi_3-\phi_3^2 =\langle D_x\mathbf{\phi_2}, \mathbf{\phi_2}\rangle=\frac12 D_x\langle\mathbf{\phi_2}, \mathbf{\phi_2} \rangle, 
$$ 
which upon integration localized at the eigenvalues $\l_j=-z_j^2$ gives 
$$ 
\begin{gathered} 
-\big( \int_{\R} ((D_x\phi_3(x,\l_j))^2 +\phi_3^2 (x,\l_j)) dx\big)=\frac12 \int_{\R}  D_x \langle\mathbf{\phi_2}, \mathbf{\phi_2} \rangle dx \\=\frac12 \langle\mathbf{\phi_2}, \mathbf{\phi_2} \rangle(\infty).  
\end{gathered} 
$$ 
Hence 
$$ 
||\phi_3||_{H^1}=-\frac12 \langle\mathbf{\phi_2}, \mathbf{\phi_2} \rangle(\infty). 
$$
Recalling that $\mathbf{\phi_2} (\infty,z_j)=2z_j \mathbf{B}(\l_j)$  and  $\l_j=-z_j^2$ gives the final result in the smooth case.  
We now turn to the case of $\m$ given by \eqref{eq:MMsigma}.  
Since $\phi_3$ is continuous and piecewise smooth, while $\mathbf{\phi_2}$ is piecewise constant, we have two equations in distributions: 
\begin{subequations} \label{eq:phi32distr}
\begin{align}
\phi_{3,xx} -\phi_3 +\sum_{k=1}^N [\phi_{3,x}](x_k)\delta_{x_k} &=2z_j \sum_{k=1}^N \ms_k \avg{\mathbf{\phi_2}}(x_k) \delta_{x_k}, \\
\sum_{k=1}^N [\mathbf{\phi_2}](x_k)\delta_{x_k} &=2z_j \sum_{k=1} \m_k \phi_3(x_k) \delta_{x_k},  
\end{align}  
\end{subequations}
where $[f](a)$ denotes the jump of $f$ at $x=a$.  
Multiplying the first equation by $\phi_3(x,\l_j)$ and integrating over $\R$ we obtain
$$ 
\begin{gathered} 
\int_\R \big(\phi_3(\phi_{xx}-\phi_3) \big)(x,\l_j) dx +\sum_{k=1}^N [\phi_{3,x}](x_k,\l_j)\phi_3(x_k,\l_j)\\
=2z_j \sum_{k=1}^N \ms_k \avg{\mathbf{\phi_2}}(x_k,\l_j) \phi_3(x_k,\l_j).   
\end{gathered} 
$$ 
Performing integration by parts on the left-hand side, recalling that on the spectrum $\phi_3(\infty, \l_j)=0$, and using the second equation in \eqref{eq:phi32distr} 
results in 
$$ 
||\phi_3(\cdot,\l_j)||^2_{H^1}=-\sum_{k=1} ^N [\mathbf{\phi_2^*}](x_k,\l_j)\avg{\mathbf{\phi_2}}(x_k,\l_j). 
$$
We note that the sum above involves terms of the type 
$$
(a_j^*-a_{j-1}^*)\frac{(a_j+a_{j-1})}{2} =\frac12(a_j^* a_j -a_{j-1}^* a_{j-1})=\frac12( \langle a_j, a_j \rangle - \langle a_{j-1}, a_{j-1} \rangle),  
$$ 
and thus comprises a telescoping sum, yielding $- \frac{\langle\mathbf{\phi_2}, \mathbf{\phi_2}\rangle (x_N+, \l_j)}{2}=-2z_j^2 \langle \mathbf{B} (\l_j), \mathbf{B} (\l_j) \rangle $.  Since $z_j^2=-\l_j$, the claim is proven.  
\end{proof} 
The following elementary corollary about the components of $\mathbf{B}(\l)$  is used in the next section.

\begin{corollary} \label{cor:Bsigns}
\mbox{}
\begin{enumerate} 
\item 
The components $B_{1}(\l_j), B_{2}(\l_j) $ are always non-zero.  
\item 
The components $B_{1}(\l_j),B_{2}(\l_j)$ have the same sign. 
\end{enumerate} 
\end{corollary} 

\section{Weyl functions and Hermite--Pad\'e approximations} \label{sec:weyl}

As was explained earlier, there are two ways of looking at the peakon forward problem.    In broad terms, in the first method, one uses $2N\times 2N$ matrices, while in the second method, one  
uses the products of $4\times 4$  matrices (the transition matrix).  Moreover,  in the first method, the crucial step is to construct $\Psi =[\avg{\mathbf {\psi_2}(x_j)}]$.  
By contrast (see \autoref{prop:transition}), in the second method, the objective is to characterize the polynomials $A(\l), \mathbf{B}(\l), C(\l)$ which are obtained by successive multiplications of 
one-step matrices implementing the transition from the $k-1$ site to the $k$th  site:
\begin{equation} \label{eq:Sk} 
 S_k(\l)=\begin{bmatrix} 1-\l \langle \m_k, \m_k \rangle & -2\l \ms_k e^{-x_k}&-\l^2 \langle \m_k, \m_k \rangle e^{-2x_k}\\
 \m_k e^{x_k}& \mathbf{1}  & \l \m_k e^{-x_k}\\
 \langle \m_k, \m_k \rangle e^{2x_k}&2\ms_k e^{x_k}& 1+\l \langle \m_k, \m_k \rangle \end{bmatrix}\in M_{4\times 4} ,  
 \end{equation} 
 and the k-th step is given by
 \begin{equation} \label{eq:transition} 
 \begin{bmatrix} A_k(\l) \\ \mathbf{B}_k(\l) \\ C_k(\l)  \end{bmatrix} =S_k(\l)  \begin{bmatrix} A_{k-1} (\l) \\ \mathbf{B}_{k-1}(\l)  \\ C_{k-1}(\l)  \end{bmatrix},   
 \end{equation} 
 where $A_k$ and $C_k$ are scalar polynomials in $\l$ and $\mathbf{B}_{k}(\l)  \in M_{2\times 1}$.  
 By stacking together all $S_k$s we obtain the transition matrix whose first column reproduces $A(\l), \mathbf{B}(\l), C(\l)$:
 \begin{equation} \label{eq:S} 
 \begin{bmatrix} A(\l) \\ \mathbf{B}(\l) \\ C(\l) \end{bmatrix} =S_N(\l)S_{N-1}(\l)\cdots S_1(\l) \begin{bmatrix} 1\\ 0_{2\times 1} \\ 0 \end{bmatrix}.
 \end{equation}  
The roots of $A(\l)$ are given by \autoref{prop:eigenvalues}. 
 
 To address the forward problem, and in line with our deformation technique used to construct the spectrum, we will need the long-time asymptotics of individual $A_k,\mathbf{B}_k, C_k$.  
 \begin{proposition} \label{prop:ABCasympt}
 \mbox{}
For $t\rightarrow \infty$, 
\begin{align} 
&A_k(t, \l)=\prod_{i=1}^k (1-\l \langle \m_i, \m_i\rangle(t) ) +O(e^{-\alpha t}),  \\
&\mathbf{B}_k(t,\l) =e^{x_k(t)} \big( \m_k(t) \prod_{i=1}^{k-1} (1-\l \langle \m_i, \m_i\rangle(t))+O(e^{-\alpha t})\big),\\
&C_k(t,\l) =e^{2x_k(t)} \big( \langle \m_k, \m_k \rangle (t)  \prod_{i=1}^{k-1} (1-\l \langle \m_i, \m_i\rangle (t))+O(e^{-\alpha t})\big), 
\end{align} 
where $\alpha=\tfrac12 \min\limits_{1\leq j\leq k-1}(\frac{1}{\l_{j+1}}-\frac{1}{\l_{j}}).$ 
\end{proposition} 
\begin{proof} 
The proof proceeds by induction on $k$.  The base case is straightforward as this is just the first column of $S_k(\l) $ for $k=1$ (see \eqref{eq:Sk}).  
We will only carry out the proof for $\mathbf{B}_k$, leaving other cases for an interested reader.  We initially keep the remainder term as $o(1)$, and only at the end of the proof do we provide a sharper estimate of the rate at which the remainder goes to 0.  
Assuming the claim for $k-1$, we write (dropping the $t$ dependence from the notation)
$$ 
\begin{gathered} 
\mathbf{B}_k=\m_ke^{x_k}  \big(\prod_{i=1}^{k-1}  (1-\l \langle \m_i, \m_i\rangle) +o(1)\big)+e^{x_{k-1}} \big( \m_{k-1} \prod_{i=1}^{k-2} (1-\l \langle \m_i, \m_i\rangle)+o(1)\big)\\
+\l \m_k e^{-x_k} \big( e^{2x_{k-1} } \big( \langle \m_{k-1}, \m_{k-1} \rangle  \prod_{i=1}^{k-2} (1-\l \langle \m_i, \m_i \rangle)+o(1)\big)\\
=e^{x_k} \big(\m_k\prod_{i=1}^{k-1}  (1-\l \langle \m_i, \m_i\rangle) +o(1)\big).  
\end{gathered} 
$$ 
Finally, the remainder term contains sums of the products of $e^{-(x_{j+1}-x_j)}$, with at least one term involving 
a single factor of this kind.  This can be checked again by easy induction.  Then the asymptotic behaviour of positions (see \autoref{cor:as-speeds}) implies the claim about the order of the remainder.  
\end{proof}

\subsection{Weyl functions}

 We will now turn to define the so-called \emph{Weyl functions}.  These are meromorphic (rational in the case of discrete measures) functions that depend on  $A(\l)$, $\mathbf{B}(\l)$, and $C(\l)$, and have poles precisely on the spectrum.   
For easy comparison, we choose definitions modeled on what was used  in \cite{hone-lundmark-szmigielski:novikov}.  
There are altogether three relevant Weyl functions in our problem.  We group two of them as one vector-valued Weyl function; the third 
one is a scalar Weyl function and is not completely independent from the vector Weyl function.  Thus we define   
\begin{equation} \label{eq:defW} 
\mathbf{W(\l)}=-\frac{\mathbf{B(\l)}}{A(\l)}, \qquad Z(\l)=-\frac{C(\l)}{2A(\l)}.  
\end{equation} 
Since the zeros of $A(\l)$ are simple, all these Weyl functions have a simple partial fraction decomposition.  First, however, 
we will look at some illustrative examples of partial fraction decompositions for small values of $N$.

\noindent {\bf Notation:}  
 we adopt the following notation in the remainder of the paper: if a real vector $v\in M_{n,1}$ then $v>0$ means that each scalar 
component $v_j>0$ for $1\leq j\leq n$.
\begin{example} \label{ex:res12}
The point of this example is to get some insight into the shape of the residues of $\mathbf{W}$ and $Z$.  The most relevant question is 
that of the positivity of the residues.

For $N=1$, $\m_1$ does not depend on time, and the one-peakon moves with constant velocity
$$
m_1n_1=m_1(\infty)n_1(\infty)=\frac{\langle \m_1, \m_1\rangle (\infty)}{2}=\frac{1}{2\l_1}.  
$$
  Using \eqref{eq:S} we obtain: 
\begin{equation} 
\begin{gathered} 
A(\l) =1-\l \langle \m_1, \m_1\rangle (\infty)=1-\frac{\l}{\l_1}, \qquad \mathbf{B} (\l)=\m_1e^{x_1(t)}=\m_1(\infty) e^{\frac{t}{2\l_1} +x_1(0)}, \\
C(\l)=\langle \m_1, \m_1\rangle (\infty) e^{2x_1(t)}= \frac{e^{\frac{t}{\l_1}  +2 x_1(0)}}{\l_1}. 
\end{gathered} 
\end{equation}  

 The residues can be readily computed.  The result is:
 \begin{equation*} 
 \begin{gathered}
 \boxed{\text{Res}\{\mathbf{W}; \l_1\}=\l_1 \m_1(\infty) e^{x_1(0)}e^{\frac{t}{2\l_1}}>0,  }\\
 \boxed{\text{Res}\{Z; \l_1\}=e^{2x_1(0)}e^{\frac{t}{\l_1}}>0,}
 \end{gathered}
 \end{equation*} 
 so all residues are positive.  
 
 In preparation for a general discussion we note that we can concentrate on $\mathbf{W}$, hence on $\mathbf{B(\l)}$, since $Z$ can be constructed using $\mathbf{W}$.  
$\mathbf{B}(\l)$ evolves linearly on the spectrum, namely, 
$$ 
\frac{d\mathbf{B}(\l_j)}{dt}=\frac{\mathbf{B}(\l_j)}{2\l_j}, 
$$ 
implying that the time evolution of the residues of $\mathbf{W}$ satisfies the same equation and the sign of individual residues can not change under the time flow.  
Hence, in order to determine the sign of $\text{Res}\{\mathbf{W}; \l_j\}$, it suffices to know the sign of $\mathbf{B}(\l_j)e^{-x_j}$ in the limit as $t\rightarrow\infty$, remembering that asymptotically $x_j=\frac{t}{2\l_j}+O(1)$.  In the computation of the sign of the residues, we need to know 
the sign of $-A'(\l_j)$, which can be easily computed to be $(-1)^{N-j} $ as a reflection of the asymptotic ordering of velocities.   Thus for the residues $\text{Res}\{\mathbf{W}, \l_j\}$ to be strictly positive we need 
$\sgn(\lim_{t\rightarrow \infty} \mathbf{B}(\l_j)e^{-x_j})=(-1)^{N-j}$. 

Now, we turn to the case $N=2$.  We use equation \eqref{eq:S} to compute $\mathbf{B}(\l)$.  We obtain
$$ 
\mathbf{B}(\l)=\m_2e^{x_2} (1-\l \langle \m_1, \m_1\rangle)+\m_1e^{x_1}+\l \m_2 \langle \m_1, \m_1\rangle e^{2x_1-x_2}.
$$ 
We start with $j=2$ and compute 
$$
\lim_{t\rightarrow \infty} (\mathbf{B}(\l_2) e^{-x_2})=\m_2(\infty) (1-\frac{\l_2}{\l_1})>0,  
$$
which means that the residue of $\mathbf{W}$ at $\l_2$ is positive.  
The computation for $\l_1$ is only slightly more involved.  
\begin{equation*} 
\begin{gathered} 
\lim_{t\rightarrow \infty} (\mathbf{B}(\l_1) e^{-x_1})=\m_2(\infty)\lim_{t\rightarrow \infty} \frac{(1-\l_1\langle \m_1, \m_1\rangle}{e^{x_1-x_2}} 
+\m_1(\infty)\\=\m_2(\infty)\lim_{t\rightarrow \infty} \frac{(1-2\l_1 m_1(t) n_1(t)) }{e^{x_1-x_2}}+\m_1(\infty). 
\end{gathered} 
\end{equation*} 

We recall that 
$$ 
\frac{d (m_1(t) n_1(t))}{dt}=-m_1(t)n_1(t) (\avg{u_x}(x_1)v(x_1)+\avg{v_x}(x_1) u(x_1)), 
$$ 
from which, with the help of l'Hospital's rule, one obtains: 
$$ 
\lim_{t\rightarrow \infty} \frac{(1-2\l_1 m_1(t) n_1(t)) }{e^{x_1-x_2}}=\frac{m_2(\infty)n_1(\infty)+n_2(\infty) m_1(\infty)}{m_1(\infty)n_1(\infty)-m_2(\infty)n_2(\infty)}.  
$$
We note that the denominator is negative so we only need to show that 
$$ 
\m_2(\infty)(m_2(\infty)n_1(\infty)+n_2(\infty) m_1(\infty)+\m_1(\infty)(m_1(\infty)n_1(\infty)-m_2(\infty)n_2(\infty))>0.  
$$
In the remainder of this computation, we will skip the reference to $\infty$.  Thus we will write $m_1$ rather than $m_1(\infty)$, etc. .  
Then 
$$ 
(m_2n_1+n_2 m_1)\begin{bmatrix} m_2\\ n_2 \end{bmatrix} +(m_1n_1-m_2n_2)\begin{bmatrix} m_1\\ n_1 \end{bmatrix}=
\begin{bmatrix} n_1(m_1^2+m_2^2)\\ m_1 (n_1^2+n_2^2) \end{bmatrix} >0, 
$$ 
which concludes the proof that both residues of $\mathbf{B}$ are positive for $N=2$.    
\end{example} 
Now we turn to the general case.  
First, we prove that the residue at the smallest eigenvalue $\l_N$ is always positive.  
\begin{proposition}\label{prop:resN} 
For all initial conditions in $\mathcal{P} $ (see \autoref{prop:dotxj})
$$ 
\text{Res} \{\mathbf{W}; \l_N\}>0.  
$$ 
\end{proposition} 
\begin{proof} 
By our previous analysis it suffices to prove that $\lim_{t\rightarrow \infty} \mathbf{B}(t,\l_N)e^{-x_N(t)} >0$.  However, 
using  \autoref{prop:ABCasympt} we get 
$$ 
 \lim_{t\rightarrow \infty} \mathbf{B}(t,\l_N)e^{-x_N(t)}=\m_N(\infty) \prod_{i=1}^{N-1} (1-\l_N \langle \m_i, \m_i\rangle (\infty))=\m_N(\infty) \prod_{i=1}^{N-1} (1-\frac{\l_N}{\l_i})>0.  
 $$ 
\end{proof} 
This result generalizes to all residues, but a direct computation is not appealing.  Instead, we will prove the general case using the elementary deformation theory.  
\begin{theorem} \label{thm:respositivity}
For all initial conditions in $\mathcal{P} $ (see \autoref{prop:dotxj}) and all $1\leq j\leq N$ 
$$ 
\boxed{
\emph{Res}   \{\mathbf{W}; \l_j \}>0.} 
$$ 
\end{theorem} 
\begin{proof} 
The proof proceeds by induction on $N$.  The base case $N=1$ is covered in \autoref{ex:res12}.  
Assume the claim's validity for $N-1$ and let us set, for now, $t=0$.  Consider small initial masses $m_N(0)=n_N(0)= \epsilon>0$, which we take equal,  but the argument works if, instead, we use two small $\epsilon_1, \epsilon_2$.  
  
Both the residues and the eigenvalues are continuous functions of $\epsilon$ and for $\epsilon \rightarrow 0+$ the last transition matrix $S_N$ 
becomes the identity; we recover the case of $N-1$ sites.  We note that 
as $\epsilon \rightarrow 0+$, the largest eigenvalue $\l_1$ becomes arbitrarily large, while the term $\frac{\text{Res}\{\mathbf{W}; \l_1\}}{\l -\l_1}$ goes to $0$ since in that limit the Weyl function becomes the Weyl function for 
$N-1$ sites.  

Thus by the induction hypothesis the residues $\text{Res}   \{\mathbf{W}; \l_j \}>0$ are 
positive for small $m_N(0), n_N(0)$ and $2\leq j\leq N$.  We need to show that the same holds for $j=1$.  We note that by 
\autoref{prop:ABCasympt}, with $k=N$, 
$$ 
\mathbf{B}(t,\l)=e^{x_N(t)}\big( \m_N(t) \prod_{i=1}^{N-1} (1-\l \langle \m_i, \m_i \rangle (t)) +O(e^{-\alpha t})\big).
$$ 
Hence in the asymptotic region $t\rightarrow \infty$ the zeros of $\mathbf{B}(t,\l)$ are approximated by the zeros of 
$$ 
\prod_{i=1}^{N-1} (1-\l \langle \m_i, \m_i \rangle (t)).  
$$  
Since $\lim_{t\rightarrow \infty}\langle \m_1, \m_1\rangle (t)=\frac{1}{\l_1} $,  the first zero of $\mathbf{B}(t, \l)$ is approaching 
$\l_1$ as $t\rightarrow \infty$.  To determine the direction of approach, recall that $\langle \m_1,  \m_1 \rangle (t)=2m_1(t)n_1(t)$ and $m_1(t)n_1(t)$ is a strictly decreasing function of $t$.  
Hence 
$$ 
\frac{1}{\langle \m_1, \m_1 \rangle (t)} 
$$ 
is a strictly increasing function of $t$, which, in turn, implies that the first zero of $\mathbf{B}(t,\l)$ is approaching $\l_1$ from the left.  
Hence the first zero of $\mathbf{B}(t,\l)$ lies between $\l_2$ and $\l_1$, and since the remaining zeros are already interlacing, all zeros for small $m_1(0)$ and $n_1(0)$ are interlacing, implying the positivity of residues.  
We finally can relax the condition of small $m_1(0), n_1(0)$  since by \autoref{cor:Bsigns} the signs of $\mathbf{B}(\l_j, t)$ cannot change under continuous deformations.  

 \end{proof}

We will eventually prove that the residues of  the Weyl function $Z(\lambda)$ are also positive.  
First, however, we will establish a few 
simple symmetry properties of the transition matrix $S_k$ (see \eqref{eq:Sk}).  
\begin{lemma} \label{lem:sym_S}
The transition matrix $S_k$ satisfies
\begin{subequations}
\begin{align*}
&\det(S_k(\l))=1,\\
& K^{-1}S_k(\l)^TK = S_k(-\l), \quad S_k(\l)^{-1} =JS_k(\l)J,\\
& S_k(-\l)^T KJS_k(\l)(KJ)^{-1}=I_4,
\end{align*}
\end{subequations}
where 

\begin{equation*}
K=\begin{pmatrix}
0&0&1\\
0&2\sigma &0\\
1&0&0
\end{pmatrix},\qquad 
J=\begin{pmatrix}
1&0&0\\
0&-\mathbf{1}&0\\
0&0&1
\end{pmatrix},\qquad  \sigma =\begin{bmatrix} 0&1\\ 1&0 \end{bmatrix}.  
\end{equation*}

\end{lemma}
\begin{proof}
The proof is a straightforward computation relying on one elementary identity: $\mathbf{a}^*=\mathbf{a}^T \sigma$ if $\mathbf{a}\in M_{2,1}$, which in turn follows trivially from the definition of the bilinear form $\langle \cdot, \cdot \rangle$.   
\end{proof}
\begin{theorem}\label{thm:rel_WZ}
The Weyl functions satisfy
\begin{equation}\label{zw_res}
Z(\lambda)+Z(-\lambda)+
\langle \mathbf{W}(-\lambda), \mathbf{W}(\lambda)\rangle =0.
\end{equation}
\end{theorem}
\begin{proof}
Set $S(\l)=S_N(\l)S_{N-1}(\l)\cdots S_1(\l) $, then it follows from  \autoref{lem:sym_S} that 
$$S(-\l)^T KJS(\l)(KJ)^{-1}=I_4.$$
By calculating the $(1,4)$-entry of the left hand, we get 
$$
A(-\l)C(\l)+A(\l)C(-\l)-2\langle \mathbf{B}(-\lambda), \mathbf{B}(\lambda)\rangle =0
$$
which immediately leads to \eqref{zw_res}.

\end{proof} 
 The identity \eqref{zw_res} readily implies the counterpart of \autoref{thm:respositivity}.

\begin{theorem} \label{thm:respositivity_c}
For all initial conditions in $\mathcal{P} $ (see \autoref{prop:dotxj}) and all $1\leq j\leq N$ 
$$ 
\boxed{
\emph{Res}   \{Z; \l_j \}>0.} 
$$ 
\end{theorem} 
We conclude this section with a brief summary of relevant properties of the Weyl functions $\mathbf{W}(\l)$ and $Z(\l)$.  
\begin{proposition} \label{prop:WZ-spectral-rep}
The Weyl functions admit the Stieltjes integral representations:
\begin{equation}\label{exp:weyl_int}
\begin{aligned}
&\mathbf{W}(\lambda)=\sum_{k=1}^N\frac{\mathbf{b}_{k}}{\lambda-\lambda_k}=\int \frac{d{\bm\mu}(x)}{\l-x},\qquad  &\mathbf{b}_{k}>0,\\
&Z(\lambda)=\sum_{k=1}^N\frac{c_{k}}{\lambda-\lambda_k}=\iint\frac{\langle d\bm{\mu}(x),  d\bm{\mu}(y) \rangle }{(\l-x)(x+y)}, \ \ \qquad &c_k>0,
\end{aligned}
\end{equation}
where 
\begin{equation} \label{rel_cb}
d\bm{\mu}(x)=\sum_{j=1}^N\mathbf{b}_j\delta(x-\l_j)dx, \qquad 
c_k=\sum_{p=1}^N\frac{\langle \mathbf{b}_{k},  \mathbf{b}_{p}\rangle }{\l_k+\l_p}.  
\end{equation}
\end{proposition} 

\subsection{Approximation problems}\label{sec:ApproxP}
Let $1\leq k\leq N$ and let us set 
\begin{equation*}
S_{[N,k]}(\l):=S_N(\l)S_{N-1}(\l)\cdots S_{N-k+1}(\l).
\end{equation*}
Then by \eqref{eq:transition} we have 
\begin{align}\label{eq:S_N-k}
(A(\l), \mathbf{B}^T(\l) ,C(\l))^T=
S_{[N,k]}(\l)
(A_{N-k}(\l) ,\mathbf{B}_{N-k}^T(\l) ,C_{N-k}(\l))^T.  
\end{align}

 It is elementary to get by induction the following properties for $S_{[N,k]}(\l)$.
\begin{lemma}\label{lem:S-N-k_deg}
The degrees of the entries of $S_{[N,k]}(\l)$ are as follows:
\begin{equation*}
\deg(S_{[N,k]}(\l))=
\begin{pmatrix}
k&k&k+1\\
k-1&k-1&k\\
k-1&k-1&k
\end{pmatrix}.
\end{equation*}
Moreover, 
\begin{equation} \label{eq:SNk0} 
S_{[N,k]}[0]=\begin{bmatrix} 1& 0&0\\
 \sum\displaystyle_{j=N-k+1}^N\m_j e^{x_j}& \mathbf{1}&0\\
\times &2 \sum\displaystyle_{j=N-k+1}^N\ms_j e^{x_j}& 1 \end{bmatrix} ,  
 \end{equation} 
 \begin{equation} \label{eq:SNk1} 
 S_{[N,k]}[1]=\begin{bmatrix} \times& -2\sum_{j=N-k+1}^N\ms_je^{-x_j}&\times\\
\times& \times &\times\\
\times&\times& \times \end{bmatrix} ,  
 \end{equation}
where the notation $q[j]$ denotes the coefficient of $\l^j$ of a polynomial $q(\lambda)$. 
%\left(\sum\displaystyle_{j=N-k+1}^NM_j e^{x_j}\right)\left(\sum\displaystyle_{j=N-k+1}^NM_j^\sigma e^{x_j}\right)

\end{lemma}

Let
\begin{equation*}
S_{[N,k]}(\l)=
\begin{pmatrix}
s^{(k)}_{1,1}&\mathbf{s}^{*(k)}_{1,2}&s^{(k)}_{1,3}\\
\mathbf{s}^{(k)}_{2,1}&\mathbf{s}^{(k)}_{2,2}&\mathbf{s}^{(k)}_{2,3}\\
s^{(k)}_{3,1}&\mathbf{s}^{*(k)}_{3,2}&s^{(k)}_{3,3}
\end{pmatrix}.  
\end{equation*}
Then  \autoref{lem:S-N-k_deg} implies the following formulae. 
\begin{theorem}\label{thm:s_xm}
 For any integer $1\leq k\leq N$, we have
\begin{subequations}
\begin{align*}
-2\ms_{N-k+1}e^{-x_{N-k+1}}&=\mathbf{s}_{1,2}^{*(k)}[1]-\mathbf{s}_{1,2}^{*(k-1)}[1],\\
2\ms_{N-k+1}e^{x_{N-k+1}}&=\mathbf{s}_{3,2}^{*(k)}[0]-\mathbf{s}_{3,2}^{*(k-1)}[0],
\end{align*}
\end{subequations}
with the convention $\mathbf{s}_{1,2}^{*(0)}[1]=\mathbf{s}_{3,2}^{*(0)}[0]=0$. 
\end{theorem}

Hereafter, the upper index $k$ is often omitted for simplicity; for example, $s_{i,j}$ is short for $s_{i,j}^{(k)}$. By induction, one can determine the degrees of various quadratic expressions in the entries of $S_{[N,k]}$.   The origin of these quadratic expressions is in the matrix elements of the second exterior power of $S_{[N,k]}$, to be denoted $\Lambda^2 S_{[N,k]}$. We will use these quadratic expressions to set up the approximation problem needed for the inverse problem. We follow the recent work of one of us \cite{chang:2022:nv-pfaffians}.  
\begin{remark} The difference in setting up the inverse problem in \cite{hone-lundmark-szmigielski:novikov} and \cite{chang:2022:nv-pfaffians} is that the former uses the first and the third column of the transition matrix. In contrast, the latter uses only the second column. The approximation problem using the second column seems more efficient from the point of view of generalizations to multi-component scenarios.  
\end{remark} 
We first motivate the appearance of the second exterior power of $S_{[N,k]}$. We recall the labelling of the basis of $\C^4$ used earlier, namely, $\{e_1, e_{21}, e_{22}, e_3\}$.

Let us consider the following 
identity

\begin{equation} \label{eq:wedgeI}
\begin{bmatrix} A\\ \mathbf{B}\\ C \end{bmatrix}\wedge S_{[N,k]} e_{2b}=\Lambda^2 S_{[N,k]}\begin{bmatrix} A_{N-k} \\ \mathbf{B}_{N-k}\\ C_{N-k} \end{bmatrix} \wedge e_{2b}, 
\end{equation}
and investigate some of its consequences.  

\begin{proposition} Let us fix the indices $a,b \in \{1,2\} $.   The entries of  $S_{[N,k]}$ satisfy the following quadratic identities: 
\begin{align}
&A(S_{[N,k]})_{2a,2b}-B_{a} (S_{[N,k]})_{1,2b}=\notag\\
&(\Lambda^2 S_{[N,k]})_{1,2a; 1,2b} A_{N-k} +\sum_{c} (\Lambda^2 S_{[N,k]})_{1,2a;2c,2b}B_{N-k, c}-
(\Lambda^2 S_{[N,k]})_{1,2a;2b,3}C_{N-k}, 
\end{align}  
\begin{align}
&A(S_{[N,k]})_{3,2b}-C(S_{[N,k]})_{1,2b}=\notag \\
&(\Lambda^2 S_{[N,k]})_{1,3; 1,2b} A_{N-k} +\sum_{c} (\Lambda^2 S_{[N,k]})_{1,3;2c,2b}B_{N-k, c}-
(\Lambda^2 S_{[N,k]})_{1,3;2b,3}C_{N-k}. 
\end{align}

\end{proposition} 
\begin{proof} 
The two identities follow from equation \eqref{eq:wedgeI} by evaluating the coefficients of $e_1\wedge e_{2a}$ and 
$e_1\wedge e_3$.  
\end{proof} 
Using $e_{2a}\otimes f_{2b}$, where $f_{2b}(e_{2a})=\delta_{a,b}$, 
as a basis of $\C^2\otimes \C^{*2}$ in the first expression, and $f_{2b}$ for the second, respectively, we can rewrite these identities as 
\begin{equation}\label{eq:firstid}
\begin{gathered} 
A \mathbf{s}^{(k)}_{2,2}-\mathbf{B} \otimes \mathbf{s}^{*(k)}_{1,2}=\\
\overbrace{\big(\sum_{a,b} (\Lambda^2 S_{[N,k]})_{1,2a; 1,2b} e_{2a}\otimes f_{2b}\big)}^{T_{2,2}} A_{N-k} +\overbrace{\sum_{c}\big(\sum_{a,b} (\Lambda^2 S_{[N,k]})_{1,2a;2c,2b}e_{2a}\otimes f_{2b}\big)B_{N-k, c}}^{T_{2,4}\otimes \hat {\mathbf{B}}_{N-k} }-\\
\overbrace{\sum_{a,b} \big((\Lambda^2 S_{[N,k]})_{1,2a;2b,3} e_{2a}\otimes f_{2b} \big)}^{T_{2, 3}} C_{N-k}, 
\end{gathered}
\end{equation} 
and, 
\begin{equation}\label{eq:secondid} 
\begin{gathered} 
A\mathbf{s}^{*(k)}_{3,2}-C\mathbf{s}^{*(k)}_{1,2}=\\
\overbrace{\big(\sum_{b} (\Lambda^2 S_{[N,k]})_{1,3; 1,2b}f_{2b} \big)}^{T_{1,2}} A_{N-k} +\overbrace{\sum_{c} \big( \sum_b (\Lambda^2 S_{[N,k]})_{1,3;2c,2b}f_{2b}\big)B_{N-k, c}}^{T_{1,4} \hat{\mathbf{B}}_{N-k}}-\\
\overbrace{\big(\sum_{2b} (\Lambda^2 S_{[N,k]})_{1,3;2b,3}f_{2b} \big)}^{T_{1,3}} C_{N-k}, 
\end{gathered} 
\end{equation} 
where $\hat{\mathbf{B}}_{N-k}=(-B_{N-k, 2}, B_{N-k, 1})$, and the labelling $T_{i,j}$   refers to the block decomposition of $\Lambda^2S_{[N-k]}$ discussed  further in Appendix \autoref{app2}.  
The next proposition is also proved in Appendix \autoref{app2}.
\begin{proposition} \label{prop:extdegree} 
The polynomial order of blocks $T_{i,j}$  of $\Lambda^2 S_{[N-k]}$ is given by the following matrix of degrees:   
\begin{equation} 
\deg(\Lambda^2 S_{N-k})=\begin{pmatrix} k-1&k-1&k& k-1\\
k&k&k+1&k\\
k-1&k-1&k&k-1\\
k-1&k-1&k&k-1 \end{pmatrix} 
\end{equation} 
In particular,  for $ T_{2,2}, T_{2,3}, T_{2,4}$ appearing in \eqref{eq:firstid}, the degrees are: 
$$ 
\deg(T_{2,2})=k, \, \deg(T_{2,3})=k+1,\,  \deg(T_{2,4})=k, 
$$
while for $T_{1,2}, T_{1,3}, T_{1,4}$ in \eqref{eq:secondid} the degrees are: 
$$\deg(T_{1,2})=k-1, \, \deg(T_{1,3})=k,\,  \deg(T_{1,4})=k-1,   
$$
respectively.  
\end{proposition} 
\begin{proof} 
The proof has been relegated to Appendix \autoref{app2} (see \autoref{prop:extpowerdeg} and \autoref{cor:degT}).  
\end{proof} 

We are now ready to state an approximation problem, central to the inverse problem discussed in the next section.  
\begin{theorem} \label{th:S-N-K_appr}
As $\l \rightarrow \infty$, the Weyl functions $\mathbf{W}(\l)$ and $Z(\l)$ satisfy an asymptotic expansion
\begin{align}
\begin{pmatrix}
\mathbf{W}(\l)\\
2Z(\l)
\end{pmatrix}
\otimes\mathbf{s}^{*(k)} _{1,2}(\l)+
\begin{pmatrix}
\mathbf{s}^{(k)}_{2,2}(\l) \\
\mathbf{s}^{*(k)}_{3,2}(\l)
\end{pmatrix}=
\begin{pmatrix}
O(1)\\
O(\l^{-1})
\end{pmatrix}, 
\end{align}
and
\begin{equation}\label{sym_approx}
\mathbf{s}^{*(k)} _{3,2}(-\l)+2\mathbf{W}^*(\l) \mathbf{s}^{(k)} _{2,2}(-\l)-2Z(\l)\mathbf{s}^{*(k)} _{1,2}(-\l)=O(\l^{-(k+1)}).
\end{equation}
\end{theorem}
\begin{proof}
First, recall that $\deg(A)=N, \deg(A_{N-k})=N-k, \, \deg(\mathbf{B}_{N-k})=\deg(C_{N-k})=N-k-1$.  Then the first two approximation formulas follow from \eqref{eq:firstid}, \eqref{eq:secondid}, and \autoref{prop:extdegree}.  

Finally, let us turn to the third approximation equation. By using  \eqref{eq:S_N-k} and the identities in Lemma \ref{lem:sym_S}, it is not hard to observe that 
 \begin{equation*}
 (A(\l),\mathbf{B}^T(\l),C(\l) )KJS_{[N,k]}(-\l)=(A_{N-k}(\l),\mathbf{B}_{N-k}^T(\l),C_{N-k}(\l) )KJ, 
 \end{equation*}
 which, after carrying out the multiplication by $KJ$, gives
 \begin{equation*} 
 (C(\l) ,-2\mathbf{B}^*(\l), A(\l) ) S_{[N,k]}(-\l)=(C_{N-k}(\l), -2\mathbf{B}^*_{N-k}(\l) , A_{N-k}(\l)).  
 \end{equation*} 
  
The second block component of this equality implies 
 $$
\mathbf{s}^{*{(k)}} _{3,2}(-\l)-2\frac{\mathbf{B}^*(\l)}{A(\l)}\mathbf{s} _{2,2}(-\l)+\frac{C_N(\l)}{A(\l)}\mathbf{s}^{*(k)} _{1,2}(-\l)=-2\frac{\mathbf{B}^*_{N-k}(\l)}{A(\l)}=O(\l^{-k-1}),
 $$
 which completes the proof.
\end{proof}
\section{Inverse Spectral Problem}\label{sec:inverse}
In this section, we consider the inverse problem: given the spectral data consisting of $\{\lambda_j,\mathbf{b}_j\}_{j=1}^N$ (or equivalently, given a rational function $\mathbf{W}(\l)$ with integral representation \eqref{exp:weyl_int}), recover the positions and masses $\{x_j,\mathbf{m} _j\}_{j=1}^N$ for which $\mathbf{W}(\l)$ and $Z(\l)$ are the Weyl functions of the boundary value problem \eqref{eq:psi123}-\eqref{eq:psi3psi1}.  

To solve this inverse problem, we introduce Hermite--Pad\'{e} approximations hinted at  in \autoref{th:S-N-K_appr}.  

\begin{definition}\label{def:HP}
For a fixed integer $1\leq k\leq N$,
given 
\begin{equation*}
\mathbf{W}(\l)=\int \frac{d\bm{\mu}(x)}{\l-x},\qquad  Z(\l)=\iint\frac{\langle d\bm{\mu}(x),  d\bm{\mu}(y)\rangle}{(\l-x)(x+y)}
\end{equation*}
with a positive $2\times1$ vector discrete Stieltjes measure $d\bm{\mu}$ on $\mathbb{R}_+$, i.e., 
$$
d\bm{\mu}(x)=\sum_{j=1}^N\mathbf{b}_j\delta(x-\l_j)dx, \qquad \mathbf{b}_j>0, \ \l_1>\l_2>\cdots>\l_N>0,
$$
we seek polynomials $\mathbf{Q}^*_k(\l)_{1\times2},\mathbf{P}_k(\l)_{2\times2},\mathbf{\hat P}^*_k(\l)_{1\times2}$ with degrees $\deg(\mathbf{Q}^*_k)=k$, $\deg(\mathbf{P}_k)=\deg(\mathbf{\hat P}^*_k)=k-1$, respectively, such that, as $\l\rightarrow\infty$, 
\begin{subequations}\label{eq:HP}
\begin{align}
&\mathbf{W}(\l)\otimes\mathbf{Q}^*_k(\l)-\mathbf{P}_k(\l)=O(1),\label{HP_1}\\
&\mathbf{Q}^*_k(\l)Z(\l)-\mathbf{\hat P}^*_k(\l)=O\left(\frac{1}{\l}\right),\label{HP_2}\\
&\mathbf{\hat P}^*_k(\l)+\mathbf{W}^*(-\l)\mathbf{P}_k(\l)+Z(-\l)\mathbf{Q}^*_k(\l)=O\left(\frac{1}{\l^{k+1}}\right).\label{HP_3}
\end{align}
In addition, we require that $\mathbf{Q}^*_k(0)=\mathbf{0}_{1\times2},\mathbf{P}_k(0)=\mathbf{1}$.
\end{subequations}
\end{definition}
We introduce the following notations to formulate the solution to the above approximation problem.
\begin{definition}\label{def_Ibeta}
Given the $2\times1$ vector measure
$$
d\bm{\mu}(x)=\sum_{j=1}^N\mathbf{b}_j\delta(x-\l_j)dx, \qquad \mathbf{b}_j>0, \ \l_1>\l_2>\cdots>\l_N>0,
$$
define the moments
\begin{equation*}
 I_{i,j}=\iint \frac{x^iy^j}{x+y}\langle d \bm{\mu}(x), d\bm{\mu}(y)\rangle, \qquad \bm{\beta}_j=\int x^jd\bm{\mu}(x),
\end{equation*}
and the determinants
\begin{equation}
F_k^{(i,j)}=\det(I_{i+p,j+q})_{p,q=0}^{k-1},
\end{equation} 
with the convention that 
\begin{equation}
F_0^{(i,j)}=1,\qquad F_k^{(i,j)}=0,\quad k<0.
\end{equation}
\end{definition} 
We note that $I_{i,j}=\sum_{p=1}^N \sum_{q=1}^N \l_q^i \frac{\langle \mathbf{b}_q, \mathbf{b}_p \rangle }{\l_q+\l_p} \l_p^j $ can be written in a  matrix form 
\begin{equation*} 
I_{i,j} =\begin{pmatrix} \l_N^i&\l_{N-1} ^i&\dots&\l_1^i \end{pmatrix} \Lambda_1(N) \begin{pmatrix} \l_N^j\\\l_{N-1}^j\\ \vdots\\ \l_1^j \end{pmatrix},   
\end{equation*} 
where 
\begin{equation}\label{eq:Lambda1N}
\Lambda_1(N)=\begin{pmatrix}\frac{\langle \mathbf{b}_N, \mathbf{b}_N \rangle }{\l_N+\l_N}&\frac{\langle \mathbf{b}_N, \mathbf{b}_{N-1} \rangle }{\l_N+\l_{N-1}}&\cdots&\frac{\langle \mathbf{b}_N, \mathbf{b}_1 \rangle }{\l_N+\l_1}\\
\frac{\langle \mathbf{b}_{N-1}, \mathbf{b}_N \rangle }{\l_{N-1}+\l_N}&\frac{\langle \mathbf{b}_{N-1}, \mathbf{b}_{N-1} \rangle }{\l_{N-1}+\l_{N-1}}&\cdots&\frac{\langle \mathbf{b}_{N-1}, \mathbf{b}_1 \rangle }{\l_{N-1}+\l_1}\\
\vdots&\vdots&\ddots&\vdots\\
\frac{\langle \mathbf{b}_1, \mathbf{b}_N\rangle }{\l_1+\l_N}&\frac{\langle \mathbf{b}_1, \mathbf{b}_{N-1} \rangle }{\l_1+\l_{N-1}}&\cdots&\frac{\langle \mathbf{b}_1, \mathbf{b}_1 \rangle }{\l_1+\l_1}  
\end{pmatrix}.  
\end{equation} 
This result has the following simple consequence. 
\begin{proposition} \label{prop:Fk} 

\begin{equation} \label{eq:Fk}
F_k^{(i,j)}=\det \left(
\begin{pmatrix}
\l_N^i&\l_{N-1}^i&\cdots&\l_1^i\\
\vdots&\vdots&\ddots&\vdots\\
\l_N^{i+k-1}&\l_{N-1}^{i+k-1}&\cdots&\l_1^{i+k-1}\\
\end{pmatrix}
\Lambda_1(N) 
\begin{pmatrix}
\l_N^j&\cdots&\l_N^{j+k-1}\\
\l_{N-1}^j&\cdots&\l_{N-1}^{j+k-1}\\
\vdots&\ddots&\vdots\\
\l_1^j&\cdots&\l_1^{j+k-1}
\end{pmatrix}
\right). 
\end{equation}

\end{proposition}

We will also use the \emph{vector-valued determinants} to be defined below.  
\begin{definition} \label{def:vdet}
Given a $k\times k$ matrix $A=[A_{i,j}]$,  
\begin{equation} 
A=\begin{pmatrix} 
\mathbf{a}_{1,1}&a_{1,2}&\hdots &a_{1,k}\\
\mathbf{a}_{2,1}&a_{2,1}&\hdots &a_{2,k}\\
\vdots&\vdots&\vdots&\vdots\\
\mathbf{a}_{k,1}&a_{k,2}&\hdots&a_{k,k}, 
\end{pmatrix}
\end{equation}
where $\mathbf{a}_{j,1}\in \mathbf{C}^n, a_{j,i} \in \mathbf{C}, \quad 1\leq j\leq k, 2\leq i\leq k, $ the vector-valued determinant  
is defined in the usual way as 
\begin{equation} 
\det A=\sum_{\sigma \in S_k} \sgn(\sigma) A_{1,\sigma(1)} A_{2,\sigma(2)}\cdots A_{k,\sigma(k)},
\end{equation} 
where $S_k$ stands for the symmetric group on $k$ elements.
\end{definition} 

\begin{remark} 
Since the entries involving vectors occur exactly once in every product defining the determinant, we can map the determinant using the bilinear form $\langle, \rangle$ to take values in the dual of the vector space.  Moreover, 
the column with vectors can be in any position.  We note that all basic alternating properties of the scalar determinant carry over to the vector-valued determinant, and so do any  of Laplace's expansions.  Our definition of the vector-valued determinant is clearly a straightforward generalization of the vector product formula in $\mathbf{R}^3$: given 
$\mathbf{a}, \mathbf{b} \in \R^3$ we can define $\mathbf{a}\times \mathbf{b}=\det \begin{pmatrix}\mathbf{i} &a_1&b_1\\\mathbf{j}&a_2&b_2\\\mathbf{k}&a_3&b_3 \end{pmatrix} $, where $\mathbf{i}, \mathbf{j}, \mathbf{k}$ is the standard basis of $ 
\R^3$.  
 \end{remark}

\begin{definition} \label{def:G}
Let $k\in \mathbf{Z}$.  
Define
\begin{align*} 
&\mathbf{G}_{k}^{(i,j)}\vcentcolon=\det 
\begin{pmatrix}
\Big[ \bm\beta_{i+p-1}, & \big[I_{i+p,j+q}\big]\Big] \end{pmatrix}_{p=0,\ldots k-1}^{q=0,\ldots k-2}, 
\end{align*}
with the convention $\mathbf{G}_{1}^{(i,j)}=\bm \beta_{i-1}$ and $\mathbf{G}_{k}^{(i,j)}=0$ if $k\leq0$.  

\end{definition} 
Similar to the case of $I_{i, j} $ we have the matrix interpretation of $\bm \beta_{i+p-1}$, namely, 
\begin{equation*} 
\bm \beta_{i-1}=\begin{pmatrix} \l_N^i&\l_{N-1}^i&\dots&\l_1^i \end{pmatrix} \begin{pmatrix} \frac{\mathbf{b}_N}{\l_N}\\\frac{\mathbf{b}_{N-1}}{\l_{N-1}}\\\vdots\\ \frac{\mathbf{b}_1}{\l_1} \end{pmatrix}.  
\end{equation*} 
This results in the description of $\mathbf{G}_{k}^{(i,j)}$ similar to the description of $F_k^{(i,j)}$ in \autoref{prop:Fk}.  
\begin{proposition} \label{prop:Gk} 
\begin{equation} 
\mathbf{G}_{k}^{(i,j)}=\det \left(\begin{pmatrix} \l_N^i&\l_{N-1}^i&\cdots&\l_1^i\\
\vdots&\vdots&\ddots&\vdots\\
\l_N^{i+k-1}&\l_{N-1}^{i+k-1}&\cdots&\l_1^{i+k-1}\\
\end{pmatrix}
\Lambda_2(N) 
\begin{pmatrix}
1&0&\dots&0\\
0&\l_N^j&\cdots&\l_N^{j+k-1}\\
0&\l_{N-1}^j&\cdots&\l_{N-1}^{j+k-1}\\
\vdots&\vdots&\ddots&\vdots\\
0&\l_1^j&\cdots&\l_1^{j+k-1}
\end{pmatrix}
\right), 
\end{equation}
where 
\begin{equation}\label{eq:Lambda2N}
\Lambda_2(N)=\begin{pmatrix}\frac{\mathbf{b}_N}{\l_N}&\frac{\langle \mathbf{b}_N, \mathbf{b}_N \rangle}{\l_N+\l_N}&
\frac{\langle \mathbf{b}_N,\mathbf{b}_{N-1} \rangle}{\l_N+\l_{N-1}}&\cdots&\frac{\langle \mathbf{b}_N, \mathbf{b}_1 \rangle }{\l_N+\l_1}\\
\frac{\mathbf{b}_{N-1}}{\l_{N-1}}&\frac{\langle \mathbf{b}_{N-1}, \mathbf{b}_N \rangle }{\l_{N-1}+\l_N}&\frac{\langle \mathbf{b}_{N-1}, \mathbf{b}_{N-1} \rangle }{\l_{N-1}+\l_{N-1}}&\cdots&\frac{\langle \mathbf{b}_{N-1}, \mathbf{b}_1 \rangle }{\l_{N-1}+\l_1}\\
\vdots&\vdots&\vdots&\ddots&\vdots\\
\frac{\mathbf{b}_1}{\l_1}&\frac{\langle \mathbf{b}_1, \mathbf{b}_N\rangle }{\l_1+\l_N}&\frac{\langle \mathbf{b}_1, \mathbf{b}_{N-1} \rangle }{\l_1+\l_{N-1}}&\cdots&\frac{\langle \mathbf{b}_1, \mathbf{b}_1 \rangle }{\l_1+\l_1}  
\end{pmatrix}.  
\end{equation} 
\end{proposition}

\begin{theorem}\label{th:sol_HP}
Suppose $F_k^{(1,0)}\neq0$ and let $(\mathbf{Q}_k^*)^*(\l):=\mathbf{Q}_k(\l)=\l \mathbf{q}_k(\l)$.  Then the Hermite--Pad\'{e} approximation problem \ref{def:HP} has a unique solution $\mathbf{Q}_k^*$ or, equivalently, $\mathbf{Q}_k$, with $\mathbf{q}_k$ given by:

\begin{equation}\label{HP_Q}
\mathbf{q}_k(\l)=\frac{-1}{F_k^{(1,0)}}
\det \left(
\begin{array}{ccccc}
0&1&\l&\cdots &\l^{k-1} \\
\bm\beta_0&I_{1,0}&I_{1,1}&\cdots & I_{1,k-1}\\
\bm\beta_1&I_{2,0}&I_{2,1}&\cdots&I_{2,k-1}\\
\vdots&\vdots&\vdots&\ddots&\vdots\\
\bm\beta_{k-1}&I_{k,0}&I_{k,1}&\cdots&I_{k,k-1}
\end{array}
\right),
\end{equation}
where $\bm\beta_j, I_{i,j}$ and the vector-valued determinant are defined in \autoref{def_Ibeta} and \autoref{def:vdet}, respectively.  
Moreover, 
\begin{align}
&\mathbf{P}_k(\l)=\int \frac{d\bm\mu(x)\otimes (\mathbf{Q}^*_k(\l)-\mathbf{Q}^*_k(x))}{\l-x} -\int \frac{d\bm\mu(x)\otimes\mathbf{Q}^*_k(x)}{x}+\mathbf{1},\label{HP_P}\\
&\mathbf{\hat P}^*_k(\l)=\iint\frac{\mathbf{Q}^*_k(\l)-\mathbf{Q}^*_k(x)}{(\l-x)(x+y)}\langle d\bm\mu(x), d\bm\mu(y)\rangle.  \label{HP_hatP}
\end{align}
  \end{theorem}

\begin{proof}
First, let us introduce some auxiliary notations.  
If $f(\l)$ is a formal series $\sum_{j=-\infty } a_j \l^j$ then we denote the truncations at an index $k$ 
by $(f(\l))_{\geq k}=\sum_{j\geq k} a_j \l^j$.  Likewise, $(f(\l))_{<k}=\sum_{j<k} a_j \l^j$, and other cases are defined in a similar way.  All functions of interest to us have a Laurent series around $\l=\infty$.  With this notation in place, we have 

\begin{subequations}\label{eq:PPhat}
\begin{align}
&\mathbf{P}_k(\l)=(\mathbf{W}(\l) \otimes \mathbf{Q}^*_k(\l))_{\geq 1}+ \mathbf{1}, \\
&\mathbf{\hat P}_k^*(\l)=(Z(\l) \mathbf{Q}^*_k(\l))_{\geq 0}.   
\end{align}
\end{subequations}

One easily obtains now from \autoref{def:HP} the formulas for the projections needed in the remainder of the proof: 
\begin{subequations}\label{eq:projections} 
\begin{align} 
&(\mathbf{W}(\l) \otimes \mathbf{Q}^*_k(\l))_{\geq 1}=
\int \frac{d\bm\mu(x)\otimes (\mathbf{Q}^*_k(\l)-\mathbf{Q}^*_k(x))}{\l-x} -\int \frac{d\bm\mu(x)\otimes\mathbf{Q}^*_k(x)}{x}, \\
&(\mathbf{W}(\l) \otimes \mathbf{Q}^*_k(\l))_{\leq 0}=\l \int\frac{d\bm \mu(x)\otimes \mathbf{Q}^*_k(x)}{x(\l-x)}, \\
&(Z(\l) \mathbf{Q}^*_k(\l))_{\geq 0}=\iint\frac{\mathbf{Q}^*_k(\l)-\mathbf{Q}^*_k(x)}{(\l-x)(x+y)}\langle d\bm\mu(x), d\bm\mu(y)\rangle, \\
&(Z(\l) \mathbf{Q}^*_k(\l))_{< 0}=\iint\frac{\mathbf{Q}^*_k(x)}{(\l-x)(x+y)}\langle d\bm\mu(x), d\bm\mu(y)\rangle.  
\end{align} 
\end{subequations} 
Thus  equations \eqref{eq:PPhat}  and \eqref{eq:projections} imply \eqref{HP_P} and \eqref{HP_hatP}.  

The next step is to determine $\mathbf{Q}^*_k(\l)$ by analyzing the approximation problem \eqref{HP_3}.  To this end, using \eqref{eq:projections}, we 
rewrite \eqref{HP_3} as: 
\begin{equation*} 
(Z(\l) \mathbf{Q}^*_k(\l))_{\geq 0} +\mathbf{W}^*(-\l)\big((\mathbf{W}(\l) \otimes \mathbf{Q}^*_k(\l))_{\geq 1}+\mathbf{1}\big)+
Z(-\l) \mathbf{Q}^*_k(\l)=O(\frac{1}{\l^{k+1}}), 
\end{equation*} 
which when subtracted from \eqref{zw_res} (multiplied by $\mathbf{Q}^*_k(\l)$) yields
\begin{equation*} 
(Z(\l) \mathbf{Q}^*_k(\l))_{<0}+\mathbf{W}^*(-\l)(\mathbf{W}(\l)\otimes \mathbf{Q}^*_k(\l))_{\leq 0}-\mathbf{W}^*(-\l)=O(\frac{1}{\l^{k+1}}).  
\end{equation*} 

Using \eqref{eq:projections}, the latter expression can be simplified to
\begin{equation*} 
\iint\frac{y\mathbf{Q}^*_k(x)}{x(\l+y)(x+y)}\langle d\bm{\mu}(x), d\bm{\mu}(y)\rangle -\int\frac{1}{\l+x}d\bm{\mu}^*(x)=O(\frac{1}{\l^{k+1}}),
\end{equation*} 
or, equivalently,
\[
\iint\frac{y\mathbf{Q}^*_k(x)}{x(x+y)}\sum_{j=0}^\infty\frac{(-y)^j}{\l^{j+1}}\langle d\bm{\mu}(x),  d\bm{\mu}(y)\rangle-\int\sum_{j=0}^\infty\frac{(-x)^{j}}{\l^{j+1}}d\bm{\mu}^*(x)=O(\frac{1}{\l^{k+1}}),
\]
thus implying 
\begin{equation}\label{orth_Q}
\iint\frac{y^{j+1}\mathbf{Q}^*_k(x)}{x(x+y)}\langle d\bm{\mu}(x), d\bm{\mu}(y)\rangle -\int x^jd\bm{\mu}^*(x)=0, \quad j=0,1,\ldots,k-1.
\end{equation}
Since $\mathbf{Q}_k(\l)\vcentcolon=\l \mathbf{q}_k(\l)$ we can rewrite the latter equation 
as 
\begin{equation}\label{eq:orth_q}
\iint\frac{y^{j+1}\mathbf{q}^*_k(x)}{(x+y)}\langle d\bm{\mu}(x), d\bm{\mu}(y)\rangle -\int x^jd\bm{\mu}^*(x)=0, \quad j=0,1,\ldots,k-1.  
\end{equation}

If we now take the dual of  this relation and set $\mathbf{q}_k(\l)=\sum_{j=0}^{k-1} \mathbf{c}_j \l^j$, then the above equality implies  the linear system
\begin{equation*}
\Big [\left(
\begin{array}{cccc}
I_{1,0}&I_{1,1}&\cdots & I_{1,k-1}\\
I_{2,0}&I_{2,1}&\cdots&I_{2,k-1}\\
\vdots&\vdots&\ddots&\vdots\\
I_{k,0}&I_{k,1}&\cdots&I_{k,k-1}
\end{array}
\right)\otimes \mathbf{1}\Big]
\left(
\begin{array}{c}
\mathbf{c}_0\\
\mathbf{c}_1\\
\vdots\\
\mathbf{c}_{k-1} 
\end{array}
\right)
=
\left(
\begin{array}{c}
\bm{\beta}_0\\
\bm{\beta}_1\\
\vdots\\
\bm{\beta}_{k-1}
\end{array}
\right), 
\end{equation*}
whose solution reads 
\begin{equation} 
\begin{pmatrix} \mathbf{c}_0\\ \mathbf{c}_1\\\vdots\\\mathbf{c}_{k-1} 
\end{pmatrix} =\Big[\begin{pmatrix} I_{1,0}&I_{1,1}&\hdots&I_{1,k-1}\\
                                                 \vdots&\vdots&\ddots&\vdots\\
                                                 I_{k,0}&I_{k,1}&\hdots&I_{k,k-1} \end{pmatrix}^{-1} \otimes \mathbf{1}\Big] \begin{pmatrix} \bm{\beta}_0\\ \bm{\beta}_1\\\vdots\\\bm{\beta}_{k-1} 
\end{pmatrix} .  
\end{equation} 
Using the canonical basis $\{e_j, j=1,\dots, k\}$ of $\mathbf{C}^k$ we can write in components 
the last equation as 
\begin{equation} 
\begin{pmatrix} \mathbf{c}_0\\ \mathbf{c}_1\\\vdots\\\mathbf{c}_{k-1} 
\end{pmatrix} =\sum_{j=1}^k \Big(\begin{pmatrix} I_{1,0}&I_{1,1}&\hdots&I_{1,k-1}\\
                                                 \vdots&\vdots&\ddots&\vdots\\
                                                 I_{k,0}&I_{k,1}&\hdots&I_{k,k-1} \end{pmatrix}^{-1} e_j\Big) \otimes \bm \beta_{j-1},   
\end{equation} 
which shows that the solution can be written in terms of Cramer's Rule with the determinant being replaced by the vector-valued determinant, for example, 

$$ 
\bm c_0= \frac{ \det \begin{pmatrix} \bm \beta_0 &I_{1,1}&\hdots&I_{1,k-1}\\
\bm \beta_1&I_{2,1}&\hdots&I_{2,k-1}\\
\vdots&\vdots&\vdots&\vdots\\
\bm \beta_{k-1}&I_{k,1}&\hdots&I_{k, k-1} \end{pmatrix}}{F_k^{(1,0)}},  
$$ 
and so on.  
Therefore by Cramer's Rule, $\mathbf{q}_k(\l)$ is uniquely given by the formula \eqref{HP_Q}.   
\end{proof}

\begin{definition} 
Consider all $3N$-tuples of numbers $\{\l_k,\mathbf{b}_k\}_{k=1}^N$, satisfying
$$ \l_1>\l_2>\cdots>\l_N>0, \qquad \mathbf{b}_j>0.  $$
We will call this set the \emph{extended spectral data}, and we will denote it ${\cal S_{\text{ext}}}(N)$.  
The actual spectral data of the boundary value problem \eqref{eq:psi123}-\eqref{eq:psi3psi1} with $\mathbf{m}_k>0$ and $x_1<x_2<\cdots <x_N$ will be denoted 
${\cal S}(N)$.  
\end{definition}

\begin{remark} 
Clearly, ${\cal S}(N) \subset {\cal S_{\text{ext}}}(N)$.  However, it  is a challenging question to describe precisely how ${\cal S}(N)$ is contained in ${\cal S}_{\text{ext}}(N)$. 
In the last section of the paper, we will present some evidence that ${\cal S}(N)$ is a \emph{determinantal subvariety} of ${\cal S}_{\text{ext}}(N)$.  
In particular, we will give a complete description of the embedding ${\cal S}(2)\xhookrightarrow{} {\cal S}_{\text{ext}}(2)$.  
\end{remark} 

By comparing \autoref{th:S-N-K_appr} with the approximation problem \ref{def:HP} in conjunction with \autoref{th:sol_HP}, one obtains the following identification.
\begin{theorem}\label{th:WZ_s}
Let $\{\l_k,\mathbf{b}_k\}_{k=1}^N\in $${\cal S}(N)$ and 
suppose $F_j^{(1,0)}\neq0$ for $ 2\leq j\leq N$.  Then the second column
 $$
 \begin{pmatrix}
 {\mathbf{s}^{* {(k)}}_{1,2}}\\
 {\mathbf{s}^{(k)}_{2,2}}\\
 {\mathbf{s}^{* {(k)}}_{3,2}}
 \end{pmatrix}
 $$ of the transition matrix $S_{[N,k]}(\l)$ is explicitly given by 
\begin{align*}
\mathbf{s}^{*{(k)}}_{1,2}(\l)=-\mathbf{Q}^*_k(\l),\quad \mathbf{s}^{(k)}_{2,2}(\l)=\mathbf{P}_k(\l),\quad \mathbf{s}^{*{(k)}}_{3,2}(\l)=2\mathbf{\hat P}^*_k(\l), 
\end{align*}
where the entries $(\mathbf{Q}^*_k(\l),\mathbf{P}_k(\l),\mathbf{\hat P}^*_k(\l) )$ are determined in  \autoref{th:sol_HP}.
\end{theorem}
Furthermore, by combining \autoref{th:sol_HP}, \autoref{th:WZ_s} and \autoref{thm:s_xm}, we finally obtain the inverse map mapping the spectral data $\{\l_k,\mathbf{b}_k\}_{k=1}^N\in {\cal S}(N)$ to $\{x_k,\mathbf{m}_k\}_{k=1}^N$.

\begin{theorem}\label{th:sol_xm}
Let $\{\l_k,\mathbf{b}_k\}_{k=1}^N\in {\cal S}(N)$ and 
suppose $F_k^{(1,0)}\neq0$.   
Then   
$\{x_k,\mathbf{m}_k\}_{k=1}^N$ can be recovered from the spectral data according to the formulae
\begin{equation}\label{form:xm}
\begin{aligned}
\mathbf{m}_{N-k+1}e^{-x_{N-k+1}}&=\frac{1}{2}\frac{F_{k-1}^{(1,1)}\mathbf{G}_{k}^{(1,0)}}{F_{k}^{(1,0)}F_{k-1}^{(1,0)}},\qquad \mathbf{m} _{N-k+1}e^{x_{N-k+1}}&=\frac{F_{k}^{(0,0)}\mathbf{G}_{k}^{(1,0)}}{F_{k}^{(1,0)}F_{k-1}^{(1,0)}},
\end{aligned}
\end{equation}
where $F_k^{(i,j)}$ and $\mathbf{G}_{k}^{(i,j)}$ denote the determinants appearing in \autoref{def_Ibeta} and \autoref{def:G}.  
\end{theorem} 
\begin{proof}
First, since $F_k^{(1,0)}\neq0$ and the spectral data $\{\l_k,\mathbf{b}_k\}_{k=1}^N\in {\cal S}(N)$, the approximation problem has a unique solution which must precisely be the second column of $S_{[N,k]}(\l)$.   
A combination of \autoref{th:sol_HP}, \autoref{th:WZ_s} and \autoref{thm:s_xm} yields 
\begin{align*}
-2\mathbf{m}^*_{N-k+1}e^{-x_{N-k+1}}&=\mathbf{s}_{1,2}^{*(k)}[1]-\mathbf{s}_{1,2}^{*(k-1)}[1]=\frac{\mathbf{G}_{k-1}^{*(1,1)}}{F_{k-1}^{(1,0)}}-\frac{\mathbf{G}_{k}^{*(1,1)}}{F_{k}^{(1,0)}},\\
2\mathbf{m}^*_{N-k+1}e^{x_{N-k+1}}&=\mathbf{s}_{3,2}^{*(k)}[0]-\mathbf{s}_{3,2}^{*(k-1)}[0]=2\left(\frac{\mathbf{G}_{k}^{*(0,0)}}{F_{k-1}^{(1,0)}}-\frac{\mathbf{G}_{k+1}^{*(0,0)}}{F_{k}^{(1,0)}}\right).
\end{align*}
After taking the dual, the conclusion immediately follows from the identities
\begin{align*}
&F_{k-1}^{(i,j-1)}\mathbf{G}_{k}^{(i,j)}-F_{k}^{(i,j-1)}\mathbf{G}_{k-1}^{(i,j)}=F_{k-1}^{(i,j)}\mathbf{G}_{k}^{(i,j-1)},\\
&F_{k}^{(i,j-1)}\mathbf{G}_{k}^{(i-1,j-1)}-F_{k}^{(i-1,j-1)}\mathbf{G}_{k}^{(i,j-1)}=F_{k-1}^{(i,j-1)}\mathbf{G}_{k+1}^{(i-1,j-1)},
\end{align*} 
which, in turn, are obtained by employing the Desnanot-Jacobi identity \cite[Section 2.3, Proposition 10]{krattenthaler}, 
i.e., 
\begin{align*}
\mathcal{D}_l \mathcal{D}_l\left(\begin{array}{cc}
i_1 & i_2 \\
j_1 & j_2 \end{array}\right)=\mathcal{D}_l\left(\begin{array}{c}
i_1  \\
j_1 \end{array}\right)\mathcal{D}_l\left(\begin{array}{c}
i_2  \\
j_2 \end{array}\right)-\mathcal{D}_l\left(\begin{array}{c}
i_1  \\
j_2 \end{array}\right)\mathcal{D}_l\left(\begin{array}{c}
i_2  \\
j_1 \end{array}\right),\quad l=1,2,
\end{align*}
with
\begin{align*}
&\mathcal{D}_1=\det\left(
\begin{array}{ccc}
[0]&\big[1, 0,0\cdots &\cdots \cdots \big]\\
\big[\bm{\beta}_{i+p-1}\big]&\big[I_{i+p,j-1}, \cdots &I_{i+p,j+q}\big]
\end{array}
\right)_{p=0,1,\ldots k-1}^{q=0,1,\ldots k-2}, \\
&\mathcal{D}_2=\mathbf{G}_{k+1}^{(i-1,j-1)},\qquad i_1=j_1=1,\quad i_2=j_2=k+1.
\end{align*}
Here $\mathcal{D}_l\left(\begin{array}{cccc}
i_1&i_2 &\cdots& i_k\\
j_1&j_2 &\cdots& j_k
\end{array}\right)$, for $ i_1<i_2<\cdots<i_k,\ j_1<j_2<\cdots<j_k$, denotes the
determinant of the matrix obtained from $\mathcal{D}_l$ by removing the rows with indices
$i_1,i_2,\dots, i_k$ and the columns with indices $j_1,j_2,\dots, j_k$.
\end{proof}

\section{Multipeakons}\label{sec:peakons}
This section presents concrete examples of global multipeakons written in terms of the spectral data for the NV2 equation \eqref{eq:NV2}.
To facilitate the presentation of formulas, we will make a few extra assumptions about the determinants involving the spectral 
data $\{\l_k,\mathbf{b}_k\}_{k=1}^N\in {\cal S}(N)$.  Thus for the remainder of this section, we will assume that 
\begin{equation} \label{eq:extras}
\boxed{F_k^{(0,0)}>0,\,  F_k^{(1,1)}>0, \, F_k^{(1,0)}>0, \,  \mathbf{G}_k^{(1,0)}>0.} 
\end{equation}

 With conditions \eqref{eq:extras} in place we can rewrite  \eqref{form:xm} in an equivalent form, namely, 
\begin{align}\label{eq:xm formulas} 
x_{N-k+1}=\frac{1}{2}\ln\left( \frac{2F_{k}^{(0,0)}}{F_{k-1}^{(1,1)}}\right),\qquad \mathbf{m}_{N-k+1}=\frac{{\mathbf{G}_{k}^{(1,0)}}}{F_{k}^{(1,0)}F_{k-1}^{(1,0)}}\sqrt{\frac{F_{k}^{(0,0)}F_{k-1}^{(1,1)}}{2}},
\end{align}
so that the positions are real, the masses are well-defined and positive, and the positions satisfy the ordering condition 
\begin{equation} \label{eq:x inequalities} 
x_{N-k+1}-x_{N-k}=\frac{1}{2}\ln\left( \frac{F_{k}^{(0,0)}F_{k}^{(1,1)}}{F_{k-1}^{(1,1)}F_{k+1}^{(0,0)}}\right)=\frac{1}{2}\ln\left(1+\frac{\left(F_{k}^{(1,0)}\right)^2}{F_{k-1}^{(1,1)}F_{k+1}^{(0,0)}}\right)>0, 
\end{equation} 
where in the last equality we have used the identity 
$$F_{k-1}^{(1,1)}F_{k+1}^{(0,0)}=F_{k}^{(0,0)}F_{k}^{(1,1)}-\left(F_{k}^{(1,0)}\right)^2,$$
which follows from the Desnanot-Jacobi identity.

It remains an open question whether the conditions given by equation \eqref{eq:extras} are necessary or whether they automatically hold for determinants 
computed from the spectral data in ${\cal{S}}(N)$.  We have some partial evidence that this is indeed the case.  
\begin{conjecture} \label{conj:positivity of dets}
For the spectral data $\{\l_k,\mathbf{b}_k\}_{k=1}^N\in {\cal S}(N)$, all determinants in \eqref{eq:extras} are positive.  
\end{conjecture} 
In the next two subsections, we will elaborate on further evidence supporting the above conjecture.  For now, we just note that for $k=1$ and any $N$, the inequalities given by \eqref{eq:extras}  are automatically satisfied not only when $\{\l_k,\mathbf{b}_k\}_{k=1}^N\in {\cal S}(N)$ but even when  $\{\l_k,\mathbf{b}_k\}_{k=1}^N\in {\cal S}_{\text{ext}}(N)$.  
Indeed, 
$$ F_1^{(i,j)}=\sum_{p=1}^N\sum_{q=1}^N \frac{\l_p^{i}\l_q^{j}}{\l_p+\l_q}\langle \mathbf{b}_p, \mathbf{b}_q\rangle >0, \qquad \mathbf{G}_1^{(i,j)}=\bm{\beta}_{i-1}=\sum_{p=1}^N \l_p^{i-1} \mathbf{b}_p>0. $$

\subsection{Case: $N=1$}

This case is trivial. It is not hard to see that the NV2 equation \eqref{eq:NV2} admits the global 1+1-peakon solution of the form
\begin{equation} 
(u,v)^T=\mathbf{m}_1 e^{-\abs{x-x_1}}, 
\end{equation}
where
\begin{align*} 
&x_1=\frac{1}{2}\ln\left(\frac{\langle \mathbf{b}_1, \mathbf{b}_1\rangle}{\lambda_1}\right)= \frac{t}{2\l_1}+\frac12 \ln\left(\frac{\langle \mathbf{b}_1(0), \mathbf{b}_1(0)\rangle}{\lambda_1}\right),\\ &\mathbf{m}_1=\mathbf{b}_1 \sqrt\frac{1}{\lambda_1\langle \mathbf{b}_1,\mathbf{b}_1\rangle}=\frac{ \mathbf{b}_1(0)}{\sqrt{\l_1\langle \mathbf{b}_1(0), \mathbf{b}_1(0) \rangle}}, 
\end{align*}
with
$$
\l_1>0, \qquad \mathbf{b}_1(t)=\mathbf{b}_1(0)e^\frac{t}{2\l_1}>0.
$$
In other words, we need no constraints on the extended spectral data $\mathcal{S}_{\text{ext}}(1)$, and 
in this case $\mathcal{S}(1)=\mathcal{S}_{\text{ext}}(1)$.  
\subsection{Case: $N=2$}
In the case of $N=2$,  we have the formulae: 

\begin{subequations} \label{eq:FsGs}
\begin{align}
F_1^{(1,0)}&=(b_{1,1}+b_{2,1})(b_{1,2}+b_{2,2})=\frac12 \langle \mathbf{b}_1+\mathbf{b}_2, \mathbf{b}_1+\mathbf{b}_2\rangle,\\
F_2^{(i,j)}&=\frac{(\l_1-\l_2)^2}{(\l_1+\l_2)^2}(\l_1\l_2)^{i+j-1}(b_{1,1}b_{2,2}\l_1-b_{1,2}b_{2,1}\l_2)(b_{1,2}b_{2,1}\l_1-b_{1,1}b_{2,2}\l_2)\notag\\
&=\frac{(\l_1-\l_2)^2}{(\l_1+\l_2)^2}(\l_1\l_2)^{i+j-1}\langle \mathbf{b}_2, {\mathcal{L}} \mathbf{b}_1\rangle \langle \mathbf{b}_1, {\mathcal{L}} \mathbf{b}_2\rangle, \label{eq:F2} \\
\mathbf{G}_2^{(1,0)}&=\frac{(\l_1-\l_2)}{(\l_1+\l_2)}
\begin{pmatrix}
(b_{1,1}+b_{2,1})(b_{1,2}b_{2,1}\l_1-b_{1,1}b_{2,2}\l_2) \notag\\
(b_{1,2}+b_{2,2})(b_{1,1}b_{2,2}\l_1-b_{1,2}b_{2,1}\l_2)
\end{pmatrix}, \label{eq:G2}\\
&=\frac{(\l_1-\l_2)}{(\l_1+\l_2)}
\begin{pmatrix}
(b_{1,1}+b_{2,1})\langle \mathbf{b}_1, {\mathcal{L}} \mathbf{b}_2\rangle\\
(b_{1,2}+b_{2,2})\langle \mathbf{b}_2, {\mathcal{L}} \mathbf{b}_1\rangle
\end{pmatrix}, 
\end{align}
\end{subequations}
where ${\cal{L}}= \begin{pmatrix}\l_1&0\\0&-\l_2 \end{pmatrix}$. 
These formulae are valid for any extended spectral data in $\mathcal{S}_{\text{ext}}(2)$.  
Now, suppose the spectral data is in $\mathcal{S}(2)$, and the conditions \eqref{eq:extras} hold.  Then the 2+2-peakon solution of the NV2 equation takes the form 
\begin{equation*} 
 (u,v)^T=\sum_{k=1}^2\bold{m}_k e^{-\abs{x-x_k}}, 
\end{equation*}
where the positions $x_k$ and amplitudes $m_k$ are given by 
\begin{align*}
&x_1=\frac{1}{2}\ln\left(\frac{4\langle \mathbf{b_2}, {\cal L} \mathbf{b}_1 \rangle \langle \mathbf{b_1}, {\cal L} \mathbf{b}_2 \rangle(\l_1-\l_2)^2}{(\l_1+\l_2)\l_1\l_2\big(\left(\langle \mathbf{b}_{1}, \mathbf{b}_{1}\rangle\l_1+\langle \mathbf{b}_{2}, \mathbf{b}_{2}\rangle\l_2\right)(\l_1+\l_2)+4\langle \mathbf{b}_{1}, \mathbf{b}_{2}\rangle\l_1\l_2 \big)}\right),\\
&x_2=\frac{1}{2}\ln\left(\frac{\langle \mathbf{b}_{1}, \mathbf{b}_{1}\rangle }{\l_1}+\frac{\langle \mathbf{b}_2, \mathbf{b}_{2}\rangle}{\l_2}+4\frac{\langle \mathbf{b}_{1}, \mathbf{b}_{2}\rangle}{\l_1+\l_2}\right),\\
&m_1=\frac{1}{2(b_{1,2}+b_{2,2})}\sqrt{\frac{\langle\mathbf{b}_1, {\cal L} \mathbf{b_2} \rangle\big(\left(\langle \mathbf{b}_{1}, \mathbf{b}_{1}\rangle\l_1+\langle \mathbf{b}_{2}, \mathbf{b}_{2}\rangle\l_2\right)(\l_1+\l_2)+4\langle \mathbf{b}_{1}, \mathbf{b}_{2}\rangle\l_1\l_2\big)}{(\l_1+\l_2)\l_1\l_2
\langle\mathbf{b}_2, {\cal L} \mathbf{b_1} \rangle}},\\
&n_1=\frac{1}{2(b_{1,1}+b_{2,1})}\sqrt{\frac{\langle\mathbf{b}_2, {\cal L} \mathbf{b_1} \rangle\big(\left(\langle \mathbf{b}_{1}, \mathbf{b}_{1}\rangle\l_1+\langle \mathbf{b}_{2}, \mathbf{b}_{2}\rangle\l_2\right)(\l_1+\l_2)+4\langle \mathbf{b}_{1}, \mathbf{b}_{2}\rangle\l_1\l_2\big)}{(\l_1+\l_2)\l_1\l_2
\langle\mathbf{b}_1, {\cal L} \mathbf{b_2} \rangle}},\\
&m_2=\frac{1}{2(b_{1,2}+b_{2,2})}\sqrt{\left(\frac{\langle \mathbf{b}_{1}, \mathbf{b}_{1}\rangle }{\l_1}+\frac{\langle \mathbf{b}_2, \mathbf{b}_{2}\rangle}{\l_2}+4\frac{\langle \mathbf{b}_{1}, \mathbf{b}_{2}\rangle}{\l_1+\l_2}\right)}, \\
&n_2=\frac{1}{2(b_{1,1}+b_{2,1})}\sqrt{\left(\frac{\langle \mathbf{b}_{1}, \mathbf{b}_{1}\rangle }{\l_1}+\frac{\langle \mathbf{b}_2, \mathbf{b}_{2}\rangle}{\l_2}+4\frac{\langle \mathbf{b}_{1}, \mathbf{b}_{2}\rangle}{\l_1+\l_2}\right)}.  
\end{align*}
These formulae were obtained under the assumption that the spectral data belongs to $\mathcal{S}(2)$ and under the additional constraints 
given by \eqref{eq:extras}.  The latter assumption implies that the inequality in \eqref{eq:x inequalities} automatically holds.  
However, if one considers the extended spectral data in $\mathcal{S}_{\text{ext}}(2)$, that is 
\begin{equation} \label{eq:spectral data range}
  \l_1>\l_2>0, \text{ and } \mathbf{b}_{i}(t)=\mathbf{b}_{i}(0)e^\frac{t}{2\l_i}>0,\qquad  i=1,2, 
\end{equation} 
and one requires that the inequalities in \eqref{eq:extras} be satisfied, then \eqref{eq:x inequalities} holds and the masses are positive.  
One can ask a legitimate question whether the conditions on the extended spectral data given by \eqref{eq:spectral data range} automatically 
guarantee the quantities $x_1(t), x_2(t), m_1(t), n_1(t), m_2(t), n_2(t)$ to represent a peakon system.  
The answer is a decisive no.  In other words, $\mathcal{S}(2)$ is a proper subset of $\mathcal{S}_{\text{ext}}(2)$.  
\begin{proposition} \label{prop:two-peakon spectral data} 
The extended spectral data in $\mathcal{S}_{\text{ext}}(2)$ is in $\mathcal{S}(2)$  if and only if 
\begin{equation}\label{eq:condition1}
\langle \mathbf{b}_{1}(0),{\mathcal{L}}  \mathbf{b}_{2}(0)\rangle>0 \text{   and   }  \langle \mathbf{b}_{2}(0),{\mathcal{L}}  \mathbf{b}_{1}(0)\rangle>0.  
\end{equation} 
\end{proposition} 
\begin{proof} First, since ${\cal{L}}= \begin{pmatrix}\l_1&0\\0&-\l_2 \end{pmatrix}$,  \eqref{eq:condition1} is equivalent to 
\begin{equation} \label{eq:n=2condition}
\frac{\l_1}{\l_2}>\max\left\{\frac{b_{1,2}(0)b_{2,1}(0)}{b_{1,1}(0)b_{2,2}(0)},\frac{b_{1,1}(0)b_{2,2}(0)}{b_{1,2}(0)b_{2,1}(0)}\right\}.  
\end{equation} 
Suppose the spectral data is now in $\mathcal{S}(2)$.  That means that a global 2+2-peakon solution exists, as proven earlier.  
Let us now revisit the forward problem for $N=2$. 
From \autoref{ex:res12} we have 
\begin{equation}\label{eq:b2}
\lim_{t\rightarrow \infty} (\mathbf{B}(\l_2) e^{-x_2})=\bold{m}_2(\infty) (1-\frac{\l_2}{\l_1})=
\begin{pmatrix}
m_2(\infty)(1-\frac{\l_2}{\l_1})\\
n_2(\infty)(1-\frac{\l_2}{\l_1})
\end{pmatrix}, 
\end{equation}
\begin{equation} \label{eq:b1}
\begin{gathered} 
\lim_{t\rightarrow \infty} (\mathbf{B}(\l_1) e^{-x_1})=\lim_{t\rightarrow \infty} \frac{(1-\l_1\bold{m}_1^*\bold{m}_1)}{e^{x_1-x_2}} \bold{m}_2(\infty)
+\bold{m}_1(\infty)\\=\lim_{t\rightarrow \infty} \frac{(1-2\l_1 m_1(t) n_1(t)) }{e^{x_1-x_2}}\bold{m}_2(\infty)+\bold{m}_1(\infty)\\
=\frac{1}{{m_1(\infty)n_1(\infty)-m_2(\infty)n_2(\infty)}}\begin{bmatrix} n_1(\infty)(m_1^2(\infty)+m_2^2(\infty))\\ m_1(\infty) (n_1^2(\infty)+n_2^2(\infty)) \end{bmatrix}. 
\end{gathered} 
\end{equation} 
As $t\rightarrow \infty$, the residues satisfy the following formulae
\begin{equation*} 
\begin{gathered} 
\lim_{t\rightarrow \infty} 
\begin{pmatrix}
b_{1,1}(0)e^{\frac{t}{2\l_1}-x_1(t)}\\
b_{1,2}(0)e^{\frac{t}{2\l_1}-x_1(t)}
\end{pmatrix}
=
\begin{pmatrix}
\frac{n_1(\infty)(m_1^2(\infty)+m_2^2(\infty))}{{m_2(\infty)n_2(\infty)-m_1(\infty)n_1(\infty)}}\frac{\l_1\l_2}{\l_1-\l_2}\\
\frac{m_1(\infty)(n_1^2(\infty)+n_2^2(\infty))}{{m_2(\infty)n_2(\infty)-m_1(\infty)n_1(\infty)}}\frac{\l_1\l_2}{\l_1-\l_2}
\end{pmatrix},
\end{gathered} 
\end{equation*} 

\begin{equation*} 
\begin{gathered} 
\lim_{t\rightarrow \infty} 
\begin{pmatrix}
b_{2,1}(0)e^{\frac{t}{2\l_2}-x_2(t)}\\
b_{2,2}(0)e^{\frac{t}{2\l_2}-x_2(t)}
\end{pmatrix}
=
\begin{pmatrix}
m_2(\infty)(1-\frac{\l_2}{\l_1})\frac{\l_1\l_2}{\l_1-\l_2}\\
n_2(\infty)(1-\frac{\l_2}{\l_1})\frac{\l_1\l_2}{\l_1-\l_2}
\end{pmatrix}.
\end{gathered} 
\end{equation*} 
Thus, we get 
$$
\frac{b_{1,1}(0)b_{2,2}(0)}{b_{1,2}(0)b_{2,1}(0)}=\frac{n_1(\infty)n_2(\infty)(m_1^2(\infty)+m_2^2(\infty))}{m_1(\infty)m_2(\infty)(n_1^2(\infty)+n_2^2(\infty))}.
$$
Recalling that 
$$
\frac{1}{2 m_1(\infty)n_1(\infty)}=\l _1> \l_2=\frac{1}{2 m_2(\infty)n_2(\infty)}, 
$$
we are finally led to
\begin{align*}
\frac{\l_2}{\l_1}=\frac{m_1(\infty)n_1(\infty)}{m_2(\infty)n_2(\infty)}<&\frac{b_{1,1}(0)b_{2,2}(0)}{b_{1,2}(0)b_{2,1}(0)}\\
&=\frac{n_1(\infty)n_2(\infty)(m_1^2(\infty)+m_2^2(\infty))}{m_1(\infty)m_2(\infty)(n_1^2(\infty)+n_2^2(\infty))}<
\frac{m_2(\infty)n_2(\infty)}{m_1(\infty)n_1(\infty)}=\frac{\l_1}{\l_2},
\end{align*}
which is equivalent to the condition \eqref{eq:n=2condition}, hence to \eqref{eq:condition1}.  

Let us now prove the converse.  
Thus we assume that the spectral data $(\l_1>\l_2, \mathbf{b}_1>0, \mathbf{b}_2>0) \in \mathcal{S}_{\text{ext}}$ and conditions \eqref{eq:condition1} hold. 
Let us denote the corresponding spectral measure by $d\mu(x)=\sum_{j=1}^2 \mathbf{b}_j(t) \delta(x-\l_j)$. 
 Then, using the formulas \eqref{eq:xm formulas} and \eqref{eq:FsGs}, we see that $x_1(t), x_2(t)$ satisfy the inequalities \eqref{eq:x inequalities},  
and the masses $\mathbf{m}_1(t), \mathbf{m}_2(t)$ are positive.  Thus we have a peakon system 
satisfying the equations of motion \eqref{eq:epeakons}.  The peakon system has 
its own Weyl functions $\mathbf{W}(\l)$ and $Z(\l)$  and its own spectral measure $d\hat {\mu}(x)=\sum_{j=1}^2
\hat {\mathbf{b}}_j(t)\delta(x-\l_j)$ with the same support $\l_1>\l_2>\dots>\l_N$, since the support is entirely 
determined by the asymptotics of $x_j$ via $x_j(t)=\frac{t}{2\l_j}+O(1),\, t\rightarrow \infty$ (see \autoref{cor:as-speeds}).  
So it remains to prove that $\hat{\mathbf{b}}_j(t)=\mathbf{b}_j(t),  j=1,2$.  However, since the time evolutions for $\hat{\mathbf{b}}_j$ and $\mathbf{b}_j$ 
are the same, it suffices to show that $\lim_{t\rightarrow \infty} \hat{\mathbf{b}}_j(t) e^{-x_j(t)}=\lim_{t\rightarrow \infty} \mathbf{b}_j(t) e^{-x_j(t)}$.  
The latter statement, however, follows readily from \eqref{eq:b1} and \eqref{eq:b2} since the limiting values depend only on the masses at $t=\infty$.  

\end{proof} 
A quick look at the formula \eqref{eq:G2} leads to the following conclusion.  
\begin{corollary} 
The extended spectral data in $\mathcal{S}_{\text{ext}}(2)$ belongs to $\mathcal{S}(2)$  if and only if ~$\mathbf{G}_2^{(1,0)}>0$.  
\end{corollary}
In anticipation of our strategy to tackle the general case, we recall (see \autoref{prop:Gk}) the factorization 
\begin{align} \label{eq:G2factor}
\mathbf{G}^{(1,0)}_2=\det \begin{pmatrix} \bm \beta_0&I_{1,0}\\\bm\beta_1&I_{2,0} \end{pmatrix} =\det \left (\begin{pmatrix} \l_2&\l_1\\\l_2^2&\l_1^2 \end{pmatrix} 
\begin{pmatrix}\frac{\mathbf{b}_2}{\l_2}&\frac{\langle \mathbf{b}_2, \mathbf{b}_2 \rangle }{\l_2+\l_2}&\frac{\langle \mathbf{b}_2, \mathbf{b}_1 \rangle }{\l_2+\l_1}\\
\frac{\mathbf{b}_1}{\l_1}& \frac{\langle \mathbf{b}_1, \mathbf{b}_2 \rangle }{\l_1+\l_2}&\frac{\langle \mathbf{b}_1, \mathbf{b}_1\rangle }{\l_1+\l_1}\end{pmatrix} \begin{pmatrix} 1&0\\0&1\\0&1 \end{pmatrix}\right ), 
\end{align} 
which can be verified directly using \autoref{def:G}.  
The matrix in the middle  is $\Lambda_2(2)$ (see \eqref{eq:Lambda2N}).  Thus 
\begin{equation} \label{eq:Lambda22}
\Lambda_2(2)=\begin{pmatrix}\frac{\mathbf{b}_2}{\l_2}&\frac{\langle \mathbf{b}_2, \mathbf{b}_2 \rangle }{\l_2+\l_2}&\frac{\langle \mathbf{b}_2, \mathbf{b}_1 \rangle }{\l_2+\l_1}\\
\frac{\mathbf{b}_1}{\l_1}&\frac{\langle \mathbf{b}_1, \mathbf{b}_2 \rangle }{\l_1+\l_2}&\frac{\langle \mathbf{b}_1 \mathbf{b}_1 \rangle }{\l_1+\l_1} \end{pmatrix}.  \end{equation} 
 
We now extend the definition of Gantmacher and Krein \cite{gantmacher-krein} of \emph{totally positive } (TP) matrices to our setup 
in which one column in the matrix contains vector-valued entries for which positivity is defined entry-wise.  
Then a matrix of this type will be called \emph{totally positive} if the determinant of every square submatrix is positive, be that vector or scalar-valued.  

For $\Lambda_2(2)$, 
every $2$ by $2$ determinant of the submatrix involving the last two columns is scalar-valued, while the determinant for any $2$ by $2$ submatrix involving the first  column is vector-valued.  All  $1$ by $1$ minors of $\Lambda_2(2) $ are 
positive for any extended spectral data.  The following proposition addresses the total positivity of $\Lambda_2(2)$. 
\begin{proposition} \label{prop:Lambda22}
The extended spectral data in $\mathcal{S}_{\text{ext}}(2)$ belongs to $\mathcal{S}(2)$ 
for a two-peakon NV2 system if and only if $\Lambda_2(2) $ is totally positive.  
\end{proposition}
\begin{proof} 
As was said above, the matrix entries of $\Lambda_2(2) $ are positive.  We only need to check for positivity of the $2$ by $2$ minors.  
We have 
\begin{equation*} 
\det \begin{pmatrix} \frac{\mathbf{b}_2}{\l_2}&\frac{\langle \mathbf{b}_2, \mathbf{b}_1\rangle}{\l_2+\l_1}\\
\frac{\mathbf{b}_1}{\l_1}&\frac{\langle \mathbf{b}_1, \mathbf{b}_1\rangle}{\l_1+\l_1}\end{pmatrix} =\frac{1}{\l_1 \l_2(\l_1+\l_2)}
\begin{pmatrix} 
 b_{1,1}\langle \mathbf{b}_1, \mathcal{L} \mathbf{b}_2 \rangle \\
b_{1,2}\langle \mathbf{b}_2, \mathcal{L} \mathbf{b}_1 \rangle
\end{pmatrix}, 
\end{equation*} 
\begin{equation*} 
\det \begin{pmatrix}  \frac{\mathbf{b}_2}{\l_2}&\frac{\langle \mathbf{b}_2, \mathbf{b}_2\rangle}{\l_2+\l_2}\\
\frac{\mathbf{b}_1}{\l_1}&\frac{\langle \mathbf{b}_1, \mathbf{b}_2\rangle}{\l_1+\l_2}\end{pmatrix} =\frac{1}{\l_1 \l_2(\l_1+\l_2)}
\begin{pmatrix} 
 b_{2,1}\langle \mathbf{b}_1, \mathcal{L} \mathbf{b}_2 \rangle \\
b_{2,2}\langle \mathbf{b}_2, \mathcal{L} \mathbf{b}_1 \rangle
\end{pmatrix}, 
\end{equation*}
which is fully consistent with \eqref{eq:G2}.  Moreover, 
\begin{equation*} 
\det \begin{pmatrix}\frac{\langle \mathbf{b}_2, \mathbf{b}_2 \rangle }{\l_2+\l_2}&\frac{\langle \mathbf{b}_2, \mathbf{b}_1 \rangle }{\l_2+\l_1}\\
\frac{\langle \mathbf{b}_1, \mathbf{b}_2 \rangle }{\l_1+\l_2}&\frac{\langle \mathbf{b}_1, \mathbf{b}_1\rangle }{\l_1+\l_1}\end{pmatrix}=
\frac{1}{\l_1 \l_2 (\l_1+\l_2)^2} \langle \mathbf{b}_1, \mathcal{L} \mathbf{b}_2 \rangle \langle \mathbf{b}_2, \mathcal{L} \mathbf{b}_1 \rangle, 
\end{equation*} 
which in turn agrees with \eqref{eq:F2}.  The final statement easily follows from \autoref{prop:two-peakon spectral data}.  
\end{proof} 
We note that by \autoref{prop:Fk}, $F_{2}^{(i, j)}$ also admits a factorization formula akin to \eqref{eq:G2factor} for $\mathbf{G}_2^{(1,0)}$.  
The explicit formula for $F_{2}^{(i, j)}$ reads: 
\begin{equation} \label{eq:F2factor} 
F_2^{(i,j)}=\det \left(\begin{pmatrix} \l_2^i& \l_1^i\\
                                       \l_2^{i+1}& \l_1^{i+1} \end{pmatrix} \begin{pmatrix}\frac{\langle \mathbf{b}_2, \mathbf{b}_2 \rangle }{\l_2+\l_2}&\frac{\langle \mathbf{b}_2, \mathbf{b}_1 \rangle }{\l_2+\l_1}\\
\frac{\langle \mathbf{b}_1, \mathbf{b}_2 \rangle }{\l_1+\l_2}&\frac{\langle \mathbf{b}_1, \mathbf{b}_1 \rangle }{\l_1+\l_1}\end{pmatrix}\begin{pmatrix}
\l_2^j& \l_2^{j+1}\\
\l _1^{j}&\l_1^{j+1}\end{pmatrix} \right). 
\end{equation} 
The matrix in the middle, which is a special case of $\Lambda_1(N)$ introduced 
in \eqref{eq:Lambda1N}, has positive determinant evaluated above, so $F_2^{(i,j)}>0$ 
on the spectral data.  
\subsection{General case} \label{subsec:general}
Our conjecture for the general case of $N$ peakons is inspired by 
\autoref{prop:Lambda22}.  
\begin{conjecture} \label{conj:Lambda2N}
The extended spectral data in $\mathcal{S}_{\text{ext}}(N)$ belongs to $\mathcal{S}(N)$ 
for a  $N$-peakon NV2 system if and only if $\Lambda_2(N) $ is totally positive.  
\end{conjecture}
\begin{remark} 
We note that the scalar minors of $\Lambda _2(N)$ are shared with the minors of $\Lambda_1(N)$, so if $\Lambda_2(N)$ is totally positive, then so is $\Lambda_1(N)$. 
 Moreover, the conjecture implies the positivity of all determinants $F_k^{(i,j)}$ and $\mathbf{G}_k^{(i,j)}$.  
 \end{remark} 
\begin{proposition}\label{prop:GF positivity} 
Suppose $\Lambda_2(N)$ is totally positive.  Then 
$$ 
F_k^{(i,j)}>0, \hspace{1 cm} \text{and           } \qquad \mathbf{G}_k^{(i,j)}>0
$$ 
for all $i,j \in \mathbf{N}$ and $1\leq k\leq N$.  Furthermore, 
all determinants $F_k^{(i,j)}$ and  $\mathbf{G}_k^{(i,j)}$ are $0$ if $k>N$.   
\end{proposition} 
\begin{proof} 
First, as was pointed out in the remark above, since $\Lambda_2(N)$ is totally positive, so is $\Lambda_1(N)$.  
We now examine the form of $F_k^{(i,j)}$ given by \autoref{prop:Fk}.  The first and the third matrix in the product are 
the generalized Vandermonde matrices and, with the ordering of eigenvalues $\l_1>\l_2>\cdots>\l_N$, they are both 
totally positive \cite[p.76] {gantmacher-krein}.  
The remainder of the proof of the positivity of $F_k^{(i,j)}$ follows automatically from the Cauchy-Binet formula under the assumption that 
$k\leq N$.  
The case of $\mathbf{G}_k^{(i,j)}$ can be disposed of equally easily, using \autoref{prop:Gk}.  The only new feature is the presence of a totally non-negative third factor; 
some of the minors will be zero.  
However, there is at least one strictly non-zero minor of the last matrix, and, again, by 
the Cauchy-Binet formula $\mathbf{G}_k^{(i,j)}>0$ for $k\leq N$.  
Finally, the case $k>N$ is a consequence of the rank restrictions  \cite[Lemma 2.10]{chang2018degasperis}.  
\end{proof} 
\begin{remark} 
We note that \autoref{prop:GF positivity} together with \autoref{conj:Lambda2N} implies \autoref{conj:positivity of dets}.  
\end{remark}

Finally, we claim that one can always impose additional conditions on the 3N-tuples of numbers $\{\l_k,\mathbf{b}_k\}_{k=1}^N$ to obtain a totally positive $\Lambda_1(N),\Lambda_2(N)$.
Indeed, let us first consider $\Lambda_1(N)$. We note that $\Lambda_1(N)$ can be rewritten as
\begin{align*}
\Lambda_1(N)=\begin{pmatrix}
\frac{2b_{i,1}b_{j,2} }{\l_i+\l_j}+\frac{b_{i,1}b_{j,2} }{\l_i+\l_j}(\frac{b_{i,2}b_{j,1} }{b_{i,1}b_{j,2} }-1)
\end{pmatrix}_{i,j=1,2,\ldots,N}.
\end{align*}
It then follows from the well-known Cauchy determinant formula that
\begin{align*}
\begin{pmatrix}
\frac{2b_{i,1}b_{j,2} }{\l_i+\l_j}
\end{pmatrix}_{i,j=1,2,\ldots,N}
\end{align*}
is a totally positive matrix. Therefore, we can conclude that $\Lambda_1(N)$ is totally positive if all the terms $\frac{b_{i,2}b_{j,1} }{b_{i,1}b_{j,2} }-1$ are close enough to zero; in such case, $\Lambda_1(N)$ can be regarded as a  perturbation of a Cauchy-type matrix. Likewise, $\Lambda_2(N)$ is also totally positive under the constraint that the components $\frac{b_{i,2}b_{j,1} }{b_{i,1}b_{j,2} }-1$ are sufficiently small.   

In summary, we have the following perturbative result.
\begin{theorem}
Let $N$ be an arbitrary natural number.  Suppose $\{\l_1,\l_2,\ldots,\l_N\}$ satisfy $\l_1>\l_2>\cdots>\l_N>0$.  Let us choose the positive two-component vectors $\{\mathbf{b}_1,\mathbf{b}_2,\ldots,\mathbf{b}_N\}$ so that all the ratios $\frac{b_{i,2}b_{j,1} }{b_{i,1}b_{j,2} }$ are close enough to one. Then $\Lambda_1(N)$ and $\Lambda_2(N) $ are both totally positive.
\end{theorem}

\section{Acknowledgements}

Xiangke Chang's research is supported by the National Natural Science Foundation of China (Grant Nos. 12222119, 12171461, 12288201, 11931017) and the Youth Innovation Promotion Association CAS.

\noindent Jacek Szmigielski’s research is supported by the Natural Sciences and Engineering Research Council of Canada (NSERC).

\begin{appendices}
\section{Appendix A: \, A lemma}\label{app1} 
\renewcommand{\theequation}{A.\arabic{equation}}
\begin{lemma} \label{lem:A} 
Suppose $f\in C^1(0,\infty), f\geq0$, 
$$ 
\int_0^\infty f(t) dt<\infty, 
$$ and 
$$ 
\abs{\frac{df}{dt}}\leq C.  
$$

Then 
\begin{equation} \label{eq:zerolim}
\lim_{t\rightarrow \infty} f(t) =0. 
\end{equation} 
\end{lemma} 
\begin{proof} 
The proof proceeds by contradiction.  If the claim does not hold, there exists a positive number $\epsilon$ and a sequence $\{t_k: t_k\rightarrow \infty \}$ such that $f(t_k)\geq 2\epsilon$.  One can arrange, possibly going to a subsequence, that $t_k$ is increasing and 
$$ 
t_{k+1}-t_{k}> \frac{2\epsilon}{C}.  
$$
Now consider intervals $I_k=[t_k, s_k]$ where the sequence $\{s_k\}$ is chosen in such a way that 
$t_k<s_k$ and $\frac{\epsilon}{2C}<s_k-t_k<\frac{\epsilon}{C}$. 
 Clearly, $[t_k, s_k]\subset [t_k, t_{k+1}]$ and 
the intervals $[t_k,s_k]$ are mutually disjoint.  

Suppose now $t\in [t_k, s_k]$.  Then 
$$ 
\abs{f(t)-f(t_k)}\leq (t-t_k)C, 
$$ 
hence 
$$ 
\overbrace {f(t_k)}^{\geq 2\epsilon}  -\overbrace{(t-t_k)}^{\leq \frac{\epsilon}{C}}C\leq f(t), 
$$ 
which in turn implies 
$$ 
\epsilon \leq f(t).  
$$  
Finally $$
\sum_{k} \int_{I_k} f(t) dt \leq  \int_0^\infty f(t) dt < \infty , $$
and we obtain a sought contradiction since $\sum_{k} \int_{I_k} f(t) dt$ is divergent being bounded from below by 
$\sum_k \frac{\epsilon^2}{2C}$. 
\end{proof}

\section{Appendix B: the second exterior power of $S_{[N,k]}(\l)$  }\label{app2} 
\renewcommand{\theequation}{B.\arabic{equation}}
The purpose of this appendix is to explore the dependence of $\Lambda^2 S_{[N,k]}(\l) $ on the spectral parameter $\l$.  
We will have to do some explicit computations to get insight into the structure of the second exterior power of 
$S_{[N,k]}(\l) $.  

First, we need to choose a basis in which to compute the matrix entries of $\Lambda^2 S_{[N,k]}(\l)$.  
We will follow the convention used in the paper according to which $\C^4$ is spanned by $e_1, e_{21}, e_{22}, e_3$.  
Subsequently, we choose the basis of $\C^4\wedge \C^4$ to be $\{e_1\wedge e_3, e_1\wedge e_{21}, e_1\wedge e_{22}, 
e_{21}\wedge e_3, e_{22}\wedge e_3, e_{21}\wedge e_{22} \}$. Finally, to ease the notational burden, we will use the index set $k, k-1, k-2, \dots 1$, 
rather than $N, N-1, N-2, \cdots, N-k+1$.  In other words, we will be computing 
$\Lambda^2 S_{[k,1]}(\l) $.  
For any index $j$, by direct computation, we obtain 
\begin{equation} 
\begin{split} 
&\Lambda^2 S_j(\l) =\\
&\begin{pmatrix} 1&2n_j e^{x_j}&2m_je^{x_j}&-2\l n_je^{-x_j}&-2 \l m_j e^{-x_j}&0\\
\l m_j e^{-x_j}&1&2\l m_j^2&0&-2\l^2 m_j^2 e^{-2x_j}&2\l m_j e^{-x_j}\\
\l n_j e^{-x_j}&2\l n_j^2&1&-2\l^2 n_j^2 e^{-2x_j}&0&-2\l n_j e^{-x_j}\\
m_j e^{x_j}&0&2m_j^2 e^{2x_j}&1&-2\l m_j^2&2m_je^{x_j} \\
n_je^{x_j}&2n_j^2 e^{2x_j}&0&-2\l n_j^2&1&-2n_je^{x_j}\\
0&-n_je^{x_j}&m_je^{x_j}&\l n_j e^{-x_j}&-\l m_je^{-x_j}&1 
\end{pmatrix} . 
\end{split} 
\end{equation} 
Since we are only interested in the order of magnitude of the elements, we rewrite the latter equation as 
\begin{equation} \label{eq:extpSj}
\begin{split} 
\Lambda^2 S_j(\l) =
\begin{pmatrix} \l^0&\l^0&\l^0 &\l^1 &\l^1&0\\
\l^1&\l^0 &\l^1 &0&\l^2 &\l ^1\\
\l^1 &\l^1 &\l^0&\l^2 &0&\l^1\\
\l^0 &0&\l^0 &\l^0 &\l^1 &\l^0 \\
\l^0&\l^0 &0&\l^1 &\l^0&\l^0 \\
0&\l^0 &\l^0 &\l^1 &\l^1 &\l^0 
\end{pmatrix}, 
\end{split} 
\end{equation} 
where $\l^m$ is a shorthand for $O(\l^m)$.  
We are now ready to state the main proposition of this appendix.  
\begin{proposition} \label{prop:extpowerdeg}
\mbox{}
\begin{enumerate} 
\item Let $k$ be an odd index.  Then 
\begin{equation} \label{eq:extpSodd}
\begin{split} 
\Lambda^2 S_{[k,1]}(\l) =
\begin{pmatrix} \l^{k-1}&\l^{k-1}&\l^{k-1} &\l^k &\l^k&\l^{k-1} \\
\l^k&\l^{k-1} &\l^k &\l^k&\l^{k+1} &\l ^k\\
\l^k &\l^k&\l^{k-1}&\l^{k+1} &\l^k &\l^k\\
\l^{k-1} &\l^{k-2}&\l^{k-1} &\l^{k-1}&\l^k &\l^{k-1} \\
\l^{k-1}&\l^{k-1} &\l^{k-2}&\l^k &\l^{k-1}&\l^{k-1}\\
\l^{k-1}&\l^{k-1} &\l^{k-1}&\l^k &\l^k &\l^{k-1}
\end{pmatrix}.  
\end{split} 
\end{equation} 
\item Let $k$ be an even index.  Then 
\begin{equation} \label{eq:extpSeven} 
\begin{split} 
\Lambda^2 S_{[k,1]}(\l) =
\begin{pmatrix} \l^{k-1}&\l^{k-1}&\l^{k-1} &\l^k &\l^k&\l^{k-1} \\
\l^k&\l^{k} &\l^{k-1} &\l^{k+1}&\l^{k} &\l ^k\\
\l^k &\l^{k-1}&\l^{k}&\l^{k} &\l^{k+1} &\l^k\\
\l^{k-1} &\l^{k-1}&\l^{k-2} &\l^{k}&\l^{k-1} &\l^{k-1} \\
\l^{k-1}&\l^{k-2} &\l^{k-1}&\l^{k-1} &\l^{k}&\l^{k-1}\\
\l^{k-1}&\l^{k-1} &\l^{k-1}&\l^k &\l^k &\l^{k-1}
\end{pmatrix}. 
\end{split} 
\end{equation} 
\end{enumerate} 
\end{proposition} 
\begin{proof} 
For the base case $k=1$ we  obtain the claim from \eqref{eq:extpSj} after setting $j=1$ and noting 
that $0$ can be viewed both as $O(1)$ and $O(\l^{-1})$.  
Now we consider odd $k$, hence we use \eqref{eq:extpSodd}, and multiply on the left by \eqref{eq:extpSj}  to get
\begin{equation} \label{eq:extpSevenshift} 
\begin{split} 
\Lambda^2 S_{k+1}(\l) \Lambda^2 S_{[k,1]}(\l) =
\begin{pmatrix} \l^{k}&\l^{k}&\l^{k} &\l^{k+1} &\l^{k+1}&\l^{k} \\
\l^{k+1} &\l^{k+1} &\l^{k} &\l^{k+2}&\l^{k+1} &\l ^{k+1}\\
\l^{k+1} &\l^{k}&\l^{k+1}&\l^{k+1} &\l^{k+2} &\l^{k+1}\\
\l^{k} &\l^{k}&\l^{k-1} &\l^{k+1}&\l^{k} &\l^{k} \\
\l^{k}&\l^{k-1} &\l^{k}&\l^{k} &\l^{k+1}&\l^{k}\\
\l^{k}&\l^{k} &\l^{k}&\l^{k+1} &\l^{k+1}&\l^{k}
\end{pmatrix}, 
\end{split} 
\end{equation} 
which agrees with \eqref{eq:extpSeven} after the latter is written for $k$ replaced with $k+1$.  
The proof for even k is similar.  
\end{proof} 
\autoref{prop:extpowerdeg} has an immediate corollary which we now discuss.  
First, we note that $\C^4\wedge \C^4$ splits into four subspaces $W_j, j=1,2, 3, 4, $ where 
\begin{align*} 
&W_1=\text{span}_\C \{e_1\wedge e_3\}, &W_2=\text{span}_\C\{e_1\wedge e_{21}, e_1\wedge e_{22}\}, \\
&W_3=\text{span}_\C\{e_{21}\wedge e_{3}, e_{22}\wedge e_{3}\}, \qquad &W_4=\text{span}_\C \{e_{21}\wedge e_{22} \}. \qquad \quad \, \, \: 
\end{align*} 
Thus any linear operator $X\in \text{End}(\C^4\wedge \C^4)$ has the block decomposition: 
\begin{equation*} 
X=\begin{pmatrix} 
X_{11}&X_{12}&X_{13}&X_{14}\\
X_{21}&X_{22}&X_{23}&X_{24}\\
X_{31}&X_{32}&X_{33}&X_{34}\\
X_{41}&X_{42}&X_{43}&X_{44}
\end{pmatrix}, 
\end{equation*}
where $X_{ij}\in \text{Hom}_\C (W_j, W_i)$.  
In \autoref{sec:ApproxP}, we use such a block decomposition of $\Lambda^2_{[1,k]}$ (up to a relabeling of the indices).  
\begin{corollary}\label{cor:degT} 
Let $T$ denote $\Lambda^2_{[1,k]}$.  
Then $T$ has the block decomposition for its dominant terms as follows: 
\begin{equation*} 
T=\begin{pmatrix} \l^{k-1}&\l^{k-1}&\l^k&\l^{k-1}\\
                              \l^k&\l^{k}&\l^{k+1}&\l^k\\
                              \l^{k-1}&\l^{k-1}&\l^{k}&\l^{k-1}\\
                              \l^{k-1}&\l^{k-1}&\l^k&\l^{k-1} \end{pmatrix}.  
                              \end{equation*}
                              
\end{corollary}

\section{Appendix C: The reduction to NV1}\label{app3} 
This appendix addresses the question of the reduction from NV2 to NV1.  On the formal level, by setting $u=v$ in \eqref{eq:NV2} one obtains two copies (\eqref{eq:NV2m} and \eqref{eq:NV2n}) of \eqref{eq:NV-second}.  Likewise, in the peakon sector, 
if one starts at some initial time $t_0$ with $m_j(t_0)>0$ and $x_1(t_0)<x_2(t_0)<\dots<x_{N}(t_0)$ , solves the NV1 peakon equation 
\begin{equation*} 
\dot x_j=u^2(x_j), \qquad \dot m_j=-m_j \avg{u_x}(x_j) u(x_j), 
\end{equation*} 
and, finally, sets $n_j(t)=m_j(t)$, then $m_j(t), n_j(t), x_j(t)$ will solve the NV2 peakon equation \eqref{eq:epeakons}, and this will be the unique solution to the initial data $m_j(0)=n_j(0)$ and $x_1(t_0)<x_2(t_0)<\dots<x_{N}(t_0)$.  
So far, this analysis is not revealing anything unexpected.  However, we need to remember that 
the Lax pair for the NV1 consists of $3\times 3 $ matrices while the Lax pair for NV2 involves  $4\times 4$ matrices.  
To understand the nature of the reduction on the level of the Lax pairs, let us recall 
that when $\m^*=\m$ the Lax pair for NV2 can be explicitly written (see \eqref{eq:xLax} and \eqref{eq:tLax} )
\begin{equation*} 
\Phi_x=U\Phi, \qquad \Phi_t=V\Phi, 
\end{equation*} 
where 
\begin{equation} \label{eq:NV2U}
U=\begin{pmatrix} 
0&zm&zm&1\\0&0&0&zm\\
0&0&0&zm\\
1&0&0&0 
\end{pmatrix}
\end{equation} 
and 
\begin{equation} \label{eq:NV2V}
V=\frac12 \begin{pmatrix} 
-2u_x u&\frac{u_x}{z}-2zu^2 m&\frac{u_x}{z}-2zu^2 m&2u_x^2\\\frac uz&-\frac{1}{z^2} &0&-\frac{u_x}{z}-2zu^2 m\\
\frac uz&0&-\frac{1}{z^2} &-\frac{u_x}{z}-2zu^2 m\\
-2u^2&\frac u z&\frac uz &2u_x u
\end{pmatrix}. 
\end{equation} 
Now we will change the basis and rescale the spectral variable $z$ in a particular way to establish a direct connection with the Lax pair for NV1 (see \eqref{eq:Novikov-lax-x}, \eqref{eq:Novikov-lax-t}).  
\begin{proposition}\label{prop:NV2-to-NV1} 
Consider the basis of $\C^4$ given by $$ \big \{\hat e_1=e_1,\,  \hat e_2=\frac{1}{\sqrt 2}(e_{21}+e_{22}), \, \hat e_3=e_3, \, 
\hat e_4=e_{21}-e_{22}\big\}$$ and rescale $z\rightarrow \frac{z}{\sqrt 2}$.  Then, with respect to the new basis and rescaled spectral variable $z$, 
$U$ and $V$ read: 
\begin{equation} 
U=\begin{pmatrix} 0&zm&1&0\\0&0&zm&0\\1&0&0&0\\
0&0&0&0 \end{pmatrix}, 
\end{equation} 
and 
\begin{equation} 
V=\begin{pmatrix} -u_x u&\frac{u_x}{z}-zu^2m&u_x^2&0\\
\frac{u}{z}&-\frac{1}{z^2}&-\frac{u_x}{z}-zu^2 m&0\\
-u^2&\frac{u}{z}&u_xu&0\\
0&0&0&-\frac{1}{z^2} \end{pmatrix}.  
\end{equation} 

\end{proposition}

\end{appendices}

%\bibliographystyle{abbrv}
%\bibliography{/Users/cxk/Downloads/NV2-2023/NV2.bib}
%\bibliography{/Users/cxk/Downloads/cxk/chang1608/chang_document/research_1608/papers/me/paper_preparing/2020-two-component-Novikov/NV2/NV2-2023/CMP/NV2.bib}

\end{document}